\documentclass[12pt]{article}

\usepackage[T1]{fontenc}
\ifdefined\pdfminorversion\pdfminorversion=7\fi
\usepackage{lmodern}
\usepackage{amsmath,amssymb,mathtools,amsfonts}
\IfFileExists{stmaryrd.sty}{%
  \usepackage{stmaryrd}%
}{%
  \providecommand{\llbracket}{\mathopen{[\mkern-3mu[}}%
  \providecommand{\rrbracket}{\mathclose{]\mkern-3mu]}}%
}
\usepackage{tikz}
\usepackage{booktabs,array,multirow,tabularx,longtable}
\usepackage[a4paper,margin=25mm]{geometry}
\usepackage[numbers,sort&compress]{natbib}
\newcolumntype{Y}{>{\raggedright\arraybackslash}X}
\usepackage{graphicx}
\usepackage{microtype}
\usepackage{xurl}
\usepackage{xcolor}
\usepackage{enumitem}
\usepackage{pifont}
\usepackage{verbatim}
\usepackage{listings}
\usepackage{float}
\usepackage[section]{placeins}
\usepackage{amsthm}
\usepackage[hidelinks]{hyperref}
\usepackage[capitalize]{cleveref}

\setlist[itemize]{leftmargin=*,itemsep=2pt,topsep=3pt}
\setlist[enumerate]{leftmargin=*,itemsep=2pt,topsep=3pt}
\theoremstyle{plain}
\newtheorem{theorem}{Theorem}[section]
\newtheorem{proposition}[theorem]{Proposition}
\newtheorem{lemma}[theorem]{Lemma}
\newtheorem{corollary}[theorem]{Corollary}
\theoremstyle{definition}
\newtheorem{definition}[theorem]{Definition}
\newtheorem{assumption}[theorem]{Assumption}
\theoremstyle{remark}
\newtheorem{remark}[theorem]{Remark}

\newcommand{\Bool}{\mathbb{B}}
\newcommand{\Real}{\mathbb{R}}

\newcommand{\one}{\mathbf{1}}
\newcommand{\Law}{\operatorname{Law}}

\newcommand{\Risk}{\operatorname{Risk}}
\newcommand{\Delay}{\operatorname{Delay}}

\newcommand{\Nor}{\operatorname{NOR}}
\newcommand{\Bad}{\operatorname{Bad}}
\newcommand{\Good}{\operatorname{Good}}

\newcommand{\supp}{\operatorname{supp}}

\newcommand{\eps}{\varepsilon}
\newcommand{\thetaSet}{\Theta}
\newcommand{\tech}{\mathcal{T}}
\newcommand{\context}{\mathcal{H}}
\newcommand{\witness}{\mathcal{W}}
\newcommand{\cert}{\Pi}

\newcommand{\Uadm}{\mathcal{U}_{\mathrm{adm}}}

\newcommand{\tube}{\mathcal{Q}}

\newcommand{\gen}{\mathcal{L}}
\newcommand{\checker}{\mathsf{PCSCheck}}
\newcommand{\Env}{\mathcal{G}_{\mathrm{env}}}
\newcommand{\epsenv}{\varepsilon_{\mathrm{env}}}
\newcommand{\compiler}{\mathsf{PCSCompile}}

\newcommand{\starv}{\mathord{\ast}}
\newcommand{\low}{\mathsf{L}}
\newcommand{\high}{\mathsf{H}}

\usetikzlibrary{arrows.meta,backgrounds,calc,fit,positioning,shapes.geometric}

\definecolor{navy}{HTML}{264653}
\definecolor{teal}{HTML}{2A9D8F}
\definecolor{gold}{HTML}{E9C46A}
\definecolor{orange}{HTML}{F4A261}
\definecolor{red}{HTML}{E76F51}
\definecolor{blue}{HTML}{457B9D}
\definecolor{paleBlue}{HTML}{EAF2F8}
\definecolor{paleTeal}{HTML}{E7F4F1}
\definecolor{paleGold}{HTML}{FBF4DE}
\definecolor{paleRed}{HTML}{FBECE8}
\definecolor{ink}{HTML}{1D2730}
\definecolor{midgray}{HTML}{6B737B}
\definecolor{lightgray}{HTML}{D9DEE2}
\definecolor{verylight}{HTML}{F7F8F9}

\tikzset{
  font=\sffamily\fontsize{7.45}{8.65}\selectfont,
  panel/.style={draw=lightgray, line width=0.55pt, rounded corners=2.2pt, fill=white},
  title/.style={font=\sffamily\bfseries\fontsize{8.6}{9.2}\selectfont, text=navy},
  tag/.style={font=\sffamily\bfseries\fontsize{9.4}{10}\selectfont, text=white, fill=navy, rounded corners=1.2pt, inner xsep=3pt, inner ysep=1.2pt},
  box/.style={draw=navy, line width=0.55pt, rounded corners=1.4pt, fill=white, align=center, inner xsep=4pt, inner ysep=3pt},
  softbox/.style={draw=lightgray, line width=0.5pt, rounded corners=1.4pt, fill=verylight, align=center, inner xsep=3pt, inner ysep=2.5pt},
  active/.style={draw=teal, line width=1.1pt, fill=paleTeal},
  masked/.style={draw=midgray, line width=0.5pt, dashed, fill=verylight, text=midgray},
  arrow/.style={-{Latex[length=2.2mm,width=1.4mm]}, line width=0.65pt, draw=navy},
  activearrow/.style={-{Latex[length=2.2mm,width=1.4mm]}, line width=1.05pt, draw=teal},
  faintarrow/.style={-{Latex[length=1.8mm,width=1.1mm]}, line width=0.5pt, draw=midgray, dashed},
  gate/.style={draw=navy, line width=0.65pt, circle, minimum size=7.7mm, fill=paleBlue, align=center, inner sep=0pt},
  activegate/.style={draw=teal, line width=1.1pt, circle, minimum size=7.7mm, fill=paleTeal, align=center, inner sep=0pt},
  input/.style={draw=navy, line width=0.55pt, rounded corners=1pt, minimum width=8.7mm, minimum height=5.2mm, fill=white, align=center, inner sep=1pt},
  activeinput/.style={draw=teal, line width=1.05pt, rounded corners=1pt, minimum width=8.7mm, minimum height=5.2mm, fill=paleTeal, align=center, inner sep=1pt},
  smalltext/.style={font=\sffamily\fontsize{6.65}{7.4}\selectfont, text=ink},
  micro/.style={font=\sffamily\fontsize{6.05}{6.65}\selectfont, text=ink},
  mathline/.style={font=\fontsize{7.0}{8.0}\selectfont, text=ink},
  checker/.style={draw=navy, line width=0.8pt, rounded corners=2pt, fill=paleBlue, minimum width=23mm, minimum height=13mm, align=center},
  accepted/.style={draw=teal, line width=0.9pt, rounded corners=1.6pt, fill=paleTeal, minimum width=29mm, minimum height=11mm, align=center},
}

\begin{document}
\title{A proof-carrying architecture for synthetic genetic logic circuits under stochastic temporal contracts}
\author{Arman Ferdowsi\textsuperscript{1,*}, Laura Kovács\textsuperscript{2}\\
\small \textsuperscript{1}Faculty of Computer Science, University of Vienna\\
\small \textsuperscript{2}Faculty of Informatics, TU Wien\\
\small \textsuperscript{*}Corresponding author\\
\small arman.ferdowsi@univie.ac.at, laura.kovacs@tuwien.ac.at}
\date{}
\maketitle
\begin{abstract}
Genetic design automation maps Boolean specifications to regulatory networks
and DNA sequences, but a successful mapping does not establish the probability
that every output remains correct over time. This paper presents
Proof-Carrying Synthetic Biology (PCS-Bio), a formal architecture for acyclic
genetic logic circuits operated under fixed inputs. An ideal checker connects
Boolean equivalence, sequence and model provenance, local stochastic
contracts, and any surrogate-to-target discrepancy bounds. Under explicit
semantic and checker assumptions, these obligations imply a joint temporal
refinement guarantee without assuming independence among molecular events.
Assignment-specific sufficient cubes select the branches needed to prove each
output value. This can reduce risk and delay accounting provided that the retained local
contracts remain uniform over the physically present but logically masked
inputs. The reported evaluation exercises three separate interfaces: an
archived Cello execution with sequence reconstruction, a restricted exact
rational checker for a synthetic finite-state instance, and diagnostic IEEE~754
binary64 replay. Across 60 deterministic Boolean formula graphs, slicing
reduced the worst-assignment risk-charge sum in 55 cases and the maximum
accounting delay in 54. These results concern proof accounting and selected
implementation components. They do not establish an integrated sequence-bound
temporal certificate or physical validation.
\end{abstract}

\medskip
\section{Introduction}
\label{sec:introduction}
Synthetic genetic circuits implement decision rules inside cells by controlling
which genes are expressed. For example, let $A$, $B$, and $C$ indicate whether
three sensed chemicals are present under specified detection conditions. The
rule $F=(A\lor\neg B)\land C$ requests production of a reporter protein when
$C$ is present and either $A$ is present or $B$ is absent. A reporter provides
an observable output, such as fluorescence. Writing $\Bool=\{0,1\}$, we use
$\llbracket F\rrbracket(u)$ for the Boolean value of $F$ under input assignment
$u\in\Bool^3$. A circuit should represent that value by an agreed low or high
output regime, rather than by an exact protein count.

The physical implementation uses regulatory DNA parts in a host cell, often
called the \emph{chassis}. A \emph{promoter} is a DNA region that controls
transcription, the production of RNA from DNA. A \emph{repressor} reduces
transcription from a compatible promoter. Sensors, repressors, and output
genes can therefore implement interconnected logic gates. Cello and
Cello~2.0 established automated Boolean-to-DNA design using characterized gate
libraries~\cite{Nielsen2016,Jones2022}. Their User Constraint Files (UCFs)
record gate functions, parts, response models, and implementation
constraints~\cite{CelloUCF2021}. Such records provide a starting point for
checking what was assembled. They do not by themselves establish its
probability of correct operation over time.

\paragraph{Challenges and limitations.}
A genetic logic circuit produces a fluctuating molecular signal. Random
reaction events and promoter switching can cause variability even under fixed
inputs~\cite{Golding2005,Raj2008}. The intended output must first reach its
required regime and then remain there during measurement. Shared cellular
resources can couple different gates, so multiplying individual success
probabilities is generally unjustified. Further, a proof about a convenient
surrogate model does not automatically establish a property of the designated
target model. Boolean correctness, consistency of the assembled sequence and
model, stochastic composition, and quantified approximation error must be
connected explicitly.

This paper introduces \emph{Proof-Carrying Synthetic Biology} (PCS-Bio), a
formal architecture for acyclic genetic logic circuits operated under fixed
Boolean inputs. The central PCS-Bio question is what independently checkable
evidence connects a Boolean specification to a joint stochastic path guarantee. An untrusted producer supplies a design and its supporting
proof objects. An ideal checker reconstructs the claimed relationships and
checks the obligations against a verifier-controlled request and technology
library. The guarantee is relative to a declared stochastic model and its
assumptions. Whether that model adequately describes the assembled biological
system remains an empirical question.

For a precise target, let $\varnothing\ne\Uadm\subseteq\Bool^n$ contain the
admitted input assignments and let $\mathbf F=(F_1,\ldots,F_m)$ specify the
named outputs. For design $D$, write $Y_{D,j}$ for the observed output process
and $Q_{y_j}^b$ for the expression regime representing Boolean value $b$. The
parameter set $\thetaSet$ and family of joint path laws
$\mathfrak P_D(u,\theta)$ encode the declared uncertainty, initial conditions,
and allowed cellular contexts. Given a required probability $p$ and a
measurement window $[\tau,T]$, the target is
\begin{equation}
 \inf_{u\in\Uadm}
 \inf_{\theta\in\thetaSet}
 \inf_{\mathbb P\in\mathfrak P_D(u,\theta)}
 \mathbb P\!\left[
   \bigcap_{j=1}^{m}
   \left\{
   Y_{D,j}(t)\in Q_{y_j}^{\llbracket F_j\rrbracket(u)}
   \text{ for every }t\in[\tau,T]
   \right\}
 \right]
 \ge p.
 \label{eq:goal-intro}
\end{equation}
Thus every output must be correct throughout the same measurement window on
the same sample path. The deadline $\tau$ allows the circuit to settle, and
$T-\tau$ is the required persistence duration. The bound applies uniformly
over the admitted inputs, parameters, and model laws. Separate per-output guarantees at level $p$ do not in general establish
the same joint guarantee.

Boolean structure helps control the cost of proving this property. In the
running formula, $C=0$ forces the output to zero regardless of $A$ and $B$.
A \emph{sufficient cube} records only the input values needed for such a
Boolean conclusion. PCS-Bio uses these partial assignments to select smaller
stochastic proof witnesses. The omitted branches remain physically present,
and their signals must still satisfy the envelopes assumed by the retained
local contracts. Slicing can reduce a union-bound \emph{risk charge} or a
composed \emph{accounting delay}. It does not alter the circuit or demonstrate
an improvement in its physical reliability or response time.

\paragraph{Our contributions.}
The contribution is an explicit proof architecture connecting established
Boolean, stochastic, and compilation techniques. First, Section~\ref{sec:compiler} gives a constructive
not-OR (NOR) translation and a reconstructed Boolean comparison that connect source
formulas to mapped logic. A closed-world binding judgment then relates that
logic, exact sequence bytes, design structure, and the declared target model.
Second, Section~\ref{sec:slicing} gives sufficient-cube contracts that specify the uniformity needed to omit
logically masked branches, and a joint composition theorem counts shared
obligations once without assuming independence. Third, Sections~\ref{sec:local-certificates} and~\ref{sec:transfer} give local
reach-and-stay certificates and synchronous surrogate-transfer conditions,
providing the
stochastic premises required by that composition. Finally, Section~\ref{sec:checker} specifies the checker obligations under which acceptance entails
Equation~\eqref{eq:goal-intro}. Figure~\ref{fig:pipeline} presents these
relationships and their trust boundary.

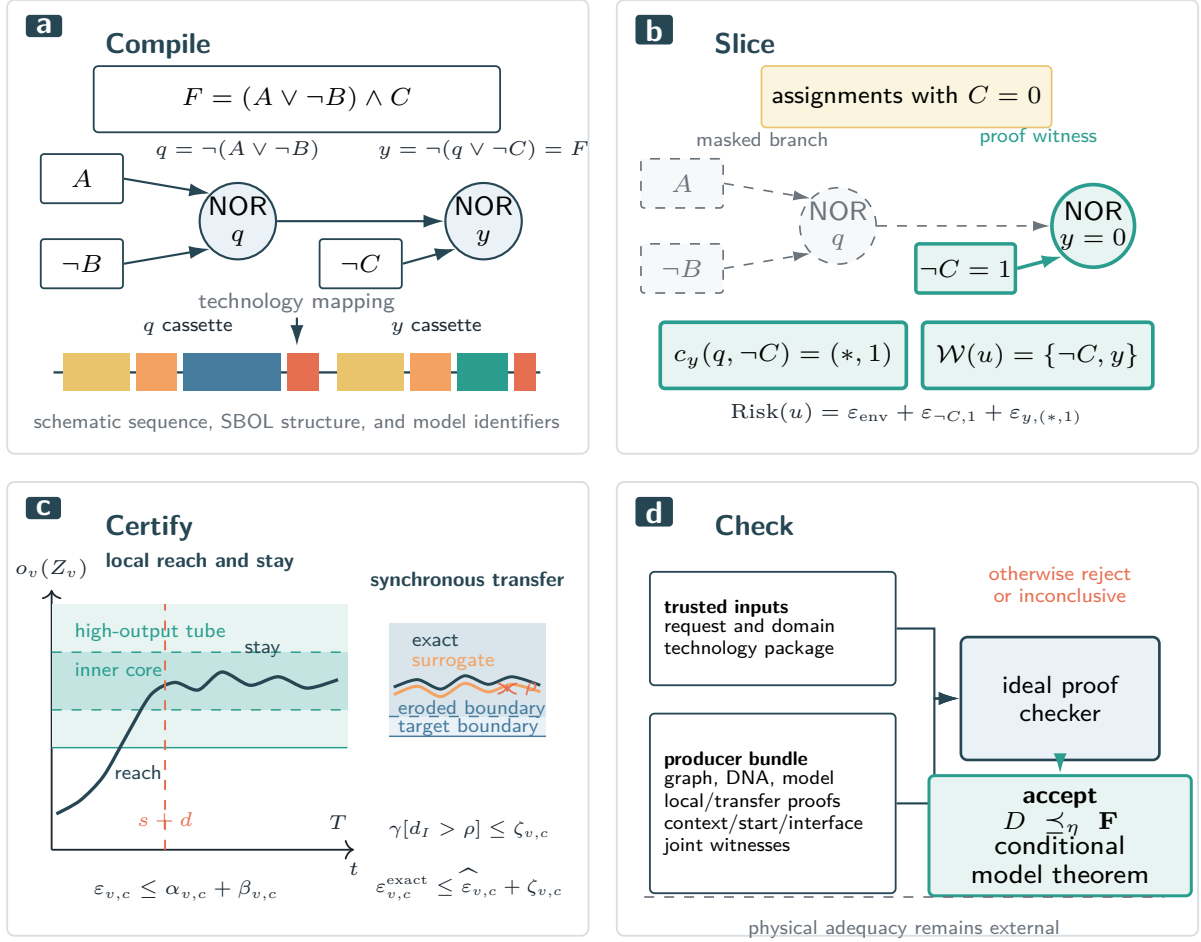
\begin{figure}[t]
\centering
\resizebox{\linewidth}{!}{%
\begin{tikzpicture}[x=1mm,y=1mm]
  \def\PW{61.5}
  \def\PH{48.0}
  \def\GX{3.0}
  \def\GY{3.0}

  \node[panel, minimum width=\PW mm, minimum height=\PH mm, anchor=north west] (pa) at (0,\PH+\GY) {};
  \node[panel, minimum width=\PW mm, minimum height=\PH mm, anchor=north west] (pb) at (\PW+\GX,\PH+\GY) {};
  \node[panel, minimum width=\PW mm, minimum height=\PH mm, anchor=north west] (pc) at (0,0) {};
  \node[panel, minimum width=\PW mm, minimum height=\PH mm, anchor=north west] (pd) at (\PW+\GX,0) {};

  \node[tag, anchor=north west] at (2.0,\PH+\GY-1.6) {a};
  \node[title, anchor=north west] at (9.0,\PH+\GY-2.2) {Compile};

  \node[box, minimum width=43mm, minimum height=7mm, anchor=north] (formula) at (30.75,\PH+\GY-7.0) {$F=(A\lor\neg B)\land C$};

  \node[input] (aA) at (8.0,\PH+\GY-19.0) {$A$};
  \node[input] (aNB) at (8.0,\PH+\GY-28.0) {$\neg B$};
  \node[gate] (aq) at (24.5,\PH+\GY-23.5) {\scriptsize NOR\\[-1pt]$q$};
  \node[input] (aNC) at (37.5,\PH+\GY-28.0) {$\neg C$};
  \node[gate] (ay) at (50.5,\PH+\GY-23.5) {\scriptsize NOR\\[-1pt]$y$};
  \draw[arrow] (aA.east) -- (aq.north west);
  \draw[arrow] (aNB.east) -- (aq.south west);
  \draw[arrow] (aq.east) -- (ay.west);
  \draw[arrow] (aNC.east) -- (ay.south west);
  \node[micro, anchor=south] at (24.5,\PH+\GY-18.5) {$q=\neg(A\lor\neg B)$};
  \node[micro, anchor=south] at (50.5,\PH+\GY-18.5) {$y=\neg(q\lor\neg C)=F$};

  \draw[arrow] (30.75,\PH+\GY-33.7) -- (30.75,\PH+\GY-36.7);
  \node[smalltext, text=midgray, anchor=south] at (30.75,\PH+\GY-34.6) {technology mapping};

  \begin{scope}[shift={(5.0,\PH+\GY-39.5)}]
    \draw[draw=navy, line width=0.7pt] (0,0) -- (51.5,0);
    \fill[gold] (1.0,-2.0) rectangle (8.0,2.0);
    \fill[orange] (8.7,-2.0) rectangle (13.0,2.0);
    \fill[blue] (13.7,-2.0) rectangle (24.0,2.0);
    \fill[red] (24.7,-2.0) rectangle (28.0,2.0);
    \fill[gold] (30.0,-2.0) rectangle (37.0,2.0);
    \fill[orange] (37.7,-2.0) rectangle (42.0,2.0);
    \fill[teal] (42.7,-2.0) rectangle (48.0,2.0);
    \fill[red] (48.7,-2.0) rectangle (51.0,2.0);
    \node[micro, anchor=south] at (14.2,2.5) {$q$ cassette};
    \node[micro, anchor=south] at (40.5,2.5) {$y$ cassette};
    \node[micro, text=midgray, anchor=north] at (25.75,-3.0) {schematic sequence, SBOL structure, and model identifiers};
  \end{scope}

  \begin{scope}[shift={(\PW+\GX,0)}]
    \node[tag, anchor=north west] at (2.0,\PH+\GY-1.6) {b};
    \node[title, anchor=north west] at (9.0,\PH+\GY-2.2) {Slice};
    \node[softbox, fill=paleGold, draw=gold, minimum width=29mm, minimum height=6.5mm, anchor=north] at (30.75,\PH+\GY-7.0) {assignments with $C=0$};

    \node[input, masked] (bA) at (7.0,\PH+\GY-19.5) {$A$};
    \node[input, masked] (bNB) at (7.0,\PH+\GY-28.5) {$\neg B$};
    \node[gate, masked] (bq) at (23.5,\PH+\GY-24.0) {\scriptsize NOR\\[-1pt]$q$};
    \node[activeinput] (bNC) at (37.0,\PH+\GY-28.5) {$\neg C=1$};
    \node[activegate] (by) at (50.5,\PH+\GY-24.0) {\scriptsize NOR\\[-1pt]$y=0$};
    \draw[faintarrow] (bA.east) -- (bq.north west);
    \draw[faintarrow] (bNB.east) -- (bq.south west);
    \draw[faintarrow] (bq.east) -- (by.west);
    \draw[activearrow] (bNC.east) -- (by.south west);
    \node[micro, text=midgray, anchor=south] at (15.5,\PH+\GY-17.0) {masked branch};
    \node[micro, text=teal, anchor=south] at (44.7,\PH+\GY-17.0) {proof witness};

    \node[box, active, minimum width=24mm, minimum height=7mm, anchor=north west] at (4.5,\PH+\GY-34.0) {$c_y(q,\neg C)=(\starv,1)$};
    \node[box, active, minimum width=24mm, minimum height=7mm, anchor=north east] at (57.0,\PH+\GY-34.0) {$\mathcal W(u)=\{\neg C,y\}$};
    \node[font=\fontsize{6.15}{6.8}\selectfont, text=ink, anchor=south, align=center] at (30.75,\PH+\GY-46.5) {$\Risk(u)=\epsenv+\eps_{\neg C,1}+\eps_{y,(\starv,1)}$};
  \end{scope}

  \node[tag, anchor=north west] at (2.0,-1.6) {c};
  \node[title, anchor=north west] at (9.0,-2.2) {Certify};

  \node[micro, font=\sffamily\bfseries\fontsize{6.2}{6.7}\selectfont, text=navy, anchor=south] at (20.5,-11.0) {local reach and stay};
  \begin{scope}[shift={(4.8,-39.0)}]
    \draw[->, draw=ink, line width=0.55pt] (0,0) -- (32.0,0) node[micro, below=-0.1mm] {$t$};
    \draw[->, draw=ink, line width=0.55pt] (0,0) -- (0,27.5) node[micro, above=-0.3mm] {$o_v(Z_v)$};
    \fill[paleTeal] (0,10.8) rectangle (31.3,26.0);
    \fill[teal!28] (0,14.8) rectangle (31.3,20.9);
    \draw[draw=teal, line width=0.45pt] (0,10.8) -- (31.3,10.8);
    \draw[draw=teal, line width=0.45pt, dashed] (0,14.8) -- (31.3,14.8);
    \draw[draw=teal, line width=0.45pt, dashed] (0,20.9) -- (31.3,20.9);
    \node[micro, text=teal, anchor=west] at (1.0,23.0) {high-output tube};
    \node[micro, text=teal, anchor=west] at (1.0,19.1) {inner core};
    \draw[draw=navy, line width=1.0pt] plot[smooth] coordinates {(0.5,3.8) (3,5.2) (5.5,7.8) (8,12.2) (10.5,16.4) (13,17.7) (15.5,16.9) (18,18.8) (21,17.4) (24,18.3) (27,17.1) (30.3,18.0)};
    \draw[draw=red, line width=0.6pt, dashed] (12.0,0) -- (12.0,26.0);
    \node[micro, text=red, anchor=south] at (12.0,0.8) {$s+d$};
    \node[micro, text=ink, anchor=south] at (30.3,0.8) {$T$};
    \node[micro, text=navy, anchor=south west] at (5.2,6.0) {reach};
    \node[micro, text=navy, anchor=south west] at (19.0,19.0) {stay};
  \end{scope}

  \begin{scope}[shift={(40.5,-34.0)}]
    \node[micro, font=\sffamily\bfseries\fontsize{6.2}{6.7}\selectfont, text=navy, anchor=south] at (8.3,21.0) {synchronous transfer};
    \fill[paleBlue] (0,7.0) rectangle (16.6,19.0);
    \fill[blue!20] (0,9.1) rectangle (16.6,19.0);
    \draw[draw=blue, line width=0.45pt] (0,7.0) -- (16.6,7.0);
    \draw[draw=blue, line width=0.45pt, dashed] (0,9.1) -- (16.6,9.1);
    \draw[draw=navy, line width=0.9pt] plot[smooth] coordinates {(0.5,12.2) (3,13.0) (5.5,12.0) (8,13.4) (10.5,12.5) (13,13.3) (16,12.4)};
    \draw[draw=orange, line width=0.9pt] plot[smooth] coordinates {(0.5,11.4) (3,12.3) (5.5,11.2) (8,12.6) (10.5,11.7) (13,12.5) (16,11.7)};
    \draw[<->, draw=red, line width=0.55pt] (12.5,12.4) -- (12.5,11.7);
    \node[micro, text=red, anchor=west] at (12.9,12.05) {$\rho$};
    \node[micro, text=navy, anchor=west, yshift=0.6mm] at (1.0,16.6) {exact};
    \node[micro, text=orange, anchor=west] at (1.0,14.8) {surrogate};
    \node[micro, text=blue, anchor=west, xshift=-1.5mm] at (1.0,8.0) {\tiny target boundary};
    \node[micro, text=blue, anchor=west, xshift=-1.5mm] at (1.0,9.9) {\tiny eroded boundary};
  \end{scope}

  \node[font=\fontsize{5.9}{6.4}\selectfont, text=ink, anchor=south] at (19.0,-45.8) {$\varepsilon_{v,c}\leq\alpha_{v,c}+\beta_{v,c}$};
  \node[font=\fontsize{5.65}{6.2}\selectfont, text=ink, anchor=south, align=center] at (49,-39.4) {$\gamma[d_I>\rho]\leq\zeta_{v,c}$};
  \node[font=\fontsize{5.65}{6.2}\selectfont, text=ink, anchor=south, align=center] at (49,-45.8) {$\varepsilon_{v,c}^{\mathrm{exact}}\leq\widehat\varepsilon_{v,c}+\zeta_{v,c}$};

  \begin{scope}[shift={(\PW+\GX,0)}]
    \node[tag, anchor=north west] at (2.0,-1.6) {d};
    \node[title, anchor=north west] at (9.0,-2.2) {Check};

    \node[box, text width=23mm, minimum width=26mm, minimum height=12mm, anchor=north west, align=left, inner xsep=4pt, font=\sffamily\fontsize{5.9}{6.7}\selectfont] (trusted) at (3.5,-9.5) {%
      \textbf{trusted inputs}\\[-1pt]
      request and domain\\
      technology package};

    \node[box, text width=23mm, minimum width=26mm, minimum height=19mm, anchor=north west, align=left, inner xsep=4pt, font=\sffamily\fontsize{5.9}{6.7}\selectfont] (bundle) at (3.5,-24.5) {%
      \textbf{producer bundle}\\[-1pt]
      graph, DNA, model\\
      local/transfer proofs\\
      context/start/interface\\
      joint witnesses};

    \node[checker, minimum width=21mm] (check) at (47.0,-23.0) {ideal proof\\checker};
    \draw[arrow] (trusted.east) -- ++(4.0,0) |- (check.west);
    \draw[arrow] (bundle.east) -- ++(4.0,0) |- (check.west);

    \node[micro, text=red, text width=24mm, align=center, anchor=south] at (47.0,-14.2) {otherwise reject or inconclusive};
    \node[accepted, text width=25mm, align=center, minimum width=27mm, minimum height=12mm, anchor=north] (accept) at (47.0,-31.0) {\textbf{accept}\\[-1pt]$D\preceq_{\eta}\mathbf F$\\[-1pt]{\scriptsize conditional model theorem}};
    \draw[activearrow] (check.south) -- (accept.north);
    \draw[draw=midgray, line width=0.5pt, dashed] (3.0,-44.0) -- (58.0,-44.0);
    \node[micro, text=midgray, anchor=north] at (30.5,-45.0) {physical adequacy remains external};
  \end{scope}
\end{tikzpicture}
}
\caption{\scriptsize PCS-Bio proof-carrying workflow. \textbf{a}, A Boolean specification
is lowered to dual-rail not-OR (NOR) logic and bound, through a
verifier-controlled technology package, to a schematic sequence-resolved
Synthetic Biology Open Language (SBOL) design and a declared target model.
\textbf{b}, For $C=0$, the cube
$c_y(q,\neg C)=(\starv,1)$ fixes only the second input. $\starv$ denotes an
unfixed input. The $A$, $\neg B$, and $q$ branch is omitted from the Boolean
proof witness, but its physical signals remain constrained by the local
envelope. \textbf{c}, Reach-and-stay evidence and, when needed, synchronous
surrogate-transfer evidence provide target-model local risk charges.
\textbf{d}, The ideal checker reconstructs every semantic link and the joint
witness. Acceptance proves conditional model-level refinement
$D\preceq_\eta\mathbf F$. Empirical adequacy of the model for the assembled
biology remains an external assumption.}
\label{fig:pipeline}
\end{figure}

The reported implementation evidence exercises three separate interfaces: an
archived Cello execution with mapping and sequence reconstruction, a restricted
exact rational checker for a synthetic finite-state instance, and diagnostic
64-bit floating-point replay. Across a fixed corpus of 60 formula graphs,
slicing reduced the reported worst-assignment risk-charge sum in 55 cases and
maximum accounting delay in 54. These are accounting results for synthetic
models. No reported execution provides an integrated certificate binding the
Cello-derived sequence to a checked stochastic temporal guarantee.

Section~\ref{sec:related-work} positions the architecture.
Sections~\ref{sec:problem}--\ref{sec:signals} define the request, compilation,
model, and signal contracts. Section~\ref{sec:slicing} develops proof slicing
and composition, and Sections~\ref{sec:local-certificates}
and~\ref{sec:transfer} establish the local and transfer obligations.
Section~\ref{sec:checker} states conditional checker soundness.
Section~\ref{sec:example} works through one proof witness.
Sections~\ref{sec:computational-methods} and~\ref{sec:computational-evaluation}
present the methods and reported evidence, followed by the remaining
limitations and conclusion in Sections~\ref{sec:limitations}
and~\ref{sec:conclusion}.

\section{Related Work}
\label{sec:related-work}
The relevant literature supplies three parts of the architecture: biological
compilation, compositional specifications, and stochastic proof techniques.
We distinguish those established components from the integrated obligations
proposed here.

Genetic design automation already spans high-level specification, part
selection, technology mapping, and sequence construction. Early programming
and design systems include GEC, GenoCAD, and Eugene~\cite{Pedersen2009GEC,Czar2009GenoCAD,Bilitchenko2011Eugene}.
End-to-end workflows, automated part selection, and sequence-level design
extended these capabilities~\cite{Beal2012EndToEnd,Yaman2012Automated,Rodrigo2013AutoBioCAD}.
Mapping based on directed acyclic graphs and genetic logic synthesis connect Boolean structure to
available devices~\cite{Roehner2014DAG,Vaidyanathan2015Framework}, while SBOL,
model-generation tools, and BioCRNpyler connect design descriptions to
mechanistic models~\cite{McLaughlin2020,Roehner2015SBOLSBML,Poole2022}.
PCS-Bio builds on this separation of source logic, biological implementation,
and model construction.

Robust and context-aware design addresses structural alternatives, parameter
uncertainty, and interactions between genetic parts~\cite{Schladt2021,Engelmann2023}.
Noise-aware and risk-averse optimization incorporates stochastic behavior
into design objectives~\cite{Sequeiros2023,Kobiela2026}.
These approaches motivate quantitative guarantees, but optimizing a design
criterion and checking the complete temporal claim in
Equation~\eqref{eq:goal-intro} are different tasks. The present work focuses on
the latter's proof obligations rather than on a new biological optimization
objective.

Formal methods have treated uncertain gene networks, modular temporal
specifications, stochastic model checking, and stochastic
hazards~\cite{Yordanov2011,Bartocci2013,Madsen2014SMC,Buecherl2021}.
iBioSim and full-stack synthetic-biology environments integrate modeling,
analysis, and compilation~\cite{Watanabe2019iBioSim,Konur2021,Konur2023}.
The preprint of Pandey et al.~\cite{Pandey2022} formulates assume-guarantee
contracts for synthetic biology using ordinary differential equations.
Incer et al.~\cite{Incer2024} establish steady-state input--output guarantees
under resource-sharing context effects. These are relevant compositional
foundations, but their contract semantics differ from the stochastic path
events used here. Pacti supplies algebraic
operations for compositional contract analysis~\cite{Incer2025}. Our contracts
specialize the relevant uniformity and joint-event requirements for
assignment-specific proof witnesses and sequence-bound target models.

The individual proof ingredients also have established origins. Sufficient
cubes are implicants of a Boolean function or its complement~\cite{Hachtel1996}.
Barrier and supermartingale methods support stochastic verification and
composition~\cite{Lavaei2022,Anand2022,Abate2024,Abate2025,Nejati2022,Anand2024}.
Their model classes and temporal properties differ. In particular, quantitative
supermartingale certificates in~\cite{Abate2025} address discrete-time Markov
chains, so they are not a substitute for the continuous-time stopping arguments
given below.
Strong approximations and error-controlled hybrid models relate reaction
processes to tractable surrogates~\cite{Kurtz1978,Ganguly2015}. Proof-carrying
code and clausal proof logging separate proof production from
checking~\cite{Necula1997,Wetzler2014}. PCS-Bio specifies how these ingredients
must fit together to connect Boolean source meaning, exact sequence identity,
masked-input stochastic contracts, and joint path guarantees. Its proposed
contribution is this conjunction of checked relationships, with the present
implementation covering only the restricted fragments described above.

\section{Problem Statement}
\label{sec:problem}

The verification problem has three inputs: a Boolean specification, a technology
library, and a temporal contract. We define their representations and the
stochastic laws over which the guarantee must hold. The final definition
separates correctness of the declared model from its adequacy for a physical
circuit. Subsequent sections construct evidence for that model-level claim.

\subsection{Source specifications and compilation request}
We use standard Boolean-function, circuit, and decision-diagram terminology
\cite{Hachtel1996}. Let $X=\{x_1,\ldots,x_n\}$ be the Boolean sensor
variables. The core source grammar is
\begin{equation}
 \varphi ::= 0\mid 1\mid x_i\mid\neg\varphi
 \mid(\varphi\land\varphi)\mid(\varphi\lor\varphi).
 \label{eq:source-grammar}
\end{equation}
A source specification is a vector
$\mathbf F=(F_1,\ldots,F_m)$ of formulas, where $m\ge1$. The formulas are
stored in a shared directed acyclic graph, so a subformula used in several
places is represented once. A directed acyclic graph (DAG) is a directed graph
with no directed cycle. It therefore admits an order in which every gate is
evaluated after its parents. A derived connective, such as implication or
exclusive OR, is accepted only after the checker verifies its translation into
the core grammar in~\eqref{eq:source-grammar}. For an assignment $u\in\Bool^n$, the notation
$\llbracket\varphi\rrbracket(u)$ denotes the Boolean value obtained by
evaluating formula $\varphi$ at $u$. Evaluation of a vector is componentwise:
\[
 \llbracket\mathbf F\rrbracket(u)
 =\bigl(\llbracket F_1\rrbracket(u),\ldots,
         \llbracket F_m\rrbracket(u)\bigr).
\]

The trusted input contract specifies a nonempty admissible set
$\varnothing\ne\Uadm\subseteq\Bool^n$. In the ideal proof language, this set
may be represented by a canonical binary decision diagram (BDD), meaning a
graph representation under a declared variable order with the standard reduction rules, or by a Boolean
predicate with a clause-level equivalence proof that the checker replays. The
restricted implementation instead requires an explicit, canonically ordered
list. Assignments outside $\Uadm$ are don't-cares: synthesis may choose their
behavior freely, and no correctness claim is made for them. The checker proves
Boolean equivalence and computes worst-case risk-charge and accounting-delay
values only on $\Uadm$.

A \emph{compilation instance} collects these inputs:
\begin{equation}
 \mathcal I=(\mathbf F,X,\Uadm,\mathbf y,\tech,\thetaSet,\context,
 \mathcal I_0,\eta,B_{\max}).
 \label{eq:instance}
\end{equation}
Here $\mathbf y=(y_1,\ldots,y_m)$ names the output ports, $\tech$ is the
technology library, and $\eta$ is the temporal contract defined below. The
nonempty set $\thetaSet$ contains admissible global parameter vectors. If
technology entries use aliases for one parameter, those aliases refer to the
same coordinate of one vector $\theta\in\thetaSet$.

The \emph{trusted initial-condition schema} $\mathcal I_0$ in
\eqref{eq:instance} is instantiated after the
checker reconstructs a candidate state space. For each $u\in\Uadm$, it induces
a finitely represented nonempty set
\[
\mathcal I_{0,D}(u)\subseteq\mathcal P(\mathsf X_D)
\]
of initial probability laws. Here $\mathcal P(S)$ denotes the set of Borel
probability measures on $S$. A deterministic initial state is represented by
its Dirac, or point-mass, law.

The component $\context$ of \eqref{eq:instance} is a finitely represented
nonempty set of \emph{admissible context policies}. These describe how the
host cell and its environment may influence the circuit, for example through
available transcriptional resources. A policy may use only information available through the
current time. Whenever a context value enters a stochastic coefficient or
intensity at time $t$, the model uses a predictable version, measurable with
respect to information immediately before $t$. Random exogenous context is
represented through declared initial coordinates or driving noise and remains
part of the same global joint path law.

Finally, the optional vector $B_{\max}$ in \eqref{eq:instance} gives componentwise upper bounds on
checked quantities such as resource use, cellular burden, fan-out (the number
of downstream users of a signal), and construct size.

Each named input or output is a typed port recording its molecular carrier,
direction, state space, physical unit, and observation map. For output $y_j$,
$\mathsf Y_{y_j}$ is a Polish state space, meaning that its topology is induced
by some complete separable metric. The continuous map
$o_{y_j}\colon\mathsf Y_{y_j}\to\mathbb R$ converts a molecular state to the
scalar quantity compared with the low and high thresholds in the same unit.
The ideal checker compares this port declaration with the generated model. In
the restricted finite-state implementation, the observation map is an exact
rational table.

The temporal output contract is
\begin{equation}
 \eta=\bigl((\low_j,\high_j)_{j=1}^{m},p,\tau,T\bigr),
 \label{eq:eta}
\end{equation}
with $\low_j<\high_j$ for every output, $0<p\le1$, and $0\le\tau<T<\infty$. The probability $p$ applies to the joint event that every output is correct throughout the same measurement window. A marginal contract for one output is obtained by setting $m=1$. Surrogate-transfer radii are positive obligation-specific quantities stored in the corresponding transfer evidence. They are not global components of $\eta$. Every set or function discharged by checker replay must have a finite representation in an admitted proof language. A well-posedness or
integrability clause without such a representation remains an explicit
external mathematical assumption.
\subsection{Technology entries and the trust boundary}
This subsection identifies which biological information the checker may trust and which claims it must reconstruct. The distinction is essential because sequence provenance can establish that the intended bytes were used, but it cannot establish that a kinetic model is empirically adequate for those bytes.

A \emph{technology entry} $e\in\tech$ is a tuple
\begin{equation}
 e=(\mathsf{id},\mathsf{seq},b_e,\mathsf{ports},\mathsf{grammar},
       \mathcal M_e,\thetaSet_e,\context_e,\mathsf{certs}_e).
 \label{eq:technology-entry}
\end{equation}
In the ideal architecture, the nucleotide word $\mathsf{seq}\in\{A,C,G,T\}^{*}$ is supplied in a versioned, verifier-controlled technology package. Each record retains its raw source bytes and a separately parsed canonical uppercase nucleotide word. The checker verifies the raw-record digest, validates every nucleotide symbol, and constructs designs from canonical words. Raw-record equality and canonical-sequence equality are reported separately. The package is loaded from declared files before compilation, and a cryptographic digest identifies the exact bytes of each referenced object. The compiler neither generates these library words nor retrieves undeclared replacements. The function $b_e$ is the ideal Boolean gate. Ports bind regulatory inputs and molecular outputs using the interface definition above. The grammar specifies how DNA parts may be assembled. A promoter initiates
transcription, an operator is a regulatory binding site, a ribosome-binding
site (RBS) controls translation initiation, a coding sequence (CDS) specifies
the protein, and a terminator ends transcription. Insulation elements reduce
unintended interactions between neighboring parts. The grammar also declares
orientation, topology, and placement constraints. The model template $\mathcal M_e$, uncertainty interface $\thetaSet_e$, and context interface $\context_e$ define the designated nonsurrogate mathematical semantics, called the exact target model below. The word exact here distinguishes that target from a surrogate. It does not assert empirical exactness of biology. The field $\mathsf{certs}_e$ identifies local obligations and admitted proof backends. A referenced item becomes evidence only after checker replay.

The compiler does not infer a trustworthy kinetic model from raw nucleotide sequence alone. The ideal checker establishes artifact provenance and consistency. Exact sequence bytes, model identifiers, certificate identifiers, and typed ports must refer to the same technology record. Trust-boundary artifacts are byte-hashed or canonically parsed and hash-bound, while semantic identifiers are cross-checked within those objects. This is a syntactic and model-construction property rather than an empirical conclusion. No single current executable schema instantiates the complete tuple in~\eqref{eq:technology-entry}. The Cello layer binds real library and sequence evidence without a temporal model, the restricted formal checker receives a separate synthetic finite model without a sequence, and binary64 replay uses a synthetic technology file.

\subsection{Trace abstraction and refinement}
A static Boolean function does not specify when a molecular output becomes
correct or how long it remains correct. We therefore interpret a compiled
model through its sample paths, using standard stochastic-process
constructions~\cite{EthierKurtz1986}, and define correctness over the whole
measurement interval.

For a generated design $D$, let $\mathsf X_D$ be its declared Polish global
state space, including every wire, mode, and context coordinate used by the
model or a certificate. Define the canonical sample-path space by
\begin{equation}
 \Omega_D:=\mathbb D_T(\mathsf X_D)
 :=\left\{
 \omega\colon[0,T]\to\mathsf X_D
 \;\middle|\;
 \omega\text{ is right-continuous and has left limits}
 \right\}.
 \label{eq:canonical-path-space}
\end{equation}
Such paths are called c\`adl\`ag~\cite{EthierKurtz1986}. Equip $\Omega_D$ with
the Skorokhod $J_1$ topology and write
$\mathcal F_D=\mathcal B(\Omega_D)$ for its Borel sigma algebra. The canonical
coordinate process reads the state from a path:
$\mathbf X_t(\omega)=\omega(t)$. Under each law $\mathbb P$, use the complete,
right-continuous augmentation $\{\mathcal F_t^{\mathbb P}\}$ of its canonical
filtration.

Let $\operatorname{Sol}(\mathcal M_D,u,\theta,\nu_0,h)$ be the set of admitted
joint weak-solution laws on $\Omega_D$ for the designated model, input,
parameter vector, initial law, and context policy. A weak-solution law
specifies the probability distribution of a stochastic process satisfying the
model equations. It need not construct that process on a predetermined
probability space. Then
\begin{equation}
 \mathfrak P_D(u,\theta)
 =\bigcup_{\nu_0\in\mathcal I_{0,D}(u)}
  \bigcup_{h\in\context}
 \operatorname{Sol}(\mathcal M_D,u,\theta,\nu_0,h).
 \label{eq:exact-law-family}
\end{equation}
The outer unions in \eqref{eq:exact-law-family} collect every solution law
allowed by any admitted initial distribution $\nu_0$ and context policy $h$,
while holding $D$, $u$, and $\theta$ fixed. Thus $\mathfrak P_D(u,\theta)$
is a set of probability laws, not a single probability. Its elements are
joint laws of the assembled process, so shared stochastic dependence is
retained without taking products of gate-wise laws. The output vector is a measurable map
\[
 \mathbf Y_D\colon\Omega_D\longrightarrow
 \mathbb D_T\!\left(\prod_{j=1}^{m}\mathsf Y_{y_j}\right),
\]
and its pushforward law is
\begin{equation}
 \Law_{\mathbb P}(\mathbf Y_D)
 :=\mathbb P\circ\mathbf Y_D^{-1}.
 \label{eq:output-path-law}
\end{equation}

A Hill response curve describes a steady-state input--output relation. It does
not bound threshold-crossing events over time. Even agreement of stationary
output distributions would leave settling and persistence unresolved. The
following events express the required path property directly.

For output $j$, the low and high \emph{tubes} are the molecular states whose
observations satisfy the corresponding thresholds:
\begin{equation}
 Q_{y_j}^0=\{z\in\mathsf Y_{y_j}\mid o_{y_j}(z)\le\low_j\},
 \qquad
 Q_{y_j}^1=\{z\in\mathsf Y_{y_j}\mid o_{y_j}(z)\ge\high_j\}.
 \label{eq:output-tubes}
\end{equation}
Because $\low_j<\high_j$, a state strictly between the thresholds represents
neither Boolean value. For $b\in\Bool$ and a compact interval
$I\subseteq[0,T]$, define
\begin{align}
 \Good_{y_j}^b(I)
 &=\left\{\omega\in\Omega_D\;\middle|\;
 \begin{aligned}
  Y_{D,j}(t,\omega)&\in Q_{y_j}^b\\[-0.2em]
  &\text{for every }t\in I
 \end{aligned}
 \right\},\nonumber\\
 \Bad_{y_j}^b(I)&=\Omega_D\setminus\Good_{y_j}^b(I).
 \label{eq:good-output}
\end{align}
For $\mathbf b=(b_1,\ldots,b_m)$, let
\begin{equation}
 \Good_{\mathbf y}^{\mathbf b}(I)
 =\bigcap_{j=1}^{m}\Good_{y_j}^{b_j}(I),
 \qquad
 \Bad_{\mathbf y}^{\mathbf b}(I)
 =\Omega_D\setminus\Good_{\mathbf y}^{\mathbf b}(I).
 \label{eq:joint-output-event}
\end{equation}
Continuity of the observation maps makes the tubes closed. For a c\`adl\`ag
path, membership in a closed tube throughout a compact interval is determined
by a countable dense set of times together with the right endpoint, so these
events are measurable.

\begin{assumption}[Assembly-level physical model adequacy]
\label{ass:physical-adequacy}
For every successfully compiled design $D$ and every $u\in\Uadm$, let
\begin{equation}
 \varnothing\ne\mathfrak P_{\mathrm{phys}}(D,u)
 \subseteq
 \mathcal P\!\left(
   \mathbb D_T\!\left(\prod_{j=1}^{m}\mathsf Y_{y_j}\right)
 \right)
\end{equation}
be the nonempty family of joint output-path laws of the assembled physical
implementation operated within the declared chassis, initial-condition, and
environmental regime. The adequacy requirement is
\begin{equation}
 \mathfrak P_{\mathrm{phys}}(D,u)
 \subseteq
 \left\{
   \Law_{\mathbb P}(\mathbf Y_D)
   \;\middle|\;
   \theta\in\thetaSet,\ 
   \mathbb P\in\mathfrak P_D(u,\theta)
 \right\}.
 \label{eq:physical-law-inclusion}
\end{equation}
Global consistency means that one parameter vector satisfies all cross-entry aliases and constraints, one admitted initial law is used for the assembled state, and every context component is generated by one jointly admissible policy and noise construction. Gate-wise choices cannot be combined if they do not arise from the same global law.
\end{assumption}

Assumption~\ref{ass:physical-adequacy} is not proved by the compiler. It is the empirical biological trust boundary and is needed only for the physical-realization corollary.

\begin{definition}[Robust stochastic temporal refinement]
\label{def:refinement}
The exact compiled model family of $D$ refines $\mathbf F$ under $\eta$, written $D\preceq_\eta\mathbf F$, when
\begin{equation}
 \inf_{u\in\Uadm}
 \inf_{\theta\in\thetaSet}
 \inf_{\mathbb P\in\mathfrak P_D(u,\theta)}
 \mathbb P\!\left[
   \Good_{\mathbf y}^{\llbracket\mathbf F\rrbracket(u)}([\tau,T])
 \right]
 \ge p.
 \label{eq:refinement}
\end{equation}
\end{definition}

The infima make the statement robust to every declared parameter, initial condition, context policy, and permitted dependence structure. Refinement requires agreement for every output on every admissible assignment. It is not an existential satisfiability claim or an average score.
The notation $D\preceq_\eta\mathbf F$ suppresses the other fixed components of instance $\mathcal I$, including $\Uadm$, $\thetaSet$, $\mathcal I_0$, $\context$, and the output bindings. The present contract is the bounded persistence property that all requested outputs remain in their tubes on $[\tau,T]$. It is not a syntax or semantics for arbitrary temporal-logic formulas.

Write $\compiler$ for the PCS-Bio proof producer. It searches for a design and
its evidence, but need not decide every compilation instance. When it
terminates, it returns one of the following outcomes:
\begin{equation}
 \compiler(\mathcal I)\in
 \left\{
 \begin{array}{l}
 \mathsf{success}(D,G_D,\mathsf{SBOL}_D,\mathcal M_D,\cert),\\
 \mathsf{unrealizable}(\cert_{\mathrm{unr}}),\\
 \mathsf{unknown}
 \end{array}
 \right\}.
 \label{eq:compiler-output}
\end{equation}
The present soundness theorem concerns only a returned $\mathsf{success}$ tuple. A claimed $\mathsf{unrealizable}$ result requires a separate independently checkable infeasibility certificate. Failure to find a mapping is not such a proof. Write $\checker$ for the
ideal PCS-Bio verifier, which checks the producer's output against the trusted
instance. Its mathematical specification is given in Section~\ref{sec:checker}.
The required soundness implication is
\begin{equation}
 \checker(\mathcal I,D,G_D,\mathsf{SBOL}_D,\mathcal M_D,\cert)
 =\mathsf{accept}
 \quad\Longrightarrow\quad
 D\preceq_\eta\mathbf F.
 \label{eq:target-soundness}
\end{equation}
A separate corollary uses Assumption~\ref{ass:physical-adequacy} to transfer this model-relative conclusion to the assembled physical implementation.

\section{Boolean compilation and sequence binding}
\label{sec:compiler}
To establish refinement in Definition~\ref{def:refinement}, the proof must
first identify what was compiled. Section~\ref{sec:nor-translation} constructs
a NOR graph and checks its Boolean equivalence to the source.
Section~\ref{sec:binding-stage} then binds the mapped graph to DNA records
and a model derivation. These structural checks identify the model to which
the stochastic obligations in the following sections apply.

\subsection{Dual-rail NOR compilation}\label{sec:nor-translation}
NOR, defined by $\Nor(a,b)=\neg(a\lor b)$, is functionally complete.
Characterized Cello libraries contain repressor-based entries whose Boolean
abstraction is NOR~\cite{Nielsen2016}. PCS-Bio uses a \emph{dual-rail
encoding}: for every source subformula $\varphi$, it constructs a positive
rail $p_\varphi$, intended to represent $\varphi$, and a negative rail
$n_\varphi$, intended to represent $\neg\varphi$. These meanings are proved in
Proposition~\ref{prop:boolean-preservation}.

For constants, the intended rail values are $(p_0,n_0)=(0,1)$ and
$(p_1,n_1)=(1,0)$. Any constant rail that remains after checker-replayed
simplification must be implemented by a characterized constant source. If the
technology library has no compatible source, the checker rejects that
candidate mapping. This does not by itself prove the complete compilation
instance unrealizable.

For a variable $x_i$, characterized sensor outputs provide the positive and
negative rails. If only one polarity is available, one duplicated-input NOR
produces the other:
\begin{equation}
 n_{x_i}=\Nor(p_{x_i},p_{x_i})
 \quad\text{or}\quad
 p_{x_i}=\Nor(n_{x_i},n_{x_i}),
\end{equation}
respectively. If neither polarity can be obtained from a characterized sensor,
the candidate mapping is rejected. Negation exchanges the rails:
\begin{equation}
 p_{\neg\varphi}=n_\varphi,
 \qquad
 n_{\neg\varphi}=p_\varphi.
\end{equation}
For disjunction and conjunction, define
\begin{align}
 q_{\lor}&=\Nor(p_\varphi,p_\psi),
 &p_{\varphi\lor\psi}&=\Nor(q_{\lor},q_{\lor}),
 &n_{\varphi\lor\psi}&=q_{\lor},
 \label{eq:orcompile}\\
 p_{\varphi\land\psi}&=\Nor(n_\varphi,n_\psi),
 &n_{\varphi\land\psi}&=\Nor(p_{\varphi\land\psi},p_{\varphi\land\psi}).
 \label{eq:andcompile}
\end{align}

The shared source graph stores each directed-acyclic-graph node once, and every
occurrence uses the same rail pair. Each generated NOR gate receives a distinct
auxiliary variable in the later clausal encoding. If the compiler applies
constant propagation, structural sharing, or another semantics-preserving
rewrite, the checker independently replays it.

\begin{proposition}[Boolean preservation]
Let $C_{\mathbf F}$ be the dual-rail NOR directed acyclic graph produced by the
rules above. For every $u\in\Bool^n$ and every source subformula $\varphi$,
\begin{equation}
 \llbracket p_\varphi\rrbracket_{C_{\mathbf F}}(u)
 =\llbracket\varphi\rrbracket(u),
 \qquad
 \llbracket n_\varphi\rrbracket_{C_{\mathbf F}}(u)
 =1-\llbracket\varphi\rrbracket(u).
 \label{eq:boolean-preservation}
\end{equation}
If the shared source graph has $N$ nodes, the construction introduces at most
$2N$ NOR gates before simplification and therefore has linear size.
\label{prop:boolean-preservation}
\end{proposition}

\begin{proof}
Proceed by structural induction. The constant rails have the stated values.
For a variable, a characterized sensor supplies at least one correct rail, and
a duplicated-input NOR, when needed, computes its complement. Negation
exchanges two correct rails. Equations~\eqref{eq:orcompile} and
\eqref{eq:andcompile} are De Morgan identities. A variable introduces at most
one inverter, a binary connective at most two NOR gates, and constants and
negations none. Sharing prevents repeated generation for the same source-graph
node.
\end{proof}

After technology mapping, the checker reconstructs the mapped Boolean output
vector $\mathbf A_D=(A_{D,1},\ldots,A_{D,m})$ from the typed graph and the
validated truth table of each selected technology entry. It also reconstructs
a predicate $A_{\mathrm{adm}}$ satisfying
\[
 A_{\mathrm{adm}}(u)=1
 \quad\Longleftrightarrow\quad
 u\in\Uadm.
\]
The domain-restricted Boolean miter is
\begin{equation}
 \Phi_D^{\mathrm{mit}}(X)
 =A_{\mathrm{adm}}(X)\land
 \bigvee_{j=1}^{m}
 \bigl(F_j(X)\oplus A_{D,j}(X)\bigr).
 \label{eq:boolean-miter}
\end{equation}
A satisfying assignment is therefore an admissible input on which at least one
mapped output differs from its source formula.

The ideal checker converts the miter to conjunctive normal form (CNF) by a
fixed checker-defined construction and validates a deletion resolution
asymmetric tautology (DRAT) refutation~\cite{Wetzler2014}. A DRAT refutation is
a machine-checkable sequence of clause steps that derives the empty clause and
thereby certifies unsatisfiability. The current executable paths accept only
proofs containing reverse-unit-propagation (RUP) clause additions, a
deletion-free fragment accepted by DRUP/DRAT checkers. The checker applies the
proof only to CNF clauses that it reconstructed itself. Acceptance therefore
proves equivalence on $\Uadm$ and makes no claim about inputs outside that set.

\subsection{Technology mapping and sequence binding}\label{sec:binding-stage}
Boolean equivalence identifies the computed function but leaves the biological
implementation unspecified. The next judgment connects the mapped components,
assembled sequence, and designated model so that a later stochastic proof
cannot be attached to an unrelated design.

A mapped netlist assigns each logical node $v$ to a technology entry
$\mu(v)\in\tech$. The ideal verifier $\checker$ verifies that the selected entry implements
the required Boolean function, its ports have compatible carriers, directions,
and units, its declared orthogonality and fan-out constraints hold, its output
ranges satisfy downstream thresholds, and the design satisfies the declared
chassis, placement, and resource bounds. These checks establish consistency
with the technology package. They do not infer empirical performance from a DNA
sequence alone.

For every selected entry $e$, let $\pi_e$ be the declared projection from the
global parameter vector to the coordinates used by $e$. The checker proves
\begin{equation}
 \pi_e(\thetaSet)\subseteq\thetaSet_e.
\end{equation}
If two entries use an aliased parameter, both projections refer to the same
global coordinate. The proof producer $\compiler$ may search by satisfiability modulo theories
(SMT), mixed-integer optimization, or another method, but acceptance depends
on a concrete mapping and
independently replayable evidence for its constraints.

Write $w_D\in\{A,C,G,T\}^{*}$ for the exact ordered nucleotide sequence of
design $D$. Each mapped gate is expanded by the selected assembly grammar into
promoter, operator, ribosome-binding-site, coding-sequence, terminator,
insulation, and other required elements. The checker constructs the canonical
sequence from validated library records and compares it directly with $w_D$.
Cryptographic digests identify and authenticate the source records but do not
replace this sequence comparison. For a circular construct, the design records
its topology and orientation and specifies a canonical cut point that yields a
unique linear representation.

\begin{definition}[Closed-world sequence, graph, and model binding]
\label{def:sequence-binding}
Here \emph{closed world} means that every sequence feature, regulatory edge,
and model component relevant to the claim has a declared origin, and that no
required element may be silently omitted. The judgment
\begin{equation}
 \tech\vdash
 \mathsf{bound}(w_D,\mathsf{SBOL}_D,G_D,\mathcal M_D)
 \label{eq:sequence-binding}
\end{equation}
is relative to the verifier-controlled technology library $\tech$. It holds
exactly when the checker establishes all of the following conditions:
\begin{enumerate}
 \item $w_D$ is the exact grammar-directed concatenation, or canonical
 linearization of a circular assembly, of the selected library sequences, with
 checked topology, orientation, and coordinates,
 \item the canonical nucleotide sequence represented by $\mathsf{SBOL}_D$ is
 exactly $w_D$, and its sequence-bearing features and locations agree with the
 grammar-derived assembly,
 \item every SBOL feature and location is valid under the selected SBOL version
 and assembly grammar, an overlap is accepted exactly when both permit it,
 \item checker-replayed regulatory-graph extraction from $\mathsf{SBOL}_D$
 yields exactly $G_D$, together with explicitly declared chassis or context
 interactions, so no relevant edge is missing and no unexplained edge is
 present,
 \item the checker regenerates $\mathcal M_D$ from mechanism templates
 identified by cryptographic digests, or replays a complete derivation yielding
 the same canonical model,
 \item every species, reaction, mode, parameter, and observation map required
 by the selected parts, mechanisms, chassis, and context appears in
 $\mathcal M_D$, and every such model element has one of those declared origins,
 \item parameter aliases, uncertainty constraints, initial conditions,
 resource and context interfaces, port types and units, output bindings,
 observation maps, and Boolean thresholds are mutually consistent across all
 artifacts.
\end{enumerate}
\end{definition}

Definition~\ref{def:sequence-binding} states the complete judgment required by
the ideal architecture. No present execution verifies all of its clauses. The
Cello track reconstructs evidence from a pinned User Constraint File (UCF),
which records the gate library and design constraints, together with mapping
and sequence evidence. The
restricted exact checker separately verifies finite truth tables and temporal
obligations for a synthetic finite-state test instance. Because no
checker-verified derivation links the Cello-derived sequence to that model, the
current evidence does not establish~\eqref{eq:sequence-binding} and cannot
support an integrated acceptance verdict.

The binding judgment establishes consistency among the declared artifacts. It
does not establish Assumption~\ref{ass:physical-adequacy}. A characterized non-NOR implementation may replace a NOR subnetwork only if
its technology entry exposes every state variable needed to determine its declared dynamics,
specifies the admitted initialization, supplies a checked typed Boolean
abstraction, and discharges the same stochastic interface obligations. At the
circuit level, the exposed input--output abstraction must still fit the acyclic
combinational graph. Stateful recombinase memory is outside the present
theorem. Covering it would require extending both the source semantics and the
refinement theorem to an automaton-valued contract.

\section{Sequence-bound stochastic semantics}
\label{sec:semantics}

The binding judgment identifies the declared model, but does not yet bound its
behavior. This section describes the molecular processes that model may use
and the regularity conditions required by the proofs. The logical graph
organizes the composition argument. It does not eliminate physical coupling
between modules.

\subsection{Feed-forward molecular modules}
The checked logical structure is the output-indexed directed acyclic graph
$G_D=(V,E,\mathbf r)$, where $V$ is the component set, $E$ contains directed
signal-dependency edges, and $\mathbf r=(r_1,\ldots,r_m)$ binds named outputs
to graph roots. The same root may occur more than once. Each occurrence is then
a distinct named output bound to the same internal node. The checker verifies
compatible units, observation maps, tubes, and Boolean interpretations for all
such occurrences.

Acyclicity concerns only the Boolean dependency graph. It does not imply
stochastic independence of the molecular modules. The target model may couple
modules through explicit host or shared-resource variables, but every such
coupling must occur in the global process or in a certified context envelope.
Write $X_v$ for the molecular coordinates of module $v$. They may be discrete,
continuous, or mixed, and the module may also have a finite promoter mode
$M_v$. A well-mixed stochastic chemical reaction network with integer molecule
counts is naturally modeled as a continuous-time Markov chain
(CTMC)~\cite{Gillespie1977}. More generally, the target may be a declared
nonexplosive c\`adl\`ag jump, diffusion, or hybrid process. Here
\emph{nonexplosive} means that it remains defined on $[0,T]$ without
finite-time blow-up or an accumulation of event times.

A promoter can switch between an inactive state and an active state that
permits transcription. Transcription produces messenger RNA (mRNA), and
translation of mRNA produces protein. Random switching and molecule production
make the protein count fluctuate even under a fixed logical input. The
following telegraph mechanism represents this stochastic expression process~\cite{Golding2005,Raj2008}.
\begin{equation}
 \begin{aligned}
 G_{\mathrm{off}}
 &\xrightleftharpoons[k_{\mathrm{off}}(U_{\mathrm{rep}},H)]{k_{\mathrm{on}}(U_{\mathrm{rep}},H)}
 G_{\mathrm{on}},
 &G_{\mathrm{on}}&\xrightarrow{k_{\mathrm{tx}}}G_{\mathrm{on}}+X_{\mathrm{mRNA}},\\
 X_{\mathrm{mRNA}}&\xrightarrow{\delta_m}\varnothing,
 &X_{\mathrm{mRNA}}&\xrightarrow{k_{\mathrm{tl}}}X_{\mathrm{mRNA}}+X_{\mathrm{prot}},\\
 X_{\mathrm{prot}}&\xrightarrow{\delta_p}\varnothing.
 \end{aligned}
 \label{eq:telegraph}
\end{equation}
In \eqref{eq:telegraph}, $G_{\mathrm{off}}$ and $G_{\mathrm{on}}$ are the two
promoter modes. The upper rate $k_{\mathrm{on}}$ governs the transition from
off to on, and the lower rate $k_{\mathrm{off}}$ governs the reverse
transition. The variables $X_{\mathrm{mRNA}}$ and $X_{\mathrm{prot}}$ denote
messenger RNA and protein, and $\varnothing$ denotes removal from the modeled
system. The signal $U_{\mathrm{rep}}$ is the repressor input, while $H$ is the
realized host or resource context generated by a policy $h\in\context$. The
rates $k_{\mathrm{tx}}$ and $k_{\mathrm{tl}}$ govern transcription and
translation, $\delta_m$ and $\delta_p$ degradation, and the switching rates
may depend on both the repressor input and the context. With molecule-count
states and one promoter copy, the transcription propensity is
$k_{\mathrm{tx}}\mathbf 1_{\{G_{\mathrm{on}}\}}$, the translation propensity
is $k_{\mathrm{tl}}X_{\mathrm{mRNA}}$, and the degradation propensities are
$\delta_mX_{\mathrm{mRNA}}$ and $\delta_pX_{\mathrm{prot}}$. Each propensity
has units of events per unit time.

For proof search and rapid counterexample detection, one may use the following
finite-activity switching jump-diffusion as a surrogate:
\begin{align}
 dX_t={}&f_{M_{t^-}}(t,X_{t^-},U_{t^-},H_{t^-},\theta)\,dt
 +B_{M_{t^-}}(t,X_{t^-},U_{t^-},H_{t^-},\theta)\,dW_t\nonumber\\
 &+\int_{\mathcal Z}
 g_{M_{t^-}}(t,X_{t^-},U_{t^-},H_{t^-},\theta,z)\,
 N(dt,dz),
 \label{eq:sjd}
\end{align}
Such a surrogate supports a claim about the designated target model only after a checked
transfer certificate bounds their path discrepancy, as specified in \cref{sec:transfer}.
Here $X_t$ is the continuous surrogate state, $M_t$ is a finite-valued mode,
$U_t$ and $H_t$ are the input and context processes, and $W$ is a standard
Brownian motion of the declared dimension. The Brownian integral uses the It\^o convention. The functions $f$, $B$, and $g$ are
the drift, diffusion coefficient, and marked jump increment. We use standard
jump-process notation~\cite{Applebaum2009}.

The integral in~\eqref{eq:sjd} uses the raw (uncompensated) integer-valued
random measure $N$. Its compensator is not subtracted from the jump term or
absorbed into $f$. Its predictable compensator is
\begin{equation}
 \nu_{M_{t^-}}(t,X_{t^-},U_{t^-},H_{t^-},\theta,dz)\,dt.
 \label{eq:jump-compensator}
\end{equation}
For this representation, $N$ is simple in time, so at most one mark occurs
at any event time, almost surely. Its event times and the mode-switching times
are disjoint almost surely. Their rates may depend on the same state, inputs,
and context, so these conventions do not imply independence of the resulting
processes. An event that changes both $X$ and $M$ must instead be represented
by one joint transition with its own landing map and compensator. Several
simultaneous marks are aggregated into one mark.

For distinct modes $\mu$ and $\mu'$, let $\lambda_{\mu\mu'}$ be the
predictable transition rate from $\mu$ to $\mu'$. Finite activity requires,
under every admitted surrogate law,
\begin{align}
 &\int_0^T
 \nu_{M_{s^-}}(s,X_{s^-},U_{s^-},H_{s^-},\theta,\mathcal Z)\,ds<\infty,
 \nonumber\\
 &\int_0^T\sum_{\mu'\ne M_{s^-}}
 \lambda_{M_{s^-}\mu'}(s,X_{s^-},U_{s^-},H_{s^-},\theta)\,ds<\infty
 \qquad\text{almost surely}.
 \label{eq:finite-activity}
\end{align}
Thus the jump measure and promoter mode have finitely many events on each
bounded interval. An infinite-activity L\'evy surrogate requires a separately
declared decomposition and integrability conditions.

\begin{assumption}[Existence and regularity of target and surrogate laws]
\label{ass:wellposed}
For every successfully compiled $D$, every $u\in\Uadm$,
$\theta\in\thetaSet$, $\nu_0\in\mathcal I_{0,D}(u)$, and
$h\in\context$,
\[
 \operatorname{Sol}(\mathcal M_D,u,\theta,\nu_0,h)\ne\varnothing.
\]
Uniqueness in law is not required. If several admitted weak-solution laws
exist, all belong to $\mathfrak P_D(u,\theta)$. Under every admitted target or
surrogate law, as applicable, the following conditions hold:
\begin{enumerate}
 \item The canonical process is adapted, c\`adl\`ag, nonexplosive on $[0,T]$,
 and preserves the declared state space or follows its stated boundary rule.
 \item Inputs and exogenous context are adapted. Whenever they enter a
 coefficient or intensity, the model uses a declared predictable version, such
 as the left limit of an adapted c\`adl\`ag process.
 \item Every observation map used in a threshold event or discrepancy bound is
 continuous.
 \item Every invoked hitting, exit, and assumption-violation time is a stopping
 time under the topology and prefix-measurability conditions stated for that
 construction.
 \item Every certificate function belongs to the relevant extended-generator
 domain. The required It\^o--Dynkin identity and all stated integrability and
 tail bounds hold uniformly over the declared uncertainty sets.
 \item After the stopping operation used in a certificate proof, each local
 martingale term is a true uniformly integrable martingale, so the stated
 expectation and optional-stopping steps are valid.
\end{enumerate}
Any synchronous coupling used for transfer is defined on a complete filtered
probability space and satisfies the corresponding measurability,
integrability, and common-data conditions.
\end{assumption}

A diffusion surrogate does not automatically preserve nonnegative molecular
coordinates. For an unreflected model on a nonnegative orthant, a
state-preservation argument must cover diffusion at the boundary, the drift
direction, and every jump or reset landing. Reflection, clipping, or absorption
changes the boundary rule and must be declared in the model, its generator,
and any transfer proof. These choices cannot be inferred from a biological
interpretation of the state variable. The finite-state grid process evaluated
in Section~\ref{sec:computational-methods} is a separate model and is not
changed by these continuous-state requirements.

\cref{tab:regularity-status} separates the status of the clauses in Assumption~\ref{ass:wellposed}. The ideal theorem permits a clause to be discharged by a restricted-model theorem or an auxiliary checked certificate. If neither route is present, the clause remains an external mathematical assumption. The floating-point prototype does not implement these discharge routes.
These assumptions specify which process laws the proofs may use. We next
express Boolean values as measurable events on those laws.

\begin{table}[t]
\centering
\small
\caption{Status of the semantic and regularity obligations used by the soundness theorem. TCB denotes the trusted computing base.}
\label{tab:regularity-status}
\begin{tabularx}{\linewidth}{@{}>{\raggedright\arraybackslash}p{0.28\linewidth}>{\raggedright\arraybackslash}p{0.24\linewidth}Y@{}}
\toprule
Obligation & Formal status & Accepted basis in the ideal architecture \\
\midrule
Nonexplosion and state preservation & Model theorem or external assumption & A theorem for a restricted admitted model class or a separately checked auxiliary proof \\
Adapted c\`adl\`ag paths and filtration & Model semantics and TCB & The declared process class and its trusted semantic implementation \\
Policy nonanticipation & Checker-replayed syntax & Predictable policy kernels and predictable coefficient arguments \\
Observation continuity and set topology & Checker-replayed syntax & Admitted continuous expressions and declared closed or open sets \\
State hitting and exit times & Mathematical consequence & Path regularity and the topology of declared state sets \\
Assumption-violation times & Checker-replayed syntax and mathematical consequence & Prefix measurability and the interval-consistency conditions in \eqref{eq:assumption-prefix} and \eqref{eq:assumption-stopping-event} \\
Generator domain and martingale integrability & Backend proof or external assumption & A backend-specific theorem, checked integrability certificate, or explicit external assumption \\
Jump and tail integrability & Checker-replayed bound or external assumption & Exact finite sums, a checked truncation and tail bound, or an explicit external assumption \\
\bottomrule
\end{tabularx}
\end{table}

\section{Signal tubes and source-leaf contracts}
\label{sec:signals}
Composition needs a shared meaning for a Boolean value on every molecular
wire. We first specify that meaning by low and high observation thresholds,
then state the source contracts that initiate the gate-by-gate proof.
Source leaves include sensors and any constant signals that remain after
Boolean simplification.
For each molecular wire $v$, let $Z_v$ be the corresponding wire-state
process, taking values in a state space $\mathsf X_v$. The wire has a
continuous scalar observation map
$o_v\colon\mathsf X_v\to\mathbb R$. Define the closed low and high sets in
the physical observation space by
\begin{equation}
 \mathsf T_v^0=(-\infty,L_v],
 \qquad
 \mathsf T_v^1=[H_v,\infty),
 \qquad
 L_v<H_v.
 \label{eq:observation-tubes}
\end{equation}
Their pullbacks to the exact wire state are
\begin{equation}
 \tube_v^b=o_v^{-1}(\mathsf T_v^b),
 \qquad b\in\Bool.
 \label{eq:wire-tubes}
\end{equation}
Values in the open band $(L_v,H_v)$ receive no Boolean interpretation. For a compact interval $I\subseteq[0,T]$, define the path events
\begin{equation}
 \Good_v^b(I)=
 \{\omega\mid Z_v(t,\omega)\in\tube_v^b
 \text{ for every }t\in I\},
 \qquad
 \Bad_v^b(I)=\Omega_D\setminus\Good_v^b(I).
 \label{eq:wire-events}
\end{equation}

Input sensors are clamped logically but may remain noisy physically. Define the
finite source-label set
\begin{equation}
 \mathcal L_{\mathrm{src}}
 =\{0,1\}\cup\{x_j,\neg x_j\mid 1\le j\le n\}.
\end{equation}
Let $\mathcal S_D$ be the physical source leaves, including sensor rails and
any residual characterized constant sources. Each leaf $i$ carries a label
$\ell_i\in\mathcal L_{\mathrm{src}}$ and requests
$b_i(u)=\llbracket\ell_i\rrbracket(u)$. Value-specific source-leaf data
$(\tau_{i,b},\eps_{i,b})\in[0,T]\times[0,1]$
satisfy the following bound for each $b\in\Bool$ attained by $b_i$ on $\Uadm$,
\begin{equation}
 \sup_{\substack{u\in\Uadm,\ b_i(u)=b\\
                  \theta\in\thetaSet\\
                  \mathbb P\in\mathfrak P_D(u,\theta)}}
 \mathbb P\!\left[
   \Bad_i^{b}([\tau_{i,b},T])
 \right]
 \le \eps_{i,b}.
 \label{eq:sensor-contract}
\end{equation}
Complement sensors have separate low and high contracts. A generated complement is an ordinary gate. A residual constant leaf must have its own characterized source contract for
the requested constant value. These bounds provide the base cases for the
composition proof. Internal gates require an additional condition describing
which parent values are sufficient to force their output.

\section{Sufficient-cube contracts and proof slicing}
\label{sec:slicing}
For each requested gate output, this section identifies a partial Boolean input
assignment that already forces that output. It then associates this Boolean
condition with a stochastic gate contract. The resulting witness follows only
the branches needed to force the output, while still constraining every
physically present but logically masked input through a verified envelope.

\subsection{Sufficient cubes}
Let $g_v\colon\Bool^{k_v}\to\Bool$ be the local Boolean function of gate $v$.
A cube is a vector $c\in\{0,1,\starv\}^{k_v}$, where $\starv$ denotes an
unfixed, or ``don't-care,'' coordinate. A Boolean vector
$a\in\Bool^{k_v}$ completes $c$, written $a\models c$, if $a_j=c_j$ at every
fixed coordinate.

\begin{definition}[Sufficient cube]
A cube $c$ is sufficient for output $b\in\Bool$ of gate $v$ if
\[
 \forall a\in\Bool^{k_v},\qquad a\models c\Longrightarrow g_v(a)=b.
\]
Its support is $\supp(c)=\{j\mid c_j\ne\starv\}$. A sufficient cube is
minimal if changing any one fixed coordinate to $\starv$ makes it
insufficient.
\end{definition}

For a NOR gate, the unique minimal cube for output one fixes every input to
zero. Each minimal cube for output zero fixes one input to one. For an AND gate, the unique minimal cube for output one fixes every input to
one, while each minimal cube for output zero fixes one input to zero. For a characterized finite-input gate that
is not decomposed into primitives, the checker validates cubes directly
against the checked truth table or a binary decision diagram (BDD)
reconstructed from that table.

For a source leaf $i$, set $b_i(u)=\llbracket\ell_i\rrbracket(u)$. Acyclic evaluation then defines, for every internal gate $v$,
\begin{equation}
 a_v(u)=\bigl(b_{\operatorname{pa}_1(v)}(u),\ldots,
 b_{\operatorname{pa}_{k_v}(v)}(u)\bigr),
 \qquad
 b_v(u)=g_v(a_v(u)).
 \label{eq:node-boolean-value}
\end{equation}
For assignment $u$, the PCS-Bio producer $\compiler$ selects a cube $c_v(u)$ completed by the
evaluated Boolean input $a_v(u)$:
\begin{equation}
 a_v(u)\models c_v(u)
 \quad\text{and}\quad
 c_v(u)\text{ is sufficient for }b_v(u)=g_v(a_v(u)).
 \label{eq:cube-selection-valid}
\end{equation}
For example, a local heuristic may enumerate the sufficient cubes compatible
with $a_v(u)$ and choose one with the smallest support, breaking ties by a
fixed port order. At a NOR gate with two high inputs, this chooses one of
$(1,\starv)$ and $(\starv,1)$. Such a rule need not minimize the total risk
or delay. Alternatively, a global optimizer may select the family
$\{c_v(u)\}$ jointly to account for reconvergent branches, the longest
witness path, and charges shared by several outputs. These are possible search
strategies, not an assertion about the heuristic used by the frozen runs. Optimality is not needed for soundness: the checker validates every
selected cube and reconstructs the resulting risk and delay functions.

\subsection{Masked-input robust gate contracts}
Let $\operatorname{pa}_j(v)$ be the parent connected to input port $j$ of gate
$v$. The binding checker identifies the local input process $U_{v,j}$ with the
process carried by this parent. If the physical connection transforms the
signal, that transformation must appear as an explicit modeled node.

Fix a selected cube $c$ that requests output $b$. For every
$j\in\supp(c)$, the parent process must remain in the tube representing the
fixed value $c_j$. An input with $j\notin\supp(c)$ is logically masked but is
still physically present. The contract must therefore hold uniformly over
every adapted c\`adl\`ag path for that input inside its declared envelope.

For an interval $[s,t]$, let $\mathcal E_{v,j}[s,t]$ be the allowed masked-input
path restrictions and let $\mathcal E_v^{\mathrm{ctx}}[s,t]$ be the
corresponding set for local context $H$. Whenever an input or context value
enters a jump coefficient, switching rate, or stochastic integral, use its
predictable left limit. Thus the coefficient at time $t$ depends only on
information available immediately before $t$. Define
\begin{align}
 A_{v,c}(s,t)={}&
 \bigcap_{j\in\supp(c)}
 \Good_{\operatorname{pa}_j(v)}^{c_j}([s,t])\nonumber\\
 &\cap
 \bigcap_{j\notin\supp(c)}
 \{U_{v,j}|_{[s,t]}\in\mathcal E_{v,j}[s,t]\}
 \cap
 \{H|_{[s,t]}\in\mathcal E^{\mathrm{ctx}}_v[s,t]\}.
 \label{eq:assumption-event}
\end{align}
The assumption family has two consistency properties. First,
$A_{v,c}(s,t)$ is determined by the path prefix through $t$, and extending the
right endpoint can only add requirements:
\begin{equation}
 A_{v,c}(s,t_2)\subseteq A_{v,c}(s,t_1)
 \quad\text{whenever }s\le t_1\le t_2\le T.
 \label{eq:assumption-prefix}
\end{equation}
Second, discarding an initial part of an interval can only remove requirements.
For every
$0\le s\le r\le t\le T$,
\begin{equation}
 A_{v,c}(s,t)\subseteq A_{v,c}(r,t).
 \label{eqassumptionrestriction}
\end{equation}
These properties follow because every masked-input and context envelope is
closed under restriction to subintervals.
The first violation time is
\begin{equation}
 \tau_A^s(\omega)=
 \inf\{t\in[s,T]\mid \omega\notin A_{v,c}(s,t)\},
 \qquad \inf\varnothing=\infty.
 \label{eq:assumption-violation-time}
\end{equation}
The checker verifies that $\tau_A^s$ is a stopping time and that
\begin{equation}
 \{\tau_A^s>t\}\subseteq A_{v,c}(s,t)
 \subseteq\{\tau_A^s\ge t\}.
 \label{eq:assumption-stopping-event}
\end{equation}
The two inclusions in~\eqref{eq:assumption-stopping-event} do not impose one
endpoint convention: at the first violation time, the assumption may already
have failed, or it may hold there and fail immediately afterward.

The construction must also support a random left endpoint. Let $\sigma$ be a
stopping time with $\sigma(\omega)\in[s,t]$ almost surely. Define
$A_{v,c}(\sigma,t)$ and its first violation time by applying the same interval
predicate pathwise on $[\sigma(\omega),t]$:
\begin{equation}
 \tau_A^\sigma(\omega)
 =\inf\{r\in[\sigma(\omega),T]\mid
 \omega\notin A_{v,c}(\sigma(\omega),r)\},
 \qquad \inf\varnothing=\infty,
 \label{eq:random-assumption-violation-time}
\end{equation}
The checker verifies the random-start restriction property
\begin{equation}
 A_{v,c}(s,t)\subseteq A_{v,c}(\sigma,t)
 \quad\text{almost surely under every applicable law}.
 \label{eqassumptionstoppedrestriction}
\end{equation}
It also verifies
\begin{equation}
 \{\tau_A^\sigma>t\}\subseteq A_{v,c}(\sigma,t)
 \subseteq\{\tau_A^\sigma\ge t\}
 \quad\text{almost surely under every applicable law}.
 \label{eq:random-assumption-stopping-event}
\end{equation}
The first property allows a stay certificate to begin at the random time at
which the process first reaches its core.

Choose global wire and context envelopes whose restrictions satisfy every
invoked local masked-input and context envelope. Let $\Env$ be the event that
these global envelopes hold throughout $[0,T]$. A global envelope cannot be
broader than a local admissibility condition unless additional evidence
establishes the required implication. The invariant
certificate establishes
\begin{equation}
 \sup_{u\in\Uadm}
 \sup_{\theta\in\thetaSet}
 \sup_{\mathbb P\in\mathfrak P_D(u,\theta)}
 \mathbb P[\Env^c]\le\epsenv,
 \label{eq:envelope-contract}
\end{equation}
where $\epsenv\in[0,1]$ is the failure-probability allowance for the global
envelope. For every invoked pair $(v,c)$ and every $0\le s\le t\le T$, the
checker also verifies that $\Env$ supplies the local envelope clauses in
$A_{v,c}(s,t)$:
\begin{equation}
 \Env\subseteq
 \bigcap_{j\notin\supp(c)}
 \{U_{v,j}|_{[s,t]}\in\mathcal E_{v,j}[s,t]\}
 \cap
 \{H|_{[s,t]}\in\mathcal E^{\mathrm{ctx}}_v[s,t]\}.
 \label{eqenvsupplieslocal}
\end{equation}
If $\Env$ is an exact invariant, then $\epsenv=0$. Otherwise its allowance
is included once in the global risk sum.

Let $I_{v,c}$ be the local start set for this gate obligation. For fixed $s$, define the measurable start event
\begin{equation}
 B_{v,c}(s)=\{Z_v(s)\in I_{v,c}\}.
 \label{eq:start-event}
\end{equation}
It belongs to $\mathcal F_s^{\mathbb P}$ under every admitted global law $\mathbb P$.

The contract $\mathcal C_{v,c}$ supplies the local failure charge used in the
global union bound. It bounds the probability of the joint event that the gate
starts in $I_{v,c}$, the selected-parent, masked-input, and context assumptions
hold, and the requested output nevertheless leaves its tube after the declared
delay. For an applicable local law $P$, it therefore bounds
\[
 P\!\left[B_{v,c}^{\mathrm{loc}}(s)\cap A_{v,c}^{\mathrm{loc}}(s,T)
          \cap\Bad_v^{b,\mathrm{loc}}([s+d_{v,c},T])\right],
\]
not the conditional probability of output failure given the start and
assumption events. This distinction is used by the composition proof below.

For an invoked obligation $(v,c,s)$, let $\Omega_{v,c}^{s}$ be its complete
local canonical path space on $[s,T]$. Let
\[
 \mathcal R_{v,c,s}\colon\Omega_D\longrightarrow\Omega_{v,c}^{s}
\]
be the measurable map that restricts a global path to the complete local state
used by that obligation. This state includes $Z_v$, every local input and
context coordinate, and every resource, mode, or hidden coordinate used by the
local characteristics or by the start, assumption, or output events.

Let $B_{v,c}^{\mathrm{loc}}(s)$,
$A_{v,c}^{\mathrm{loc}}(s,t)$, and
$\Bad_v^{b,\mathrm{loc}}(I)$ denote the canonical local events. The
model-interface proof verifies
\begin{equation}
 \begin{aligned}
 B_{v,c}(s)
   &=\mathcal R_{v,c,s}^{-1}\!\left(B_{v,c}^{\mathrm{loc}}(s)\right),\\
 A_{v,c}(s,t)
   &=\mathcal R_{v,c,s}^{-1}\!\left(A_{v,c}^{\mathrm{loc}}(s,t)\right),\\
 \Bad_v^b(I)
   &=\mathcal R_{v,c,s}^{-1}\!\left(\Bad_v^{b,\mathrm{loc}}(I)\right)
 \end{aligned}
 \label{eq:global-local-event-pullbacks}
\end{equation}
for every interval used by the obligation.

Every applicable global law $\mathbb P$ induces the local law
$\mathbb P\circ\mathcal R_{v,c,s}^{-1}$. A reusable technology certificate
declares a family $\mathfrak P_{v,c}^{s}$ of admissible local laws. After the
global parameters and ports have been bound to the local model, this family
must contain every induced law for which $(v,c,s)$ is invoked. It may contain
additional laws. This is conservative but permits certificate reuse. For each
$P\in\mathfrak P_{v,c}^{s}$, the local canonical process has the complete,
right-continuous filtration $\{\mathcal G_t^P\}_{t\in[s,T]}$ and satisfies the
applicable clauses of Assumption~\ref{ass:wellposed}.

\begin{definition}[Delayed stochastic gate contract]
\label{def:gate-contract}
Gate implementation $v$ satisfies
\begin{equation}
 \mathcal C_{v,c}=(b,c,I_{v,c},d_{v,c},\eps_{v,c},\tube_v^b)
 \label{eq:gate-contract-tuple}
\end{equation}
when $d_{v,c}\in[0,T]$, $\eps_{v,c}\in[0,1]$, and, for every $s\in[0,T-d_{v,c}]$ and every $P\in\mathfrak P_{v,c}^{s}$,
\begin{equation}
 P\!\left[
 B_{v,c}^{\mathrm{loc}}(s)\cap A_{v,c}^{\mathrm{loc}}(s,T)\cap
 \Bad_v^{b,\mathrm{loc}}([s+d_{v,c},T])
 \right]
 \le\eps_{v,c}.
 \label{eq:gate-contract}
\end{equation}
The bound is uniform over every stochastic dependence structure represented
in $\mathfrak P_{v,c}^{s}$. Inputs and context remain random adapted processes. The contract does not condition on a completely specified future trajectory.
\end{definition}

Write $\operatorname{out}(\mathcal C_{v,c})$ for the requested output bit $b$ in~\eqref{eq:gate-contract-tuple}. For every $u\in\Uadm$ and every internal gate whose contract is selected for $u$, the checker requires
\begin{equation}
 \operatorname{out}(\mathcal C_{v,c_v(u)})=b_v(u).
 \label{eq:selected-contract-output}
\end{equation}

Masked inputs are not omitted physically. They are universally quantified within their envelopes. If a masked input can antagonize a controlling input, the corresponding contract is invalid unless that input is fixed by the cube or its envelope excludes the dangerous regime.

\subsection{Joint witness subgraphs}
Use the Boolean node values from~\eqref{eq:node-boolean-value}. For
$u\in\Uadm$, let $\witness(u)$ be the smallest subset of $V$ such that
\begin{equation}
 \{r_1,\ldots,r_m\}\subseteq\witness(u)
 \quad\text{and}\quad
 \left.
 \begin{array}{c}
  v\in\witness(u),\ v\text{ is internal},\\
  j\in\supp(c_v(u))
 \end{array}
 \right\}
 \Longrightarrow
 \operatorname{pa}_j(v)\in\witness(u).
 \label{eq:witness-recursion}
\end{equation}
Thus only parents fixed by the selected cube are followed from a given gate.
A parent omitted at one gate may nevertheless enter through another selected
branch. The recursion terminates at source leaves because $G_D$ is acyclic.
Since $\witness(u)$ is a set, every reconvergent node is counted once.

Define witness timing recursively by
\begin{align}
 \tau_i(u)&=\tau_{i,b_i(u)}
 &&\text{for a source leaf }i,
 \label{eq:leaf-delay}\\
 s_v(u)&=\max_{j\in\supp(c_v(u))}
          \tau_{\operatorname{pa}_j(v)}(u),
 \qquad
 \tau_v(u)=s_v(u)+d_{v,c_v(u)}
 &&\text{for an internal gate }v,
 \label{eq:gate-delay}
\end{align}
The maximum over an empty support is zero. Thus $s_v(u)$ is the time at which
all parents selected by the cube are available, and $\tau_v(u)$ is the end of
the gate's own certified delay.
The model-interface proof must establish, for every $u\in\Uadm$, $\theta\in\thetaSet$, $\mathbb P\in\mathfrak P_D(u,\theta)$, and internal $v\in\witness(u)$,
\begin{equation}
 \mathbb P\circ\mathcal R_{v,c_v(u),s_v(u)}^{-1}
 \in\mathfrak P_{v,c_v(u)}^{s_v(u)}.
 \label{eq:global-local-law-inclusion}
\end{equation}
The joint delay and one-time-counted risk are
\begin{align}
 \Delay(u)&=\max_{1\leq j\leq m}\tau_{r_j}(u),
 \label{eq:joint-witness-delay}\\
 \Risk(u)&=\epsenv+
 \sum_{i\in\witness(u)\cap\mathcal S_D}
 \eps_{i,b_i(u)}
 +\sum_{v\in\witness(u)\setminus\mathcal S_D}
 \eps_{v,c_v(u)}.
 \label{eq:witness-risk}
\end{align}
The quantity $\Risk(u)$ is an uncapped union-bound budget, not an independently
estimated physical failure probability. The envelope allowance and every node
shared among branches or outputs are included once. Counting a shared
allowance more than once would remain sound but would weaken the bound.

The interface and context evidence must also show that every invoked gate
starts in the set for which its local contract was proved. Specifically, for
$u\in\Uadm$, an internal $v\in\witness(u)$, $c=c_v(u)$,
$s=s_v(u)$, $\theta\in\thetaSet$, and
$\mathbb P\in\mathfrak P_D(u,\theta)$, it must establish
\begin{equation}
 \mathbb P\!\left[\Env\cap B_{v,c}(s)^c\right]=0.
 \label{eq:start-set-coverage}
\end{equation}
Equivalently, whenever the global envelope holds, the gate is in its declared
start set almost surely at the invocation time.

\begin{theorem}[Mask-aware joint compositional soundness]
\label{thm:composition}
Assume that every selected cube satisfies~\eqref{eq:cube-selection-valid} and
that each selected contract requests the corresponding Boolean value. Assume
that the global envelope and all source leaves satisfy
\eqref{eq:envelope-contract}, \eqref{eqenvsupplieslocal}, and \eqref{eq:sensor-contract}. For every
selected internal-gate obligation, assume the exact-model gate contract, event
pullbacks, induced-law inclusion, and start-set coverage in
\eqref{eq:gate-contract}, \eqref{eq:global-local-event-pullbacks}, \eqref{eq:global-local-law-inclusion}, and \eqref{eq:start-set-coverage}.

Then, for every $u\in\Uadm$, $\theta\in\thetaSet$, and
$\mathbb P\in\mathfrak P_D(u,\theta)$ with $\Delay(u)\le T$,
\begin{equation}
 \mathbb P\!\left[
  \bigcup_{j=1}^{m}
  \Bad_{r_j}^{b_{r_j}(u)}([\tau_{r_j}(u),T])
 \right]
 \le\Risk(u).
 \label{eq:composition-bound}
\end{equation}
Thus, when $\Risk(u)\le1$, every root remains correct after its own witness
delay with probability at least $1-\Risk(u)$. No independence assumption is
required.
\end{theorem}

\begin{proof}
Fix $u$, $\theta$, and $\mathbb P$. Every witness node lies on a selected
path to a root, and all delays are nonnegative. Thus $\Delay(u)\le T$
implies $s_v(u)+d_{v,c_v(u)}\le T$ for each invoked gate, as required by its
contract. By~\eqref{eq:cube-selection-valid} and \eqref{eq:selected-contract-output}, each selected contract requests the Boolean value $b_v(u)$. For every internal witness node $v$, set
\begin{equation}
 s_v=s_v(u),
 \qquad
 J_v=[s_v+d_{v,c_v(u)},T],
 \label{eq:composition-node-interval}
\end{equation}
and define
\begin{equation}
 E_v=B_{v,c_v(u)}(s_v)\cap A_{v,c_v(u)}(s_v,T)
 \cap\Bad_v^{b_v(u)}(J_v).
 \label{eq:composition-local-event}
\end{equation}
Global-to-local law inclusion~\eqref{eq:global-local-law-inclusion} makes
$\mathbb P\circ\mathcal R_{v,c_v(u),s_v}^{-1}$ a member of
$\mathfrak P_{v,c_v(u)}^{s_v}$. The event $E_v$ is the pullback of the
corresponding local event under $\mathcal R_{v,c_v(u),s_v}$, so the gate
contract gives
\begin{equation}
 \mathbb P[E_v]\le\eps_{v,c_v(u)}.
 \label{eq:composition-local-bound}
\end{equation}

Suppose that $\Env$ holds and that each selected parent $w$ of $v$ is good on
$[\tau_w(u),T]$. Since $\tau_w(u)\le s_v$, each such parent is good on
$[s_v,T]$. The envelope event supplies every masked-input and context clause in
$A_{v,c_v(u)}(s_v,T)$. Moreover, \eqref{eq:start-set-coverage} implies that
$B_{v,c_v(u)}(s_v)$ holds on $\Env$ outside a $\mathbb P$-null set. All event
inclusions in the remainder of this proof are therefore understood modulo
$\mathbb P$-null sets. Consequently,
\begin{equation}
 \Env\cap\Bad_v^{b_v(u)}(J_v)
 \subseteq
 E_v\cup
 \bigcup_{j\in\supp(c_v(u))}
 \Bad_{\operatorname{pa}_j(v)}^{c_{v,j}(u)}
 \bigl([\tau_{\operatorname{pa}_j(v)}(u),T]\bigr).
 \label{eq:composition-one-step}
\end{equation}
Repeated substitution terminates because $G_D$ is acyclic. Applying the result to every root and taking their union gives
\begin{align}
 &\bigcup_{j=1}^{m}
 \Bad_{r_j}^{b_{r_j}(u)}([\tau_{r_j}(u),T])\nonumber\\
 &\quad\subseteq
 \Env^c
 \cup\bigcup_{v\in\witness(u)\setminus\mathcal S_D}E_v
 \cup\bigcup_{i\in\witness(u)\cap\mathcal S_D}
 \Bad_i^{b_i(u)}([\tau_{i,b_i(u)},T]).
 \label{eq:composition-expanded-inclusion}
\end{align}
Subadditivity, the source-leaf and envelope bounds, and~\eqref{eq:composition-local-bound} prove~\eqref{eq:composition-bound}. Shared random variables may occur in several events, which does not affect the union bound.
\end{proof}

\begin{corollary}[Uniform joint refinement from witness bounds]
\label{cor:witness-refinement}
Under all hypotheses of \cref{thm:composition}, assume that the model interface defines output $Y_{D,j}$ as the declared projection associated with root $r_j$. In particular, for every admitted global law and every
$j\in\{1,\ldots,m\}$,
\begin{equation}
 \mathbb P\!\left[
 o_{y_j}(Y_{D,j}(t))=o_{r_j}(Z_{r_j}(t))
  \text{ for every }t\in[0,T]
 \right]=1,
 \label{eq:root-output-process-binding}
\end{equation}
and the root and named output use the same low and high thresholds. Assume also that
\begin{equation}
 b_{r_j}(u)=\llbracket F_j\rrbracket(u)
 \quad\text{for every }u\in\Uadm
 \text{ and }j\in\{1,\ldots,m\}.
 \label{eq:root-boolean-equivalence}
\end{equation}
If
\begin{equation}
 \max_{u\in\Uadm}\Delay(u)\le\tau
 \qquad\text{and}\qquad
 \max_{u\in\Uadm}\Risk(u)\le1-p,
 \label{eq:global-witness-bounds}
\end{equation}
then $D\preceq_\eta\mathbf F$ for the exact target family used by the local contracts.
\end{corollary}

\begin{proof}
Fix $u\in\Uadm$. Root binding and~\eqref{eq:root-boolean-equivalence} identify failure of the joint output event on $[\tau,T]$ with $\bigcup_{j=1}^{m}\Bad_{r_j}^{b_{r_j}(u)}([\tau,T])$. Since $\tau_{r_j}(u)\leq\Delay(u)\leq\tau$, this union is contained in $\bigcup_{j=1}^{m}\Bad_{r_j}^{b_{r_j}(u)}([\tau_{r_j}(u),T])$. \cref{thm:composition} and~\eqref{eq:global-witness-bounds} bound its probability by $1-p$. Taking the infima in \cref{def:refinement} proves the claim.
\end{proof}

For the computational comparison, \emph{full-cone accounting} uses the set
$\witness_{\mathrm{full}}$ containing every ancestor of the output roots, with
each node included once. For each assignment, it uses the same selected local
contract $\mathcal C_{v,c_v(u)}$ as the sliced calculation, but includes every
node in $\witness_{\mathrm{full}}$. Its delay recurrence takes the maximum over
all parents, and its risk sum replaces $\witness(u)$ by
$\witness_{\mathrm{full}}$. This baseline changes only proof accounting. It does not change the physical
circuit or assert that its biochemical reliability or latency has changed.
To interpret full-cone accounting as a certified probability bound, start-set
coverage and the local-law interface must also hold at its own, potentially
later invocation times. The numerical comparison alone does not establish
those additional obligations.

\begin{remark}
A cube has proper support when it leaves at least one input unfixed.
Sufficient-cube slicing can reduce proof-accounting risk only when a valid
masked-input-robust contract exists and at least one omitted node has positive
risk charge. It reduces accounting delay only when an omitted branch
contributes to a maximizing path in the full-cone recurrence. A controlling
Boolean value alone therefore does not imply a strict numerical improvement.
\end{remark}

The composition theorem reduces the joint path property to a finite family of
local, source, envelope, and interface obligations. The next two sections
show how local stochastic evidence can supply those obligations directly or
through a checked surrogate transfer.

\section{Local certificates for switching jump-diffusions}
\label{sec:local-certificates}

The semantic gate contract~\eqref{eq:gate-contract} can be proved using
nonnegative certificate functions whose stopped values are supermartingales
while the local assumptions hold. A supermartingale is a stochastic process
whose conditional expected future value does not exceed its present value.
Here, that property holds before the relevant stopping time. The same
construction can certify
a stricter, observation-eroded tube for a surrogate model. A separate transfer
argument must then relate that result to the exact target.

Each application declares a certificate-state space $\mathsf S_{v,c}$,
an adapted coordinate process $Z$, a continuous observation map $q_{v,c}$,
a Borel initial set $I_{v,c}$, a closed state success set
$\mathcal Q_{v,c}^{b}$, and a family $\mathfrak Q_{v,c}^{s}$ of complete
local path laws. The complete local path includes $Z$, the inputs, context,
modes, and every resource or hidden coordinate needed by the obligation.
The certificate functions below depend on $Z$ only. Other local coordinates
may enter their generator through predictable coefficients. If a certificate
also depends on such a coordinate, that coordinate must be included in $Z$,
and its dynamics must be included in the generator, with any cross-diffusion
and joint-jump terms. A direct exact-model proof uses
\[
 \mathsf S_{v,c}=\mathsf X_v,\quad Z=Z_v,\quad q_{v,c}=o_v,\quad
 \mathfrak Q_{v,c}^{s}=\mathfrak P_{v,c}^{s},\quad
 \mathcal Q_{v,c}^{b}=\tube_v^b.
\]
A surrogate proof uses its own corresponding objects and an eroded success set
in the common observation space. Transfer is permitted only when exact and
surrogate start-and-assumption events are the pullbacks of the same checked
common-data condition.

Within this section, $B_{v,c}$ and $A_{v,c}$ denote the canonical local events
with the superscript $\mathrm{loc}$ omitted. Define
$\Bad_{v,\mathcal Q}^{b}(I)$ as the event that $Z(t)$ lies outside
$\mathcal Q_{v,c}^{b}$ at some $t\in I$.

\subsection{Operating, safe, core, and assumption sets}

Fix a gate $v$, sufficient cube $c$, requested output $b$, and delay
$d=d_{v,c}$. Write the certificate coordinate as $Z_t=(X_t,M_t)$. Its complete local law
also carries the predictable inputs and context. Let
$O_{v,c}$ be an open operating region with compact closure. All local generator
inequalities for the reach phase are checked there. Let
$I_{v,c}\subset O_{v,c}$ be a Borel initial set whose closure is compactly
contained in $O_{v,c}$. Define the open safe region
\begin{equation}
 D_{v,c}^b
 =O_{v,c}\cap
 \operatorname{int}_{\mathsf S_{v,c}}(\mathcal Q_{v,c}^b),
 \label{eq:open-safe-domain}
\end{equation}
and choose a nonempty compact core $K_{v,c}^b\Subset D_{v,c}^b$. Here
$\Subset$ means compact containment. The reach certificate must drive the
state into the core by the declared delay. The stay certificate must then
prevent exit from the safe region. Because a jump may overshoot a boundary,
the certificate functions must also be defined at every certified reachable
post-jump or post-switch landing state.

For a deterministic start time $s\in[0,T-d]$, define
\begin{align}
 \tau_K^s&=\inf\{t\in[s,T]\mid Z_t\in K_{v,c}^b\},\\
 \tau_O^s&=\inf\{t\in[s,T]\mid Z_t\notin O_{v,c}\}.
 \label{eq:reach-stopping-times}
\end{align}
For a stopping time $\sigma$, define
\begin{equation}
 \tau_D^\sigma(\omega)
 =\inf\{t\in[\sigma(\omega),T]\mid
 Z_t(\omega)\notin D_{v,c}^b\}.
 \label{eq:safe-exit-time}
\end{equation}
Thus $\tau_K^s$ is the first time after $s$ at which the process reaches the
core, $\tau_O^s$ is the first time it leaves the operating region, and
$\tau_D^\sigma$ is the first time after $\sigma$ at which it leaves the safe
region. We use $\inf\varnothing=\infty$ and $\tau_D^\infty=\infty$. Since
$D_{v,c}^b$ uses the interior of the closed success set, reaching the threshold
boundary is conservatively counted as a safe-region exit.

\subsection{Reach certificate}

Write $\gen_t^{\theta,U,H}$ for the local extended generator of the declared process under parameter $\theta$ and predictable input and context paths $U$ and $H$. For the switching jump-diffusion, its drift, diffusion, jump, and mode-switching terms are given in Appendix~\ref{app:generator}.

A reach certificate consists of a nonnegative function
\begin{equation}
 R:[0,T-d]\times[0,d]\times\widetilde O_{v,c}
 \longrightarrow\Real_{\ge0}
 \label{eq:reach-function-domain}
\end{equation}
and $\alpha\in[0,1]$, where $\widetilde O_{v,c}$ contains $\overline O_{v,c}$ and every certified reachable operating-exit landing state. Let $\mathcal L_{v,c}^{O}\subseteq\widetilde O_{v,c}$ denote the certified operating-exit landing set. It contains continuous boundary points, compound-Poisson overshoots, and promoter-mode transitions that leave the operating set. The first argument of $R$ is the global start time $s$, and the second is elapsed time $r$. Uniformly over every admissible mode, parameter, predictable input, and predictable context value while the local assumptions hold, require
\begin{align}
 R(s,0,z)&\le\alpha
 &&s\in[0,T-d],\ z\in I_{v,c},
 \label{eq:reach-init}\\
 R(s,d,z)&\ge1
 &&s\in[0,T-d],\ z\in\overline O_{v,c}\setminus K_{v,c}^b,
 \label{eq:reach-deadline}\\
 R(s,r,z')&\ge1
 &&\substack{s\in[0,T-d],\ r\in[0,d],\ z'\in\mathcal L_{v,c}^{O}},
 \label{eq:reach-boundary}\\
 \bigl(\partial_r+\gen_{s+r}^{\theta,U,H}\bigr)R(s,r,z)&\le0
 &&\substack{s\in[0,T-d],\ r\in[0,d],\\ z\in O_{v,c}\setminus K_{v,c}^b}.
 \label{eq:reach-generator}
\end{align}
The initial inequality bounds the certificate by $\alpha$ on the start set.
The deadline and exit-landing inequalities make it
at least one when the core is not reached in time or the process first leaves
the operating region. The generator inequality makes the stopped certificate
a supermartingale. The exit-landing set includes continuous boundary points,
jump overshoots, and mode-transition states outside the operating region. If
jump marks have unbounded support, the checker must combine a bounded-region
calculation with a rigorous tail bound.

\begin{lemma}[Conditional reach bound]
\label{lem:reach}
Under the clauses of Assumption~\ref{ass:wellposed} applicable to $\mathfrak Q_{v,c}^{s}$, conditions~\eqref{eq:reach-init} through~\eqref{eq:reach-generator} imply, for every $s\in[0,T-d]$ and every $\mathbb P\in\mathfrak Q_{v,c}^{s}$,
\begin{equation}
 \mathbb P\!\left[
 A_{v,c}(s,s+d)\cap
 \bigl(
   \{\tau_K^s>s+d\}
   \cup\{\tau_O^s<\tau_K^s\}
 \bigr)
 \;\middle|\;\mathcal G_s^{\mathbb P}
 \right]
 \le\alpha
 \quad\text{a.s. on }B_{v,c}(s).
 \label{eq:reach-bound}
\end{equation}
In words, conditional on the information available at time $s$, the
probability that the assumptions hold through $s+d$ but the process either
misses the core by that deadline or leaves the operating region first is at
most $\alpha$.
\end{lemma}

\begin{proof}
Let
\begin{equation}
 \eta_s=\tau_K^s\wedge\tau_O^s\wedge\tau_A^s\wedge(s+d).
 \label{eq:reach-stop}
\end{equation}
On $B_{v,c}(s)$, condition~\eqref{eq:reach-init} gives $R(s,0,Z_s)\le\alpha$. The conditional It\^o-Dynkin identity, stopped at $\eta_s$, and~\eqref{eq:reach-generator} give
\begin{equation}
 \mathbb E\!\left[
 R(s,\eta_s-s,Z_{\eta_s})
 \;\middle|\;\mathcal G_s^{\mathbb P}
 \right]
 \le R(s,0,Z_s)\le\alpha
 \quad\text{a.s. on }B_{v,c}(s).
 \label{eq:reach-dynkin-conditional}
\end{equation}
On the event inside~\eqref{eq:reach-bound}, one has $\tau_A^s\ge s+d$ by~\eqref{eq:assumption-stopping-event}. The stopped state is therefore either outside the core at the deadline or in a certified operating-exit landing set. A tie between $\tau_A^s$ and the deadline is covered by the deadline inequality. Conditions~\eqref{eq:reach-deadline} and~\eqref{eq:reach-boundary} make the stopped certificate at least one on that event. Nonnegativity and~\eqref{eq:reach-dynkin-conditional} prove~\eqref{eq:reach-bound}.
\end{proof}

\subsection{Stay certificate}
A stay certificate is a nonnegative function of absolute time,
\begin{equation}
 S:[0,T]\times\widetilde D_{v,c}^b\longrightarrow\Real_{\ge0},
 \label{eq:stay-function-domain}
\end{equation}
together with $\beta\in[0,1]$. Absolute time is used because the stay phase can
begin at a random core-hitting time. The set $\widetilde D_{v,c}^b$ contains
$\overline D_{v,c}^b$ and every certified state at which the process can land
when it exits the safe region, including continuous boundary points, jump
overshoots, and mode-transition states. Let
$\mathcal L_{v,c}^{D,b}\subseteq\widetilde D_{v,c}^b$ denote this certified
safe-exit landing set. Uniformly over every admitted parameter and every
predictable input and context value, require
\begin{align}
 S(t,z)&\le\beta
 &&t\in[0,T],\ z\in K_{v,c}^b,
 \label{eq:stay-init}\\
 S(t,z')&\ge1
 &&t\in[0,T],\ z'\in\mathcal L_{v,c}^{D,b},
 \label{eq:stay-boundary}\\
 (\partial_t+\gen_t^{\theta,U,H})S(t,z)&\le0
 &&t\in[0,T],\ z\in D_{v,c}^b.
 \label{eq:stay-generator}
\end{align}
The first inequality bounds the certificate by $\beta$ whenever the stay phase
starts in the core. The second makes it at least one at every safe-exit landing
state. The generator inequality makes the stopped certificate a supermartingale
until the assumptions fail or the process exits the safe region.

\begin{lemma}[Conditional stay bound]
\label{lem:stay}
Under the clauses of Assumption~\ref{ass:wellposed} applicable to $\mathfrak Q_{v,c}^{s}$, fix an invocation time $s$ and $\mathbb P\in\mathfrak Q_{v,c}^{s}$. Let $\sigma$ be a stopping time taking values in $[s,T]$, and define
\begin{equation}
 G_\sigma=\{Z_\sigma\in K_{v,c}^b\}.
 \label{eq:stay-start-event}
\end{equation}
Under $\mathbb P$,
\begin{equation}
 \mathbb P\!\left[
 A_{v,c}(\sigma,T)\cap\{\tau_D^\sigma\le T\}
 \;\middle|\;\mathcal G_\sigma^{\mathbb P}
 \right]
 \le\beta
 \quad\text{a.s. on }G_\sigma.
 \label{eq:stay-bound}
\end{equation}
\end{lemma}

\begin{proof}
On $G_\sigma$, stop the absolute-time process at
\begin{equation}
 \eta_\sigma=\tau_D^\sigma\wedge\tau_A^\sigma\wedge T.
 \label{eq:stay-stop}
\end{equation}
The conditional It\^o-Dynkin identity and~\eqref{eq:stay-generator} give
\begin{equation}
 \mathbb E\!\left[
 S(\eta_\sigma,Z_{\eta_\sigma})
 \;\middle|\;\mathcal G_\sigma^{\mathbb P}
 \right]
 \le S(\sigma,Z_\sigma)\le\beta
 \quad\text{a.s. on }G_\sigma.
 \label{eq:stay-dynkin-conditional}
\end{equation}
On the event in~\eqref{eq:stay-bound}, \eqref{eq:random-assumption-stopping-event} gives $\tau_A^\sigma\ge T$. The stopped state is therefore a certified safe-exit landing state. A tie at $T$ is covered by the exit-landing inequality. Condition~\eqref{eq:stay-boundary} makes the stopped certificate at least one there. Nonnegativity and~\eqref{eq:stay-dynkin-conditional} prove~\eqref{eq:stay-bound}. The argument uses the uniform predictable-control generator inequality and does not invoke a strong-Markov property for an incomplete local state.
\end{proof}

\begin{theorem}[Local reach-and-stay contract]
\label{thm:local-contract}
Assume the applicable clauses of Assumption~\ref{ass:wellposed}. Assume also
the deterministic and random-start consistency conditions in
\eqref{eq:assumption-prefix}, \eqref{eqassumptionrestriction}, \eqref{eqassumptionstoppedrestriction}, \eqref{eq:assumption-stopping-event}, and \eqref{eq:random-assumption-stopping-event}.
If the reach conditions
\eqref{eq:reach-init}, \eqref{eq:reach-deadline}, \eqref{eq:reach-boundary}, and \eqref{eq:reach-generator}
and the stay conditions
\eqref{eq:stay-init}, \eqref{eq:stay-boundary}, and \eqref{eq:stay-generator} hold, then every
$s\in[0,T-d]$ and $\mathbb P\in\mathfrak Q_{v,c}^{s}$ satisfy
\begin{equation}
 \mathbb P\!\left[
 B_{v,c}(s)\cap A_{v,c}(s,T)\cap
 \Bad_{v,\mathcal Q}^{b}([s+d,T])
 \right]
 \le \min\{1,\alpha+\beta\}.
 \label{eq:local-certified-bound}
\end{equation}
In particular, the direct exact-model choice $\mathfrak Q_{v,c}^{s}=\mathfrak P_{v,c}^{s}$ and $\mathcal Q_{v,c}^{b}=\tube_v^b$ establishes~\eqref{eq:gate-contract} with
\begin{equation}
 \eps_{v,c}=\min\{1,\alpha+\beta\}.
 \label{eq:local-epsilon}
\end{equation}
\end{theorem}

\begin{proof}
Fix $s\in[0,T-d]$ and write $B_s=B_{v,c}(s)$. Define the reach-failure event
\begin{equation}
 \mathcal E_s^{\mathrm{reach}}=
 A_{v,c}(s,s+d)\cap
 \left(
   \{\tau_K^s>s+d\}
   \cup\{\tau_O^s<\tau_K^s\}
 \right).
 \label{eq:local-reach-failure}
\end{equation}
and the successful-hit event
\begin{equation}
 \mathcal H_s=
 \{\tau_K^s\le s+d,\ \tau_K^s\le\tau_O^s\}.
 \label{eq:local-hit-event}
\end{equation}
If $B_s$, $A_{v,c}(s,T)$, and the $\mathcal Q_{v,c}^{b}$ output failure on $[s+d,T]$ all occur, then either $\mathcal E_s^{\mathrm{reach}}$ occurs or the process hits the core on $\mathcal H_s$ and later exits $D_{v,c}^b$ by time $T$. Interval consistency gives $A_{v,c}(\tau_K^s,T)$ on the latter event. Hence
\begin{align}
 &B_s\cap A_{v,c}(s,T)\cap\Bad_{v,\mathcal Q}^b([s+d,T])\nonumber\\
 &\quad\subseteq
 (B_s\cap\mathcal E_s^{\mathrm{reach}})
 \cup
 \bigl(
 \mathcal H_s\cap A_{v,c}(\tau_K^s,T)
 \cap\{\tau_D^{\tau_K^s}\le T\}
 \bigr).
 \label{eq:local-failure-decomposition}
\end{align}
Let $\one_{B_s}$ denote the indicator of $B_s$. Then \cref{lem:reach} and the
tower property give
\begin{equation}
 \mathbb P[B_s\cap\mathcal E_s^{\mathrm{reach}}]
 =\mathbb E\!\left[
 \one_{B_s}\mathbb P[\mathcal E_s^{\mathrm{reach}}\mid\mathcal G_s^{\mathbb P}]
 \right]
 \le\alpha.
 \label{eq:local-reach-probability}
\end{equation}
Set $\sigma_T=\tau_K^s\wedge T$. This is an $[s,T]$-valued stopping time. By the standard comparison properties for stopping times,
$\mathcal H_s\in\mathcal G_{\sigma_T}^{\mathbb P}$. On $\mathcal H_s$, one has
$\sigma_T=\tau_K^s\le T$ and $Z_{\sigma_T}\in K_{v,c}^b$. \cref{lem:stay} applied at $\sigma_T$ and the tower property therefore give
\begin{equation}
 \mathbb P\!\left[
 \mathcal H_s\cap A_{v,c}(\tau_K^s,T)
 \cap\{\tau_D^{\tau_K^s}\le T\}
 \right]
 \le\beta.
 \label{eq:local-stay-probability}
\end{equation}
Applying a union bound to~\eqref{eq:local-failure-decomposition} proves~\eqref{eq:local-certified-bound}. The final statement follows by the direct exact-model choice.
\end{proof}

\subsection{Machine-checkable local obligations}
The procedure that searches for certificate functions is outside the trusted
core. It may use analytic derivations, sum-of-squares optimization, Bernstein
representations, neural synthesis, or interval branch-and-bound. Acceptance
depends only on replay in a fixed exact certificate language.

The proof object represents each candidate by an expression directed acyclic
graph with rational coefficients. It stores rational boxes, a checked cover of
every verification region, and exact identifiers for all model functions. The
checker recomputes interval enclosures with outward rounding, checks the drift,
diffusion, switching, jump, boundary, and nonnegativity inequalities, and
verifies coverage of every reachable continuous, jump, or mode-transition exit
state. For unbounded jump marks it also checks a truncation and rigorous tail
or moment bound. A transcendental function is used only through a separately
verified rational enclosure.

Shared resources must either be explicit coordinates of the local state or be
bounded by a separately certified invariant envelope. In the latter case, all
local inequalities hold uniformly over that envelope and its failure allowance
is counted once globally. Declared regulator orthogonality is not evidence of stochastic independence.
Once replay establishes the local bound, it may be used directly for its
certified model. If that model is only a surrogate, a separate discrepancy
argument is still needed before its charge can enter the target-model sum.

\section{Transfer from fast surrogates to exact target semantics}
\label{sec:transfer}
A certificate proved for a computational surrogate does not by itself imply
\cref{def:refinement}. A separate discrepancy certificate must transfer the
relevant path event to the exact target while preserving the assignment,
initialization, inputs, and context on which the local obligation depends.

\subsection{Margin erosion and synchronous coupling}
Fix a compact interval $I=[a,b]\subseteq[0,T]$. Let $X$ be an exact process
and $\widehat X$ a surrogate process. Their continuous observation maps $o_X$
and $o_{\widehat X}$ take values in the same normed physical observation
space and use the same units. Define
\begin{equation}
 d_I(x,\widehat x)=
 \sup_{t\in I}
 \left\|o_X(x_t)-o_{\widehat X}(\widehat x_t)\right\|.
 \label{eq:path-metric}
\end{equation}
For two paths $y$ and $\widetilde y$ in the common observation-path space,
define
\begin{equation}
 d_I^{\mathrm{obs}}(y,\widetilde y)
 =\sup_{t\in I}\|y(t)-\widetilde y(t)\|.
 \label{eq:observation-path-metric}
\end{equation}
For c\`adl\`ag paths, either supremum equals the supremum over a fixed
countable dense subset of $I$ together with $b$, and is therefore measurable.
For scalar port obligations, the norm is absolute value. A joint vector
transfer is permitted only when the certificate declares a product observation
map and a norm on that product space. Otherwise, the checker transfers the
scalar output obligations separately and combines their failure events through
\cref{thm:composition}.

For a closed tube $Q$ in the common physical observation space and $\rho>0$,
let
\begin{equation}
 Q^{-\rho}=\{y\in Q\mid\overline B(y,\rho)\subseteq Q\},
 \qquad
 \overline B(y,\rho)=\{z\mid\|z-y\|\le\rho\}.
 \label{eq:erosion}
\end{equation}
Thus $Q^{-\rho}$ is the part of $Q$ that remains in $Q$ after every
observation perturbation of norm at most $\rho$. In particular,
$[H,\infty)^{-\rho}=[H+\rho,\infty)$ and
$(-\infty,L]^{-\rho}=(-\infty,L-\rho]$.

For a Borel set $C$ of observation paths, define its erosion with respect to
the synchronous observation discrepancy by
\begin{equation}
 C^{-\rho}
 =\left\{y\in C\;\middle|\;
 \overline B_{d_I^{\mathrm{obs}}}(y,\rho)\subseteq C\right\},
 \qquad
 \overline B_{d_I^{\mathrm{obs}}}(y,\rho)
 =\left\{\widetilde y\mid
 d_I^{\mathrm{obs}}(y,\widetilde y)\le\rho\right\}.
 \label{eq:path-set-erosion}
\end{equation}
The checker admits this construction only when both $C$ and $C^{-\rho}$
have checker-verifiable Borel representations in the canonical path sigma
algebra. Erosion uses the uniform observation discrepancy, not the Skorokhod
time-change metric. The transfer proof requires measurability and the defining
containment of the eroded set. It does not require $C$ to be closed in the
full-horizon $J_1$ topology. A closed-tube condition on a strict subinterval
need not be closed in that topology.

Let $\Omega_{\mathrm{ex}}$ and $\Omega_{\mathrm{surr}}$ be the exact and
surrogate local path spaces, and let $\mathsf E$ be a common standard Borel
space. Measurable maps
$E_{\mathrm{ex}}\colon\Omega_{\mathrm{ex}}\to\mathsf E$ and
$E_{\mathrm{surr}}\colon\Omega_{\mathrm{surr}}\to\mathsf E$ record exactly
the data required to agree in the coupling: the Boolean assignment, parameter,
invocation time, initialization datum, and declared local input and context
paths. A context component whose evolution depends on the exact or surrogate
state cannot be frozen as common exogenous data. It must be included in the
coupled state and discrepancy proof. Transfer is used only when the localized
start-and-assumption event is the pullback of one measurable predicate
$A_0\in\mathcal B(\mathsf E)$ under both maps.

For a transfer invoked at time $s$, let $\Xi_{v,c}$ be the common
initialization-data space and let $\pi_\xi\colon\mathsf E\to\Xi_{v,c}$ be its
measurable coordinate map. Define
\[
 \xi_s^{\mathrm{ex}}=\pi_\xi\circ E_{\mathrm{ex}},
 \qquad
 \xi_s^{\mathrm{surr}}=\pi_\xi\circ E_{\mathrm{surr}}.
\]
Checker-validated maps initialize the two continuations from this common datum
in their respective state spaces and verify that both start predicates are
pullbacks of the same common-data predicate.

Let $\Gamma(P,\widehat P)$ be the set of probability couplings of $P$ and
$\widehat P$. Define the synchronous couplings by
\begin{equation}
 \Gamma_{\mathrm{sync}}(P,\widehat P)
 =\left\{\gamma\in\Gamma(P,\widehat P)\;\middle|\;
 E_{\mathrm{ex}}\circ\operatorname{pr}_{\mathrm{ex}}
 =E_{\mathrm{surr}}\circ\operatorname{pr}_{\mathrm{surr}}
 \quad\gamma\text{-almost surely}\right\}.
 \label{eq:synchronous-couplings}
\end{equation}
In particular, every $\gamma\in\Gamma_{\mathrm{sync}}(P,\widehat P)$ satisfies
\[
 \xi_s^{\mathrm{ex}}\circ\operatorname{pr}_{\mathrm{ex}}
 =\xi_s^{\mathrm{surr}}\circ\operatorname{pr}_{\mathrm{surr}}
 \qquad\gamma\text{-almost surely}.
\]
Thus, a synchronous coupling places target and surrogate paths on one
probability space while forcing their declared assignment, parameters,
initialization data, input paths, and context paths to agree.
The surrogate may be restarted at $s$ from this shared datum. It need not be
the restriction of one global surrogate trajectory.

Let $\kappa$ index one fixed local obligation, with declared exact and
surrogate law families $\mathfrak P_\kappa^{\mathrm{ex}}$ and
$\widehat{\mathfrak P}_\kappa^{\mathrm{surr}}$. For a reusable gate contract
$k=(v,c)$ invoked at time $s$, take $\kappa=(k,s)$ and use
$\mathfrak P_\kappa^{\mathrm{ex}}=\mathfrak P_{v,c}^{s}$ and
$\widehat{\mathfrak P}_\kappa^{\mathrm{surr}}
=\widehat{\mathfrak P}_{v,c}^{s,\mathrm{surr}}$.

\begin{theorem}[Synchronous transfer of a local path obligation]
\label{thm:transfer}
Fix an interval $I$, a Borel exact success set $C$ in the common
observation-path space, and an erosion radius $\rho>0$ for which
$C^{-\rho}$ is Borel. Let
$P\in\mathfrak P_\kappa^{\mathrm{ex}}$ be an applicable exact local law. Suppose that
there exist a surrogate law
$\widehat P\in\widehat{\mathfrak P}_\kappa^{\mathrm{surr}}$ and a
synchronous coupling $\gamma\in\Gamma_{\mathrm{sync}}(P,\widehat P)$.
Define
\begin{equation}
 \mathcal A_{\mathrm{ex}}=E_{\mathrm{ex}}^{-1}(A_0),
 \qquad
 \mathcal A_{\mathrm{surr}}=E_{\mathrm{surr}}^{-1}(A_0).
 \label{eq:common-transfer-event}
\end{equation}
Their pullbacks to the product path space agree up to a $\gamma$-null set. Write $\mathcal A$ for this common event. If
\begin{align}
 \gamma\!\left[\mathcal A\cap
   \{o_{\widehat X}(\widehat X_\cdot)\notin C^{-\rho}\}\right]
 &\le\widehat\eps,
 \label{eq:surrogate-eroded-proof}\\
 \gamma[d_I(X,\widehat X)>\rho]&\le\zeta,
 \label{eq:coupling-bound}
\end{align}
then
\begin{equation}
 P\!\left[\mathcal A_{\mathrm{ex}}\cap
   \{o_X(X_\cdot)\notin C\}\right]
 \le\widehat\eps+\zeta.
 \label{eq:exact-transfer}
\end{equation}
Consequently, a reusable transfer certificate must establish
\begin{equation}
 \forall P\in\mathfrak P_\kappa^{\mathrm{ex}}\;
 \exists\widehat P\in\widehat{\mathfrak P}_\kappa^{\mathrm{surr}}\;
 \exists\gamma\in\Gamma_{\mathrm{sync}}(P,\widehat P)
 \label{eq:transfer-family-quantifiers}
\end{equation}
such that the two bounds hold uniformly with the declared $\widehat\eps$ and
$\zeta$.
\end{theorem}

\begin{proof}
On $\mathcal A$, if the surrogate observation path lies in $C^{-\rho}$ and
$d_I(X,\widehat X)\le\rho$, then the exact observation path lies in $C$ by
\eqref{eq:path-set-erosion}. Exact failure is therefore contained in the union
of the two events bounded in the theorem. The union bound and the fact that
$P$ is the first marginal of $\gamma$ prove~\eqref{eq:exact-transfer}.
\end{proof}

\begin{corollary}[Uniform measurable-path-set transfer]
\label{cor:closed-path-transfer}
Suppose that the family-level quantifiers in
\eqref{eq:transfer-family-quantifiers} hold for every invocation time $s$, with
one Borel success set with a Borel erosion, erosion radius, surrogate charge,
and discrepancy charge declared for each obligation. Then every applicable exact law satisfies
\eqref{eq:exact-transfer} with the same obligation-specific charges.
\end{corollary}

\begin{proof}
Apply \cref{thm:transfer} to the surrogate law and synchronous coupling
supplied for the chosen exact law and invocation time.
\end{proof}

\subsection{Synchronous transport certificates}
Comparing only unconditional exact and surrogate output distributions is
insufficient: it may pair an exact path generated under one input or context
with a surrogate path generated under another. Define
\begin{equation}
 \mathsf W_{1,\mathrm{sync}}^I(P,\widehat P)
 =\inf_{\gamma\in\Gamma_{\mathrm{sync}}(P,\widehat P)}
 \mathbb E_\gamma[d_I(X,\widehat X)],
 \label{eq:path-wasserstein}
\end{equation}
with value $+\infty$ if no synchronous coupling exists. The preferred proof
object exhibits a concrete synchronous coupling $\gamma$ and proves
\begin{equation}
 \mathbb E_\gamma[d_I(X,\widehat X)]\le\chi.
 \label{eq:coupling-first-moment}
\end{equation}
Markov's inequality then gives
\begin{equation}
 \gamma[d_I(X,\widehat X)>\rho]
 \le\min\{1,\chi/\rho\},
 \qquad
 \zeta=\min\{1,\chi/\rho\}.
 \label{eq:wasserstein-penalty}
\end{equation}
If only an infimum transport bound is known and no minimizing coupling has been
certified, then for every $\delta>0$ the checker may instead use an exhibited
coupling with expected discrepancy at most $\chi+\delta$ and set
$\zeta=\min\{1,(\chi+\delta)/\rho\}$. Thus
\begin{equation}
 \eps=\min\{1,\widehat\eps+\zeta\}
 \label{eq:transfer-epsilon}
\end{equation}
is a valid transferred charge only when the corresponding synchronous
coupling has been certified. Agreement of single-time distributions or autocorrelation functions is
diagnostic evidence, not a pathwise synchronous coupling proof. The order of
the supremum and expectation also matters: a bound on
$\sup_{t\in I}\mathbb E\|o_X(X_t)-o_{\widehat X}(\widehat X_t)\|$
does not by itself bound $\mathbb E[d_I(X,\widehat X)]$. In particular,
strong error estimates stated in the former order, such as the reaction-network
approximation bound of~\cite{Ganguly2015}, require an additional pathwise
argument before they can supply \eqref{eq:coupling-first-moment}.

\subsection{Transferred composition}
For each fixed-interval obligation $k$, let $C_k$ be its checked Borel success
set in the common observation-path space over $I_k$. A reusable gate
obligation $k=(v,c)$ instead has
\begin{equation}
 I_{k,s}=[s+d_{v,c},T],
 \qquad
 C_{k,s}=\{y\mid y(t)\in\mathsf T_v^b
 \text{ for every }t\in I_{k,s}\}.
 \label{eq:closed-obligation-set}
\end{equation}
The surrogate proof is checked on $C_k^{-\rho_k}$ or
$C_{k,s}^{-\rho_k}$. The checker verifies the observation maps, physical
units, threshold correspondence, Borel representations, and positivity of
$\rho_k$. A context predicate without a checked Borel erosion requires separate
transfer evidence. Every obligation proved only on
a surrogate receives its own synchronous discrepancy penalty. Set
\begin{equation}
 \eps_k^{\mathrm{exact}}
 =\min\{1,\eps_k^{\mathrm{surr}}+\zeta_k\}.
 \label{eq:transferred-local-risk}
\end{equation}

For every invoked surrogate-backed obligation $(v,c,s)$ and every exact law
$P\in\mathfrak P_{v,c}^{s}$, the transfer evidence must supply a surrogate law
$\widehat P\in\widehat{\mathfrak P}_{v,c}^{s,\mathrm{surr}}$ and a coupling
$\gamma\in\Gamma_{\mathrm{sync}}(P,\widehat P)$ for which the checked
discrepancy bound holds with the same declared radius $\rho_k$ and charge
$\zeta_k$. This requirement is uniform over
$s\in[0,T-d_{v,c}]$.

For a surrogate-backed gate, let $Z_v^{\mathrm{ex}}$ and $\widehat Z_v$ be the
exact and surrogate gate-state coordinates. Let $\widehat o_v$ map into the
same physical output space as $o_v$. Apply \cref{thm:local-contract} to the
surrogate family with
\begin{equation}
 \mathcal Q_{v,c}^{b}
 =\widehat o_v^{-1}\!\left((\mathsf T_v^b)^{-\rho_k}\right).
 \label{eq:surrogate-state-tube}
\end{equation}
Thus erosion is performed in the common observation space before it is pulled
back to the surrogate state space. For each $s$, the checker also supplies
$A_{0,k,s}\in\mathcal B(\mathsf E)$ such that
\begin{equation}
 \begin{aligned}
 E_{\mathrm{ex}}^{-1}(A_{0,k,s})
 &=B_{v,c}^{\mathrm{ex}}(s)\cap A_{v,c}^{\mathrm{ex}}(s,T),\\
 E_{\mathrm{surr}}^{-1}(A_{0,k,s})
 &=B_{v,c}^{\mathrm{surr}}(s)\cap A_{v,c}^{\mathrm{surr}}(s,T).
 \end{aligned}
 \label{eq:transfer-common-assumption-event}
\end{equation}
The local theorem gives the uniform eroded-surrogate bound, and
\cref{thm:transfer} gives the exact gate contract
\eqref{eq:gate-contract} with
$\eps_{v,c}=\eps_k^{\mathrm{exact}}$. If the other hypotheses of
\cref{thm:composition} hold, these exact-model charges enter
\eqref{eq:witness-risk}. The resulting bounds and the checked Boolean miter
yield exact-model refinement through \cref{cor:witness-refinement}.
The remaining task is to ensure that a finite proof object binds every premise
to the same request, design, and target model. The next section specifies that
checker boundary.

\section{Proof object, checker, and conditional soundness}
\label{sec:checker}
This section closes the producer--checker boundary. It states which finite
evidence the ideal checker reconstructs, distinguishes that checker from the
executable evaluation tracks, and gives the assumptions under which acceptance
entails model-level refinement.

\subsection{Certificate and checker}
The symbol $\checker$ denotes the ideal mathematical checker specified in this
section. Its theorem-relevant outcome is $\mathsf{accept}$. Within its finite
rational fragment, the restricted exact checker may return
$\mathsf{accept}$, $\mathsf{reject}$, $\mathsf{inconclusive}$,
$\mathsf{malformed\mbox{-}artifact}$, or
$\mathsf{unsupported\mbox{-}proof\mbox{-}backend}$. The binary64 replay
program, denoted $\checker_{64}$, reports only
$\mathsf{replay\mbox{-}pass}$ or $\mathsf{replay\mbox{-}fail}$. Neither
outcome entails the theorem below. The archived Cello track reports compiler
integration and sequence provenance. Without a checked sequence-to-model
derivation, that evidence cannot be submitted for theorem-level temporal
acceptance.

The proof object is
\begin{equation}
 \cert=(\cert_{\mathrm{seq}},\cert_{\mathrm{bool}},
 \cert_{\mathrm{map}},\cert_{\mathrm{local}},
 \cert_{\mathrm{wit}},\cert_{\mathrm{tr}},\cert_{\mathrm{ctx}},
 \cert_{\mathrm{sem}}).
 \label{eq:proof-object}
\end{equation}
Its components have the following roles:
\begin{itemize}
 \item[$\cert_{\mathrm{seq}}$:] source-byte provenance, canonical sequence and
 SBOL data, and sequence-binding evidence,
 \item[$\cert_{\mathrm{bool}}$:] the derived admissible-domain encoding and a
 DRAT refutation of the checker-reconstructed Boolean miter,
 \item[$\cert_{\mathrm{map}}$:] technology-mapping, typing, parameter, and
 resource witnesses,
 \item[$\cert_{\mathrm{local}}$:] local exact or surrogate certificates,
 \item[$\cert_{\mathrm{wit}}$:] sufficient-cube and joint-witness decision
 diagrams,
 \item[$\cert_{\mathrm{tr}}$:] synchronous discrepancy evidence,
 \item[$\cert_{\mathrm{ctx}}$:] envelope and start-set evidence,
 \item[$\cert_{\mathrm{sem}}$:] the semantic-assumption manifest, local
 restriction maps, event pullbacks, and law-inclusion evidence.
\end{itemize}

The compilation instance $\mathcal I$ and technology library are independent
checker inputs. The proof object cannot replace them. The restricted
implementation does not implement this full language. It uses explicit domain
enumeration, proofs containing only RUP clause additions, exact finite tables,
full cubes, and exact rational arithmetic.

Whenever a resource bound is present, $\cert_{\mathrm{map}}$ proves
$\mathsf{usage}(D)\preceq B_{\max}$, where $\preceq$ denotes componentwise
comparison. The context component establishes both
\eqref{eq:envelope-contract} and every start-set inclusion in
\eqref{eq:start-set-coverage}. The semantic component identifies each complete
local restriction map and supplies an admitted model-interface theorem for the
event pullbacks in \eqref{eq:global-local-event-pullbacks} and the law inclusion
in \eqref{eq:global-local-law-inclusion}. The checker validates these hypotheses
against the reconstructed global model, parameter projection, and local law
family.

The checker reconstructs the source-bound artifacts rather than trusting
compiler-supplied encodings. It validates exact sequence and model binding,
gate truth tables and mapping constraints, every local interval inequality,
every sufficient cube and requested output bit, every transfer margin, and the
exact rational risk and delay maxima. The proof format declares one fixed
variable order and a canonical reduced node representation for binary and
algebraic decision diagrams. The checker reconstructs every diagram's
semantics and compares it with the source-bound recurrences for cube selection,
witness membership, risk, and delay. A producer-supplied diagram hash is
provenance metadata, not semantic evidence.

The ideal checker first verifies
$m\ge1$, $\thetaSet\ne\varnothing$, $\context\ne\varnothing$,
$\Uadm\ne\varnothing$, and that every resolved representation
$\mathcal I_{0,D}(u)$ contains at least one syntactically valid candidate
initial law. This syntactic check does not prove weak-solution existence. The
universal existence requirement in Assumption~\ref{ass:wellposed} is either
discharged for the entire declared family by an admitted backend theorem or
listed explicitly as an external mathematical assumption in
$\cert_{\mathrm{sem}}$. The checker also verifies that any producer-derived
encoding of the admissible domain is equivalent to the verifier-supplied
predicate.

Acceptance requires the maximum joint delay over $\Uadm$ to be at most $\tau$
and the maximum exact-model risk to be at most $1-p$. The checker performs no
stochastic simulation. When a check fails, an interval backend may report a
box on which it could not verify an inequality, and the restricted checker may
report an admissible assignment maximizing the failed risk or delay expression.
These are diagnostic witnesses to failure of the supplied proof. They need not
be counterexamples to the underlying semantic property.

A transfer certificate is not an empirical discrepancy estimate. It
identifies the exact model, surrogate model, localized common-data map, common
initialization maps, observation-space path discrepancy, erosion radius,
and synchronous coupling rule. It also contains a checker-replayable proof
of the claimed tail or first-moment bound. Acceptance requires that the
checker derive the bound from admitted exact arithmetic or verified
outward-rounded inequalities. Simulation output alone is not accepted.

\subsection{Soundness theorem}

\begin{assumption}[Acyclic combinational operation]
\label{ass:acyclic}
The logical dependency graph used by the composition theorem is acyclic, and
the Boolean inputs remain fixed throughout the run. This interface condition
does not permit biochemical feedback or memory to be omitted from the
stochastic semantics. Every feedback interaction, hidden memory variable,
growth-mediated effect, and shared resource that can affect a certified gate
is included in the exact target model. Each such component also appears in the
relevant local restriction and is covered by the local uniform contract or a
certified context envelope.
\end{assumption}

\begin{assumption}[Trusted checker primitives]
\label{ass:checker}
The trusted components implement their stated specifications. They comprise
the exact-byte comparator, canonical source, SBOL and model parsers,
model-derivation and admitted exact-model proof kernels, exact rational
arithmetic, verified outward rounding, admitted transcendental-enclosure
kernels, fixed-order decision-diagram operations, canonical CNF generation,
the DRAT checker, and the proof-format parser. Here \emph{outward rounding} means that each
interval endpoint is rounded away from the enclosed exact value, so the
reported interval remains a sound enclosure.
\end{assumption}

\begin{theorem}[Conditional model-level checker soundness]
\label{thm:end-to-end}
Under Assumptions~\ref{ass:wellposed}, \ref{ass:acyclic}, and~\ref{ass:checker}, if
\begin{equation}
 \checker(\mathcal I,D,G_D,\mathsf{SBOL}_D,\mathcal M_D,\cert)
 =\mathsf{accept},
 \label{eq:checker-accepts}
\end{equation}
then the following statements hold.
\begin{enumerate}
 \item $(w_D,\mathsf{SBOL}_D,G_D,\mathcal M_D)$ satisfies the closed-world binding relation~\eqref{eq:sequence-binding},
 \item the Boolean output vector reconstructed from $G_D$ is equivalent to $\mathbf F$ for every $u\in\Uadm$,
 \item for every $u\in\Uadm$, $\theta\in\thetaSet$, and
 $\mathbb P\in\mathfrak P_D(u,\theta)$,
 \begin{equation}
  \mathbb P\!\left[
   \Good_{\mathbf y}^{\llbracket\mathbf F\rrbracket(u)}([\tau,T])
  \right]\ge p.
  \label{eq:end-to-end-conclusion}
 \end{equation}
\end{enumerate}
Consequently, $D\preceq_\eta\mathbf F$ in the model-relative sense of \cref{def:refinement}.
\end{theorem}

\begin{proof}
The checker first materializes the nucleotide word, canonicalizes the SBOL object, reconstructs the typed regulatory graph, and regenerates the model or replays a complete model-derivation proof. These checks establish item one.

The checker then reconstructs the admissible-input predicate and mapped Boolean output vector directly from the trusted request, accepted graph, and technology truth tables. It constructs the domain-restricted canonical miter and validates the DRAT refutation against those clauses. Unsatisfiability establishes item two and~\eqref{eq:root-boolean-equivalence}.

For each selected source, gate, and context obligation, an accepted
exact-model certificate establishes the corresponding semantic bound directly.
For a gate obligation proved on a switching jump-diffusion surrogate,
\cref{thm:local-contract} establishes the eroded surrogate bound and
\cref{thm:transfer} transfers it to the exact model. A surrogate-backed source
or context obligation likewise requires a Borel success set, a checked Borel
erosion, and the family-level synchronous-coupling quantifiers in
\eqref{eq:transfer-family-quantifiers}. Otherwise it requires a direct
exact-model certificate or a separately stated sound transfer backend. Any
other admitted backend must establish the same semantic obligation before its
charge enters the joint witness.

The accepted cube and joint-witness diagrams reconstruct $\Delay(u)$ and $\Risk(u)$ on $\Uadm$ without trusting compiler-supplied maxima. The checker verifies the two inequalities in~\eqref{eq:global-witness-bounds}. \cref{cor:witness-refinement} therefore gives~\eqref{eq:end-to-end-conclusion}. Taking the infima in \cref{def:refinement} proves the final statement.
\end{proof}

For $\mathbf b\in\Bool^m$ and a compact interval $I\subseteq[0,T]$, define the
output-path success set
\begin{equation}
 \mathcal G_{\mathbf b}(I)=
 \left\{
 y\in\mathbb D_T\!\left(\prod_{j=1}^{m}\mathsf Y_{y_j}\right)
 \;\middle|\;
 y_j(t)\in Q_{y_j}^{b_j}
 \text{ for every }j\in\{1,\ldots,m\}\text{ and }t\in I
 \right\}.
 \label{eq:output-path-success-set}
\end{equation}

\begin{corollary}[Physical-realization guarantee]
\label{cor:physical}
Under Assumptions~\ref{ass:wellposed}, \ref{ass:acyclic},
\ref{ass:checker}, and~\ref{ass:physical-adequacy}, suppose that the ideal
checker accepts. Then, for every $u\in\Uadm$ and every
$\mathbb P_{\mathrm{phys}}\in\mathfrak P_{\mathrm{phys}}(D,u)$,
\begin{equation}
 \mathbb P_{\mathrm{phys}}\!\left[
 \mathcal G_{\llbracket\mathbf F\rrbracket(u)}([\tau,T])
 \right]\geq p.
 \label{eq:physical-conclusion}
\end{equation}
\end{corollary}

\begin{proof}
By Assumption~\ref{ass:physical-adequacy}, there exist
$\theta\in\thetaSet$ and $\mathbb P\in\mathfrak P_D(u,\theta)$ such that
$\Law_{\mathbb P}(\mathbf Y_D)=\mathbb P_{\mathrm{phys}}$. The accepted
model-level theorem gives the required probability under $\mathbb P$, and
equality of the output pushforward laws gives the same probability under
$\mathbb P_{\mathrm{phys}}$.
\end{proof}

\begin{remark}[What the result does not claim]
The theorem is conditional on Assumptions~\ref{ass:wellposed}, \ref{ass:acyclic}, and~\ref{ass:checker}. It does not establish assembly-level model adequacy. The physical corollary does not cover an uncharacterized host, mutation, unmodeled crosstalk, or environmental regime outside the declared law family. An empirical discrepancy estimate is not a mathematical coupling bound.
\end{remark}

\subsection{Proof-checking complexity}
Let $N$ be the total bit length of the explicit checker input and certificate,
including the sequence, canonical graph and model representations, expression
DAGs, rational coefficients, rational boxes and endpoint precisions, decision
diagrams, transfer proofs, and DRAT proof. Numerical precision is measured by
the number of requested output bits.

For the fixed proof language, assume that the checker makes at most
polynomially many calls to its admitted primitives, every checker-generated
intermediate object has size polynomial in $N$, and each parser,
canonicalization procedure, decision-diagram operation, proof kernel,
enclosure primitive, and transfer kernel runs in time polynomial in the bit
lengths of its explicit inputs and outputs. Under these assumptions, checking
a supplied certificate takes time polynomial in $N$.

This is a verification-complexity statement, not a claim that certificates can
be found efficiently. Discovery may be combinatorial, nonconvex, or
exponential, and a valid certificate may itself be exponentially large in the
original synthesis instance.

The theorem states what acceptance would establish, while the complexity
claim concerns checking an already supplied certificate. Neither guarantees
that a producer can find the required biological mapping or stochastic proof.
The following example makes the witness construction concrete before we turn
to the restricted implementation evidence.

\section{Schematic witness construction for a three-input formula}
\label{sec:example}
For $F=(A\lor\neg B)\land C$, a complement-aware sensor library permits the two-gate mapping
\begin{equation}
  q=\Nor(A,\neg B),
  \qquad
  y=\Nor(q,\neg C)=F.
  \label{eq:example-y}
\end{equation}
The first NOR gate computes $q=\neg(A\lor\neg B)$, and the second computes
$y=\neg(q\lor\neg C)$. If complement rails are unavailable, the compiler
inserts explicit inverters. In a molecular implementation, one cassette
produces $R_q$ from a promoter repressed by $A$ and $\neg B$. A second cassette
produces the output and is repressed by $R_q$ and $\neg C$. Exact nucleotide
bases are selected only after a chassis-specific technology file is supplied.

The proof witness depends on the input assignment. Order the root inputs as
$(q,\neg C)$ and the $q$-gate inputs as $(A,\neg B)$. If $C=0$, the root cube
$(\starv,1)$ already establishes $y=0$. The branch containing $A$, $\neg B$,
and $q$ is therefore absent from the witness. If $C=1$ and $q=0$, both root
inputs are low. In that case, either a high value of $A$ or a high value of
$\neg B$ can establish $q=0$. Finally, if $(C,A,B)=(1,0,1)$, both inputs of
the $q$ gate are low, so $q=1$ and the root cube $(1,\starv)$ establishes
$y=0$. A one-input root contract is sound only if it is uniform over every
masked-input trajectory and shared-resource state in its declared envelope.
\Cref{app:example-full} gives the exact witness formulas.

The synthetic replay in \cref{sec:computational-evaluation} exercises a
corresponding synthetic mapping. It is separate from the Cello-derived
sequence and is used only to reconstruct proof-accounting quantities.

\section{Methods}
\label{sec:computational-methods}
The schematic example identifies the mathematical accounting objects. To
interpret the computational results, we must also fix the evaluated stochastic
model and numerical procedure. This section specifies the three reported
evidence tracks, their numerical semantics, and the descriptive summaries
used in Section~\ref{sec:computational-evaluation}.

\subsection{Executable evidence tracks and numerical backend}
\label{sec:evaluation-setup}
The paper distinguishes the ideal mathematical checker in
\cref{thm:end-to-end} from three executable tracks. The restricted exact
checker implements rational finite discrete-time Markov chain (DTMC)
calculations for an enforced depth-one, full-cube fragment. The synthetic
replay backend uses the IEEE~754 64-bit binary floating-point format, called
binary64 below. It reproduces the synthetic accounting studies but cannot
return formal acceptance. The archived Cello track records a version-locked
external invocation. Ordinary paper reproduction does not rerun Cello: it
verifies the immutable archive's pinned inputs and manifest-listed output
sizes and digests, and then independently reconstructs the selected mapping
and sequence from the output netlist, DNA Plotlib tables, UCF records, and
native SBOL~2 export. Passing one track does not transfer a verdict to another
track.

The synthetic replay path uses grammar~\eqref{eq:source-grammar}, lazy
dual-rail NOR compilation, exact-byte technology mapping, SBOL~3 Resource
Description Framework export, a checker-reconstructed Boolean miter, and a
proof containing only RUP clause additions. The archived Cello track imports
Cello's native SBOL~2 output. These are separate tracks. The paper does not
claim that one artifact is a version-converted form of the other.

The numerical backend declares a discrete-time immigration--death Markov chain
with step size $\Delta=1$ min. Let $X_k$ denote the chain at the minute grid.
Let $a_k\ge0$ be the production rate in molecules per minute and
$\delta_k\ge0$ the per-molecule degradation rate in inverse minutes, both
held fixed during step~$k$. Conditional on $X_k=x$, define
\begin{align}
s_k&=e^{-\delta_k\Delta},\\
\Lambda_k&=
\begin{cases}
\dfrac{a_k}{\delta_k}(1-s_k),&\delta_k>0,\\
a_k\Delta,&\delta_k=0.
\end{cases}
\end{align}
Before application of the finite-state cap, the endpoint count is
\begin{equation}
X_{k+1}^{\mathrm{uncap}}=S_k+I_k,
\qquad
S_k\sim\operatorname{Binomial}(x,s_k),
\quad
I_k\sim\operatorname{Poisson}(\Lambda_k),
\label{eq:immigration-death-step}
\end{equation}
where $S_k$ and $I_k$ are conditionally independent. For
$N=T/\Delta\in\mathbb N$, the path interpretation is
\begin{equation}
X(t)=X_k
\qquad
\text{for }t\in[k\Delta,(k+1)\Delta),
\end{equation}
for $k=0,\ldots,N-1$, with $X(N\Delta)=X_N$. The process is
right-continuous and piecewise constant. All evaluated reach-and-stay events
are events of this declared step process. This interpretation neither defines
nor rules out threshold crossings of an underlying continuous-time biochemical
process between grid points.

The finite state space is
\begin{equation}
 \mathsf S=\{0,\ldots,350\}\cup\{\dagger\},
\end{equation}
where $\dagger$ is an absorbing overflow state. One-step endpoint mass above
350 is accumulated in $\dagger$. For finite states, the observation is
$o(k)=k$ molecules. The overflow state has no molecular-count observation and
is declared a failure for both output classes. Index~351 is only its internal
array representation.

Transition matrices and dynamic-programming recurrences are evaluated in
binary64. The backend has no verified directed rounding or
interval enclosure. Its values are therefore numerical replay estimates rather
than rigorous probability bounds. We write $R_{64}(u)$ for its uncapped
assignment-wise risk-charge sum to distinguish it from the formal probability
bound $\Risk(u)$. The global envelope charge is supplied by the request. The
program does not construct a context-envelope certificate or prove start-set
coverage. It evaluates selected parameter-box corners as a diagnostic
heuristic. Because no comparison theorem covers the box interior, these values
do not bound all admitted parameter settings.

All molecule-count sets in this paragraph are subsets of the finite states
$\{0,\ldots,350\}$. The baseline low and high output sets are
$\{0,\ldots,30\}$ and $\{80,\ldots,350\}$, with inner cores
$\{0,\ldots,20\}$ and $\{100,\ldots,350\}$, respectively, over a 300-min
horizon. The five named synthetic circuits use $p=0.98$, equivalently the
diagnostic risk-charge threshold $1-p=0.02$, and $\tau=180$~min. Every Boolean
assignment is enumerated, and full-cone and sufficient-cube modes use identical
local replay inputs.

\subsection{Statistics and reproducibility}
\label{sec:statistics-reproducibility}
The computational evaluation uses deterministic reconstruction rather than
inferential hypothesis testing. The five named circuits enumerate every
admitted Boolean assignment. The formula suite contains 60 archived
deterministic instances: five instances for each combination of four input
dimensions and three requested source sizes. The fixed-input
structural-growth summaries report medians over three archived deterministic
instances at each size. These summaries describe the fixed reported cases.
They do not estimate a population effect over biological circuits or certify
parameter settings between the evaluated corners.

\section{Computational evidence}
\label{sec:computational-evaluation}
The evaluation examines three implementation tracks that exercise different
parts of the architecture. An \emph{archived Cello execution} tests external
compiler invocation and independent reconstruction of its mapping and
sequence. The \emph{restricted exact checker} verifies exact rational
consequences for a deliberately small finite-state fragment. The
\emph{binary64 replay} deterministically recomputes synthetic accounting
examples in 64-bit floating-point arithmetic. A pass in one track is not a verdict in another, and none of these experiments
is needed for the mathematical theorem. The supplied manuscript project
contains the reported figures and tables, but not the software, raw results,
invocation manifests, or archived Cello inputs and outputs described below.
Consequently, the implementation outcomes are retained as reported results,
not as independently reproduced results. Appendix~\ref{app:software-artifact}
identifies the separate evidence required for replay.

The reported experiments address three questions: whether external compiler
artifacts permit independent mapping and sequence reconstruction, whether the
restricted checker accepts its finite model under the declared premises, and
whether slicing changes proof accounting with the circuit and local numerical
inputs held fixed. Methods and
reproducibility procedures appear in
\cref{sec:evaluation-setup,sec:statistics-reproducibility}.

In this section, a \emph{risk-charge sum} is the additive quantity
reconstructed from a selected proof witness. It is a rigorous probability
upper bound only when all corresponding proof obligations have been
discharged. An \emph{accounting delay} is the value of the max-plus witness
recurrence. It is not a measured biochemical response time.
Here a \emph{model-relative verdict} is valid for the declared mathematical
model and its supplied premises. It does not by itself establish that the model
accurately describes an assembled biological system.

For clarity, \emph{reference closure} for an audited technology record means
that every identifier transitively referenced by that record resolves, with
the expected type, to exactly one object in the pinned input set. It is a claim
about the selected record, not necessarily about every record in the library.
\emph{Complete invocation-record closure} is stronger: every input, output,
path, identifier, size, and digest named by the top-level invocation record
must resolve to the archived byte object that it claims to identify, with no
dangling, ambiguous, or path-discrepant link. Thus, digest checks for all
manifest-listed outputs do not by themselves establish complete
invocation-record closure. A \emph{depth-one} mapped instance has no checked
gate-to-gate path: each mapped output is driven by a gate whose parents are
source leaves. A \emph{full cube} fixes every gate-input coordinate to either
zero or one and therefore contains no masked coordinate $\starv$. Finally,
\emph{bp} denotes base pairs, the conventional unit of double-stranded DNA
length. The reported value is the number of nucleotide positions in the
canonical represented sequence.

\begin{table}[t]
\centering
\small
\caption{What each evaluation track establishes and does not establish.}
\label{tab:evidence-levels}
\begin{tabularx}{\linewidth}{@{}>{\raggedright\arraybackslash}p{0.22\linewidth}YY@{}}
\toprule
Evidence track & Established by this track & Not established by this track\\
\midrule
Pinned A1\_AmtR UCF audit &
Exact source bytes, reference closure, metadata, and an 836-bp cassette reconstruction &
A Cello execution, temporal stochastic certification, or physical adequacy\\
Archived external Cello execution &
Successful execution, per-file size and digest checks, selected devices, and a
verifier-reconstructed 1,962-bp design sequence &
Complete invocation-record closure, wet-lab evidence, a calibrated stochastic
model, or temporal acceptance\\
Restricted exact checker &
Exact rational reconstruction and model-relative acceptance for a depth-one,
full-cube finite-state test instance &
Real-sequence binding, general masked-cube composition, or the complete
ideal-checker theorem\\
BCD and rational examples &
Exact source truth values, source-count slicing, and one exact two-state recurrence &
A mapped decoder, sequence, model, or joint decoder reliability bound\\
Synthetic binary64 replay &
Deterministic reconstruction of synthetic mappings and accounting terms &
A rigorous probability bound or formal acceptance\\
Mutation and release checks &
Rejection of the tested semantic mutations and the reported
software-installation workflows &
Exhaustive security assurance, unobserved CI configurations, or biological
validation\\
\bottomrule
\end{tabularx}
\end{table}

\subsection{Static audit of one Cello UCF entry}
\label{sec:cello-ucf-case}
The static audit uses the Eco1C2G2T2 \emph{Escherichia coli} record from
Cello-UCF release v1.0~\cite{CelloUCF2021,Jones2022}. The pinned commit
identifier is
\texttt{b8147b33\allowbreak{}9a6c6189\allowbreak{}964c779c\allowbreak{}967d23bd\allowbreak{}b98dee32}.
The header identifies the host as \emph{E. coli} NEB 10-beta and records M9
medium and an operating temperature of $37\,{}^\circ\mathrm{C}$. The library
contains 18 NOR gates. We selected A1\_AmtR because its record includes a typed
NOR gate, a regulator, a structure, an output promoter, exact part words, and
stored response-model fields. The reported hash prefix is a prefix of a
SHA-256 digest. \Cref{tab:cello-ucf-audit} lists the independently reconstructed
fields.

\begin{table}[t]
\centering
\small
\caption{Verifier-reconstructed fields for the pinned Eco1C2G2T2 A1\_AmtR static audit. Stored response and timing values remain metadata with unresolved units.}
\label{tab:cello-ucf-audit}
\begin{tabularx}{\linewidth}{@{}>{\raggedright\arraybackslash}p{0.31\linewidth}Y@{}}
\toprule
Item & Value \\
\midrule
UCF snapshot & Eco1C2G2T2 \\
Host record & \emph{Escherichia coli} NEB 10-beta \\
Library logic gates & 18 NOR gates \\
Audited gate & A1\_AmtR \\
Cassette parts & BydvJ, A1, AmtR, L3S2P55 \\
Part lengths & 76, 34, 669, and 57 bp \\
Cassette length & 836 bp \\
Stored mixed-case sequence SHA-256 prefix & \path{5c095f5dccfb67d6} \\
Canonical uppercase sequence SHA-256 prefix & \path{d00c7e3bfca89736} \\
Raw UCF response fields (not interpreted here) &
$y_{\max}=3.13661101$, $y_{\min}=0.035221608$,
$K=0.047427884$, and $n=1.656171138$\\
Raw UCF tandem fields (not interpreted here) &
$\alpha=0.274797631$ and $\beta=1.0$\\
Raw UCF timing fields (units unresolved) &
$\tau_{\mathrm{on}}=0.9$ and $\tau_{\mathrm{off}}=2.5$\\
\bottomrule
\end{tabularx}
\end{table}

The audit records these UCF fields as uninterpreted metadata. The pinned files
do not, by themselves, determine the governing temporal equation, the unit of
each parameter, or a calibration from relative promoter units to the
molecule-count observation used by the synthetic stochastic backend.
Consequently, no temporal model or certificate is derived from these values.

\subsection{Archived Cello execution and independent reconstruction}
\label{sec:archived-cello-n1}
The archived case records an external Cello execution for a two-input,
one-output NOR netlist using mutually compatible Eco1C1G1T1 UCF, input-sensor,
and output-device files. The invocation records the pinned compiler, UCF, and
Java-archive versions.

The execution completed successfully in 21.9~s and produced 34 output files.
Cello selected \texttt{S3\_SrpR}, together with
\texttt{LacI\_sensor}, \texttt{TetR\_sensor}, and
\texttt{YFP\_reporter}. The verifier independently parsed the output netlist,
DNA Plotlib tables, UCF records, and SBOL~2 export. It reconstructed the
selected mapping and an ordered 1,962-bp design sequence. Every output listed
in the invocation manifest passed its size and digest check. These per-output
checks do not establish complete invocation-record closure, which is not used
in the semantic reconstruction.

Table~\ref{tab:real-cello-n1} reports the principal evidence. Standard paper
reproduction rehashes and reimports the immutable archive. It does not claim a
new Cello execution.

\begin{table}[t]
\centering
\small
\caption{Archived Cello one-NOR execution and independently reconstructed
design evidence.}
\label{tab:real-cello-n1}
\begin{tabularx}{\linewidth}{@{}>{\raggedright\arraybackslash}p{0.20\linewidth}YY@{}}
\toprule
Item & Reconstructed result & Evidence basis\\
\midrule
Version lock &
Cello~v2 commit
\texttt{1b93575b\allowbreak{}8aa7e40e\allowbreak{}132381e3\allowbreak{}a9b29b1a\allowbreak{}a7f7a8e4}.
Cello-UCF commit
\texttt{cfb4de1c\allowbreak{}f125a6fb\allowbreak{}5af19717\allowbreak{}30844ea5\allowbreak{}d0e91623} &
Invocation record and immutable input archive\\
Java archive &
SHA-256
\texttt{040eec95\allowbreak{}3d4344e5\allowbreak{}fdd75259\allowbreak{}29ce831b\allowbreak{}d3d043e4\allowbreak{}815c805a\allowbreak{}5cb659bf\allowbreak{}895c3fc6} &
Compatibility binary used by the archived invocation\\
Execution &
Return code~0 in 21.9~s. 34 output files totaling 310,232 bytes &
Archived invocation and output manifest\\
Mapped devices &
\texttt{S3\_SrpR}, \texttt{LacI\_sensor}, \texttt{TetR\_sensor}, and
\texttt{YFP\_reporter} &
Output netlist and locked input records agree\\
Design sequence &
1,962~bp. SHA-256
\texttt{2d90b6c6\allowbreak{}282ba354\allowbreak{}6d213d2d\allowbreak{}56f3144f\allowbreak{}07f6ea5e\allowbreak{}39989c69\allowbreak{}1f011805\allowbreak{}0a7e18c0} &
DNA Plotlib table, UCF records, and SBOL~2 export agree\\
Manifest outputs &
All 34 entries pass their recorded size and digest checks &
Per-output integrity. Complete invocation-record closure is not claimed\\
Scope &
Compiler execution, mapping reconstruction, and sequence provenance &
No sequence-bound stochastic model or temporal certificate\\
\bottomrule
\end{tabularx}
\end{table}

This case uses Eco1C1G1T1 and is distinct from both the Eco1C2G2T2 A1\_AmtR
static audit and the BCD source-level example. It shows that the integration
can execute the version-locked Cello compiler and independently reconstruct the
recorded mapping and sequence. Because it contains no sequence-bound
stochastic model or temporal proof, it lies outside the scope of temporal
acceptance. It is a single-case compiler-integration test, not a broad compiler
benchmark or wet-lab validation.

\subsection{Source-level decoder slicing illustration}
A binary-coded decimal (BCD) to seven-segment decoder provides a recognizable multi-output source example~\cite{Shin2020Display}. Let $b_3,b_2,b_1,b_0$ denote the 8, 4, 2, and 1 bits. Its admissible domain is
\begin{equation}
 \Uadm^{\mathrm{BCD}}
 =
 \left\{
 (b_3,b_2,b_1,b_0)\in\Bool^4
 \,\middle|\,
 8b_3+4b_2+2b_1+b_0\le9
 \right\}.
 \label{eq:bcd-domain}
\end{equation}
Each output $a,\ldots,g$ controls one display segment. Value~1 denotes an
illuminated segment. Input words are ordered as $(b_3,b_2,b_1,b_0)$, from the
8's bit to the 1's bit. Assignments representing 10 through 15 are synthesis
don't-cares: they lie outside $\Uadm^{\mathrm{BCD}}$ and are not constrained by
the specification. The following active-high formulas were checked on all ten
admissible digits.
\begin{align}
 a &= b_1 \lor b_3 \lor (b_0\land b_2) \lor (\neg b_0\land\neg b_2),\nonumber\\
 b &= \neg b_2 \lor (b_0\land b_1) \lor (\neg b_0\land\neg b_1),\nonumber\\
 c &= b_0 \lor b_2 \lor \neg b_1,\nonumber\\
 d &= b_3 \lor (b_1\land\neg b_0) \lor (b_1\land\neg b_2) \lor (\neg b_0\land\neg b_2) \lor (b_0\land b_2\land\neg b_1),\nonumber\\
 e &= (b_1\land\neg b_0) \lor (\neg b_0\land\neg b_2),\nonumber\\
 f &= b_3 \lor (b_2\land\neg b_0) \lor (b_2\land\neg b_1) \lor (\neg b_0\land\neg b_1),\nonumber\\
 g &= b_3 \lor (b_1\land\neg b_0) \lor (b_1\land\neg b_2) \lor (b_2\land\neg b_1).
 \label{eq:bcd-seven-segment}
\end{align}

The plotted metric is a source-expression occurrence count. For a literal $\ell$, including a complemented literal, set $N(\ell)=1$. For a displayed $k$-ary connective, set
\begin{equation}
 N\!\left(\operatorname{op}(\varphi_1,\ldots,\varphi_k)\right)
 =1+\sum_{i=1}^{k}N(\varphi_i).
 \label{eq:source-node-count}
\end{equation}
Repeated literals are counted separately and no common-subexpression sharing is used. The sliced count applies the same recurrence after retaining only children selected by a sufficient cube. These counts are not Boolean supports, binary source-grammar sizes, mapped NOR counts, sensor counts, or witness sizes in \cref{thm:composition}.

For each digit and output formula, the reported sliced count is the minimum
over all valid sufficient-cube choices compatible with that assignment. If
several choices attain the same minimum, only their common count enters the
summary, so the plotted value does not depend on a further tie-breaking rule.

\Cref{fig:bcd-source-counts} reports the full count and the mean and maximum
retained counts over digits 0 through 9. The decoder is not mapped to either
Cello case. This source-level example compares formula occurrence counts. It
does not specify a mapped DNA decoder or claim stochastic reliability for such
a decoder.

\begin{figure}[t]
\centering
\includegraphics[width=0.85\linewidth]{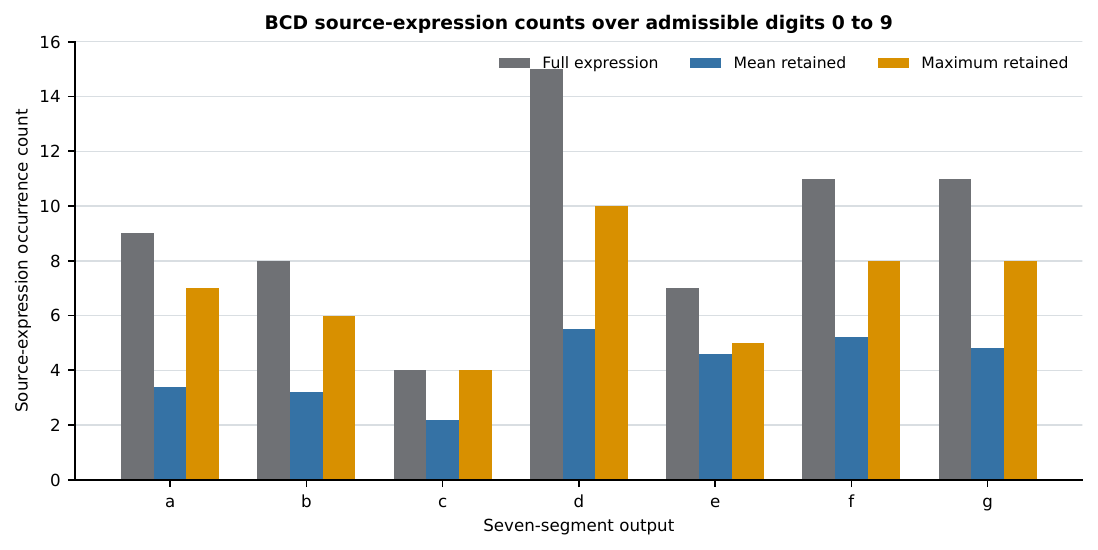}
\caption{Source-expression occurrence counts for the seven binary-coded-decimal output formulas over admissible digits 0 through 9. Bars show the full expression count and the mean and maximum counts retained by sufficient-cube slicing. These are source-level occurrence counts rather than mapped NOR counts, temporal witnesses, or sequence-bound results.}
\label{fig:bcd-source-counts}
\end{figure}

\subsection{Exact-arithmetic local illustration}
\label{sec:exact-rational-pilot}
The two-state grid chain in this illustration is independent of the Cello audit
and the synthetic immigration--death backend. Let one time step represent one
minute, and let state~1 be the desired state. With rows indexed by the current
state and columns by the next state, the transition matrix is
\[
P=\begin{pmatrix}
2/3 & 1/3\\
1/20000 & 19999/20000
\end{pmatrix}.
\]
If $q_k=\mathbb P[X_k=1]$ and $q_0=0$, then
\begin{equation}
q_{k+1}=\frac{1-q_k}{3}+\frac{19999}{20000}q_k.
\label{eq:rational-recurrence}
\end{equation}
The success event is that the chain is in state~1 at minute~20 and never
leaves that state during minutes 21 through 40. Hence its complement has exact
probability
\begin{align}
\varepsilon_{\mathrm{rat}}
&=\mathbb P\!\left[X_{20}=0\ \text{or}\
  X_k=0\ \text{for some }k\in\{21,\ldots,40\}\right]\\
&=1-\frac{20000}{20003}
\left[1-\left(\frac{39997}{60000}\right)^{20}\right]
\left(\frac{19999}{20000}\right)^{20}\nonumber\\
&\approx0.0014492855024366377.
\label{eq:rational-failure}
\end{align}

The reported software study also uses a deliberately restricted exact test instance.
Its source has two named outputs, $A\land B$ and $B\land A$, on the sole
admissible assignment $(A,B)=(1,1)$. Both outputs are bound to the same checked
gate identifier. The checker reconstructs the source formulas, mapped truth
table, domain-restricted miter, canonical conjunctive normal form (CNF), and a
proof containing only reverse-unit-propagation (RUP) clause additions. It then
replays the rational finite-state local obligation, transfer replacement,
source and context charges, sufficient cube, and joint witness.

The surrogate charge is $1/20$. As a verifier-controlled premise of this
restricted test, the finite model also supplies a rational discrepancy
distribution. The checker charges the mass at discrepancy at least the erosion
radius, which is $1/20$. It therefore uses the replacement charge
$1/20+1/20=1/10$. The context charge and each source-leaf charge are $1/100$.
The shared local node is counted once, giving the exact risk-charge sum
$13/100$, below the budget $1/5$. The accounting delay is one time step, below
the deadline of two time steps. This arithmetic is exact within the restricted
premise set, but the supplied distribution is not a checker-derived synchronous
coupling and does not instantiate \cref{thm:transfer}. \Cref{tab:restricted-checker}
records the checked result.

\begin{table}[t]
\centering
\small
\caption{Restricted exact-checker result for the synthetic finite-state test
instance.}
\label{tab:restricted-checker}
\begin{tabularx}{\linewidth}{@{}>{\raggedright\arraybackslash}p{0.25\linewidth}YY@{}}
\toprule
Checked item & Reconstructed result & Scope or evidence basis\\
\midrule
Implemented fragment &
Depth-one, full-cube rational finite-state Markov chain &
Synthetic test instance\\
Domain and outputs &
Sole assignment $(A,B)=(1,1)$. Outputs \texttt{and\_copy} and
\texttt{and\_main} &
One shared mapped root\\
Replacement charge &
$1/20+1/20=1/10$ &
Exact rational arithmetic on supplied premises\\
Risk-charge sum and budget &
$13/100$ and $1/5$ &
Shared node counted once\\
Accounting delay and deadline &
One and two time steps &
Exact integer comparison\\
Checker result &
\texttt{accept} &
Restricted model-relative fragment\\
Sequence binding &
Not provided by this test &
The finite-state model is supplied independently of any sequence\\
\bottomrule
\end{tabularx}
\end{table}

This is an exact-arithmetic verdict of the restricted checker under explicit
synthetic premises. The test instance is not linked to either Cello sequence
and cannot support the complete sequence-bound theorem or a physical-realization
claim.

\subsection{Prototype replay and the accounting effect of slicing}
The binary64 replay program receives the source request, numerical thresholds,
and synthetic technology input separately from the producer-controlled
software bundle. It reconstructs sequence bytes, SBOL structure, the mapped
synthetic model, the Boolean miter, local replay inputs, witness diagrams, and
final accounting maxima. Producer and replay checker are separate entry points
but share parsing and semantic reconstruction routines. The experiment
exercises deterministic recomputation and source-binding checks rather than
agreement between independent implementations.

\Cref{tab:contract-comparison} reports the worst-assignment accounting values
and replay outcomes for the five named synthetic circuits.

\begin{table}[t]
\centering
\caption{Worst assignment-wise binary64 accounting values. The final column gives full-cone and sliced replay outcomes under the numerical thresholds $p=0.98$ and $\tau=180$ min.}
\label{tab:contract-comparison}
\resizebox{\linewidth}{!}{
\begin{tabular}{lrrrrc}
\toprule
Design & \shortstack{Full\\risk-charge sum} & \shortstack{Sliced\\risk-charge sum} & \shortstack{Full accounting\\delay} & \shortstack{Sliced accounting\\delay} & Replay outcome \\
\midrule
worked\_three\_input & 0.00533 & 0.00523 & 75 & 70 & pass / pass \\
mux2                  & 0.00996 & 0.00906 & 95 & 90 & pass / pass \\
alarm4                & 0.00789 & 0.00720 & 75 & 70 & pass / pass \\
selector7             & 0.02404 & 0.01635 & 165 & 160 & fail / pass \\
guarded9              & 0.02267 & 0.01832 & 145 & 135 & fail / pass \\
\bottomrule
\end{tabular}}
\end{table}

Slicing lowers the worst-assignment risk-charge sum by $1.9\%$, $9.0\%$,
$8.8\%$, $32.0\%$, and $19.2\%$ in table order. It also lowers the maximal
accounting delay in every case. The two affected designs move from above to
below the $0.02$ risk-charge threshold. Neither crosses the 180-min
accounting-delay threshold. Their mapped structures and all local replay inputs
are unchanged.
\Cref{fig:contract-frontier} shows these changes relative to the two replay
thresholds. They are reductions in a conservative accounting expression rather
than measured improvements in biochemical reliability or latency.

\begin{figure}[t]
\centering
\includegraphics[width=0.85\linewidth]{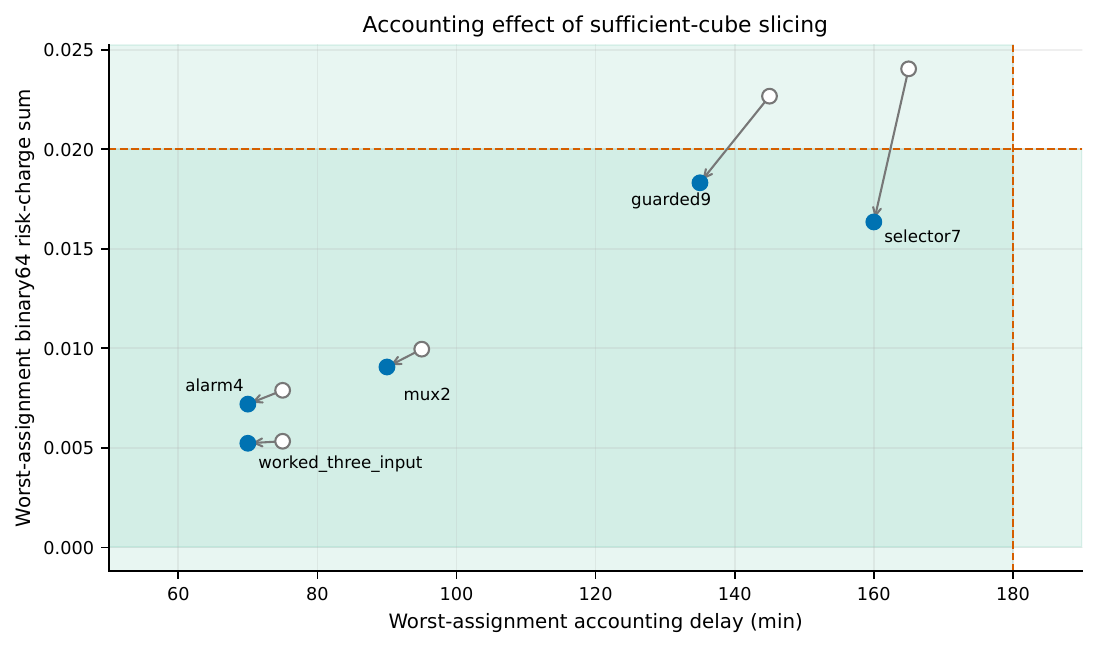}
\caption{Binary64 replay effect of sufficient-cube slicing. Open points use full-cone accounting and filled points use sufficient-cube accounting with identical local replay inputs. The overlap of the shaded regions meets both prototype thresholds. The arrows for \texttt{selector7} and \texttt{guarded9} cross the numerical risk-charge line.}
\label{fig:contract-frontier}
\end{figure}

\Cref{fig:assignment-distributions} resolves the aggregate result assignment by
assignment. Points below the diagonal have lower sufficient-cube accounting
values than full-cone values with identical local replay inputs. The
risk-charge-reduction histogram pools raw assignments across designs, so
designs with more inputs contribute more observations. The plot supports an
accounting claim rather than a probability certificate.

\begin{figure}[t]
\centering
\includegraphics[width=\linewidth]{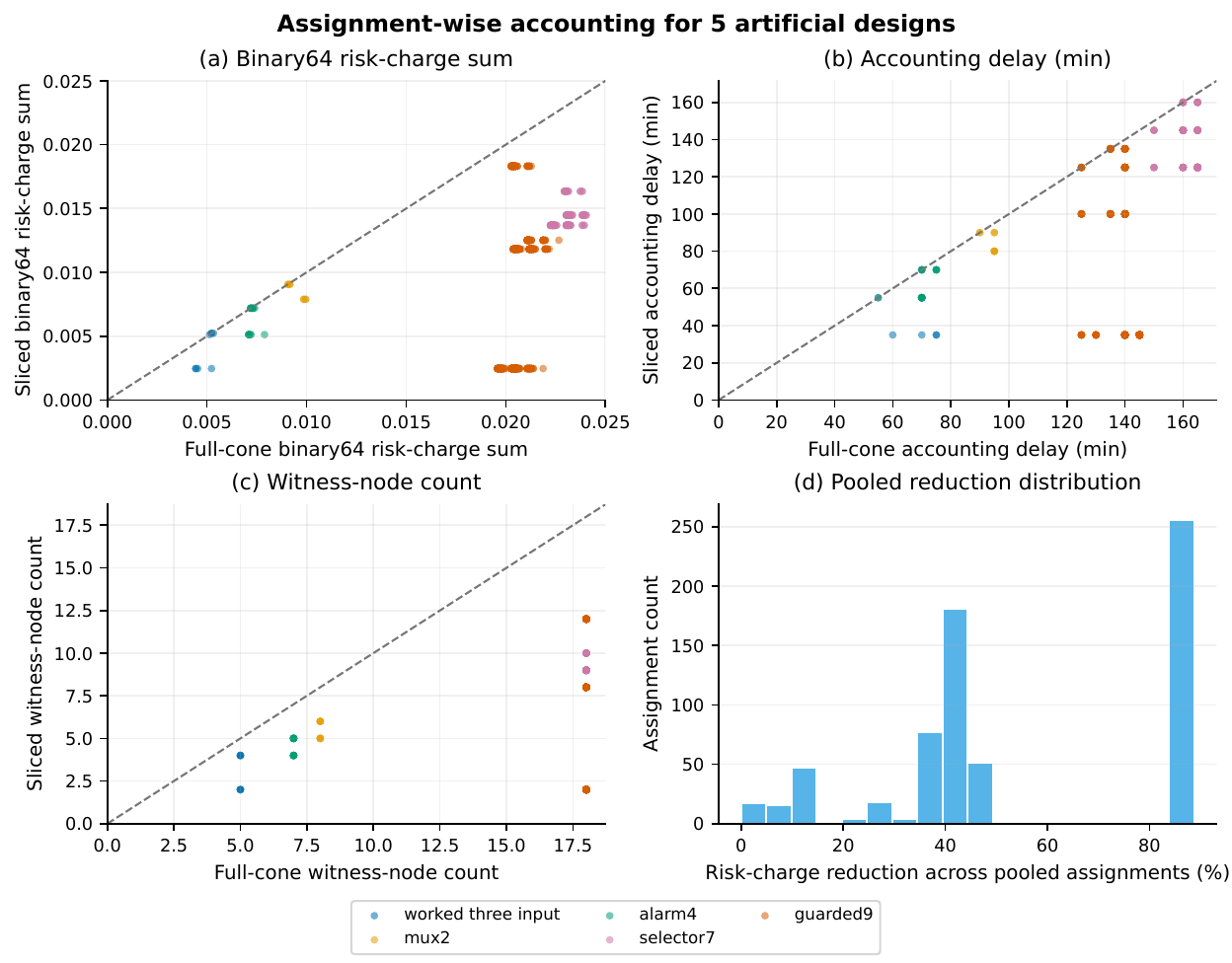}
\caption{Assignment-wise effect of sufficient-cube accounting for five
synthetic designs. Each point represents one admitted assignment. Points below
the diagonal have smaller sufficient-cube risk-charge sum, accounting delay, or
witness size than full-cone accounting with identical local replay inputs. The
histogram pools raw assignments and therefore weights designs by their number
of assignments. The displayed binary64 quantities are replay values rather
than probability certificates.}
\label{fig:assignment-distributions}
\end{figure}

\subsection{Deterministic formula stress suite}
The archived reconstruction contains 60 formula DAGs with 4, 6, 8, or 10
primary inputs and 8, 16, or 32 requested internal nodes. There are five
archived instances for each input-count and requested-size pair. The synthetic
technology and local replay inputs are fixed. Slicing lowers the
worst-assignment risk-charge sum in 55 of 60 formulas and maximal accounting
delay in 54. The median reductions are $11.5\%$ for worst-assignment
risk-charge sum, $40.3\%$ for its assignment mean, $7.8\%$ for maximal
accounting delay, and $49.4\%$ for mean witness size. Five full-cone replay
failures become sliced replay passes. \Cref{fig:random-formula-suite} shows the
four distributions, and Appendix~\ref{app:proof-slicing-distributions} treats
the 60 formulas as a fixed archived corpus and gives an additional
assignment-mean summary. This suite is a deterministic software stress set
rather than a probability model over biological designs.

\begin{figure}[ht]
\centering
\includegraphics[width=0.8\linewidth]{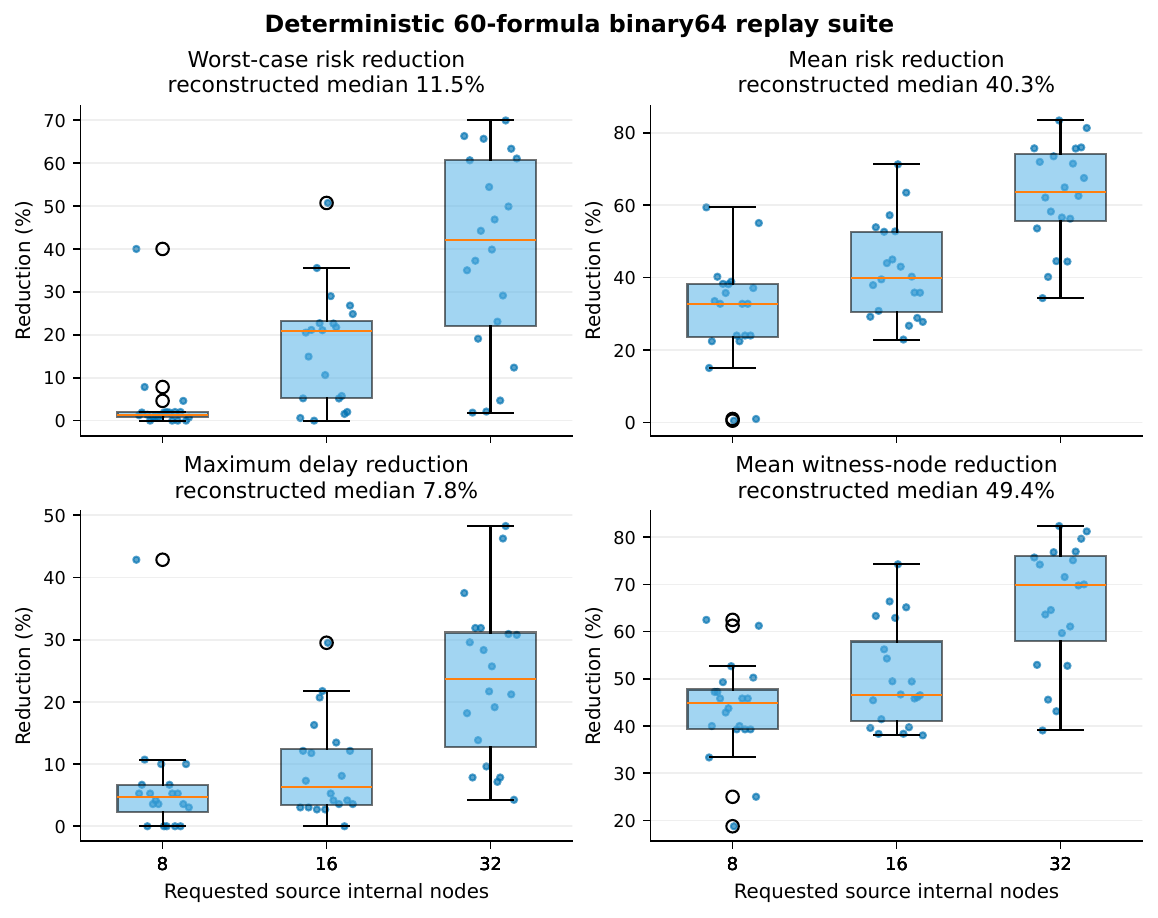}
\caption{Deterministic 60-formula binary64 replay suite. Each panel groups 20
archived formulas by requested internal-node count, pooling four input
dimensions and five archived instances per input-count--size pair. Colored
points denote individual deterministic
instances. Open circles are the boxplots' $1.5$-interquartile-range fliers.
Boxes summarize the distributions. The reductions concern proof accounting and
do not measure biochemical performance or constitute formal probability
bounds.}
\label{fig:random-formula-suite}
\end{figure}

\subsection{Trust-boundary and release checks}

The test suite mutates sequence bytes, UCF references, SBOL graphs, model
identifiers, context evidence, sufficient cubes, witness risk charges,
transfer penalties, archived-execution inputs, and archive paths. The checker
rejects each tested alteration in the validation layer responsible for the
modified field. Transfer tests show
that changing a penalty changes the applicable joint risk charge, removing a
required transfer object is rejected, and a shared transfer-backed node is
charged once. These results establish fail-closed behavior for the tested
cases rather than exhaustive security.

The internal Boolean checker reconstructs the restricted test-instance miter
and validates its proof containing only RUP clause additions. Clean Linux
installations from both the wheel and source distribution reproduced the
restricted exact verdict, the two offline Cello audits, dependency checks, and
request validation.

The supplementary computational section reports assignment-mean proof-slicing
summaries, results for the fixed archived 60-formula corpus, a fixed-input
structural-growth benchmark, and software-bundle mutation tests.
Appendix~\ref{app:software-artifact} specifies the material required for
independent replay. Together, these studies exercise provenance, exact
finite-state arithmetic, and synthetic proof accounting. Their separation
sets the limits of the conclusions discussed next.

\section{Discussion}
\label{sec:limitations}
PCS-Bio is a conditional model-level architecture. Under the stated semantic,
acyclicity, and trusted-checker assumptions, ideal-checker acceptance entails
joint temporal refinement of every named output, uniformly over the admitted
inputs, parameters, initial laws, context policies, and stochastic dependence
structures. The proof combines Boolean equivalence, closed-world
sequence--graph--model binding, local stochastic obligations,
assignment-specific joint witnesses, and, when used, synchronous
surrogate-to-target transfer.

Sufficient-cube slicing changes the proof witness, not the physical circuit.
It can reduce the union-bound risk budget by omitting nodes with positive
charges, and it can reduce accounting delay when an omitted branch lies on a
maximizing full-cone path. Masked molecular inputs remain present and must stay
inside envelopes over which the selected local contract is uniform. The joint
composition counts reconvergent nodes once and does not assume that molecular
events are independent.

The present evidence exercises three separate interfaces. The archived Cello
case establishes one external compiler execution, reconstruction of the
selected mapping, and provenance for a 1,962-bp sequence. The restricted exact
checker establishes exact rational consequences for one depth-one, full-cube
finite-state test instance under verifier-supplied premises. The binary64
studies show how assignment-sensitive slicing changes synthetic accounting
examples, but they do not provide rigorous probability bounds. These outcomes
are not combined into one acceptance verdict.

Several limitations follow directly from that separation. No supplied model
is checker-bound to the Cello-derived sequence. The restricted exact instance
does not exercise general masked-cube composition, continuous-time chemical
dynamics, or the complete proof language. The archived formula-corpus study holds the
synthetic technology fixed, and the size-growth benchmark holds the number of
primary inputs fixed. Neither supports a claim of asymptotic scalability or
physical performance. Any future surrogate-backed acceptance would still have
to supply the jointly measurable synchronous coupling required by
\cref{thm:transfer}.

A physical-realization claim would additionally require a calibrated
sequence-bound stochastic model, native-unit observation maps, uncertainty and
initial-condition coverage, local and interface certificates, and empirical
evidence for Assumption~\ref{ass:physical-adequacy}. Feedback or memory at the
Boolean interface would require extending the present theorem beyond acyclic
clamped-input computation. These are substantive future steps, not conclusions
of the current experiments.

The main contribution is therefore an explicit, independently checkable
contract connecting evidence that is usually reported separately. Once a
technology package supplies the missing sequence-to-model and local evidence,
\cref{thm:end-to-end} yields the model-level guarantee and \cref{cor:physical}
transfers it to the declared physical regime under the adequacy assumption.
The current results support a prototype of complementary interfaces. They do
not establish integrated biological certification, and independent reproduction
also requires the missing software archive.

\section{Conclusion}
\label{sec:conclusion}

PCS-Bio provides a proof-carrying architecture for acyclic, clamped-input
genetic logic. Under its stated assumptions, ideal-checker acceptance implies
a joint temporal-refinement guarantee for all named outputs without assuming
independence among molecular events. The reported evaluation exercises selected
components of this architecture: an archived Cello execution yielded a
verifier-reconstructed 1,962-bp sequence. A separate exact rational checker
accepted one restricted finite-state instance, and sufficient-cube slicing
reduced the worst-assignment binary64 risk-charge sum in 55 of 60 deterministic
formula DAGs and the maximum accounting delay in 54. An integrated
sequence-bound temporal certificate and biological validation remain future
work.


\section*{Competing interests}
\noindent
The authors declare no financial or non-financial competing interests.

\section*{Acknowledgements}
\noindent
The research of Arman Ferdowsi was funded by the Austrian Science Fund (FWF) through grant \href{https://www.fwf.ac.at/en/research-radar/10.55776/ESP1705325}{10.55776/ESP1705325} and the \href{https://ucrisportal.univie.ac.at/en/projects/symbolische-zeitanalyse-asynchroner-schaltungen/}{associated project}. The research of Laura Kovács was funded in whole or in part by the Vienna Science and Technology Fund Grant ForSmart 10.47379/ICT22007, and the SBA Research COMET Center SBA-K1 NGC managed by the Austrian Research Promotion Agency.

\bibliographystyle{unsrtnat}
\bibliography{references}

@article{Nielsen2016,
  author = {Nielsen, Alec A. K. and Der, Bryan S. and Shin, Jonghyeon and Vaidyanathan, Prashant and Paralanov, Vanya and Strychalski, Elizabeth A. and Ross, David and Densmore, Douglas and Voigt, Christopher A.},
  title = {Genetic Circuit Design Automation},
  journal = {Science},
  year = {2016},
  volume = {352},
  number = {6281},
  pages = {aac7341},
  doi = {10.1126/science.aac7341}
}

@article{Jones2022,
  author = {Jones, Timothy S. and Oliveira, Samuel M. D. and Myers, Chris J. and Voigt, Christopher A. and Densmore, Douglas},
  title = {Genetic Circuit Design Automation with {Cello} 2.0},
  journal = {Nature Protocols},
  year = {2022},
  volume = {17},
  pages = {1097--1113},
  doi = {10.1038/s41596-021-00675-2}
}

@article{Kobiela2026,
  author = {Kobiela, Michal and Oyarz{\'u}n, Diego A. and Gutmann, Michael U.},
  title = {Risk-Averse Optimization of Genetic Circuits Under Uncertainty},
  journal = {Cell Systems},
  year = {2026},
  volume = {17},
  number = {1},
  pages = {101476},
  doi = {10.1016/j.cels.2025.101476}
}

@article{McLaughlin2020,
  author = {McLaughlin, James Alastair and Beal, Jacob and M{\i}s{\i}rl{\i}, G{\"o}ksel and Gr{\"u}nberg, Raik and Bartley, Bryan A. and Scott-Brown, James and Vaidyanathan, Prashant and Fontanarrosa, Pedro and Oberortner, Ernst and Wipat, Anil and Gorochowski, Thomas E. and Myers, Christopher J.},
  title = {The {Synthetic Biology Open Language} Version 3: Simplified Data Exchange for Bioengineering},
  journal = {Frontiers in Bioengineering and Biotechnology},
  year = {2020},
  volume = {8},
  pages = {1009},
  doi = {10.3389/fbioe.2020.01009}
}

@article{Poole2022,
  author = {Poole, William and Pandey, Ayush and Shur, Andrey and Tuza, Zoltan A. and Murray, Richard M.},
  title = {{BioCRNpyler}: Compiling Chemical Reaction Networks from Biomolecular Parts in Diverse Contexts},
  journal = {PLOS Computational Biology},
  year = {2022},
  volume = {18},
  number = {4},
  pages = {e1009987},
  doi = {10.1371/journal.pcbi.1009987}
}

@article{Sequeiros2023,
  author = {Sequeiros, Carlos and V{\'a}zquez, Carlos and Banga, Julio R. and Otero-Muras, Irene},
  title = {Automated Design of Synthetic Gene Circuits in the Presence of Molecular Noise},
  journal = {ACS Synthetic Biology},
  year = {2023},
  volume = {12},
  number = {10},
  pages = {2865--2876},
  doi = {10.1021/acssynbio.3c00033}
}

@article{Konur2023,
  author = {Konur, Savas and Gheorghe, Marian and Krasnogor, Natalio},
  title = {Verifiable Biology},
  journal = {Journal of the Royal Society Interface},
  year = {2023},
  volume = {20},
  number = {202},
  pages = {20230019},
  doi = {10.1098/rsif.2023.0019}
}

@article{Konur2021,
  author = {Konur, Savas and Mierla, Laurentiu and Fellermann, Harold and Ladroue, Christophe and Brown, Bradley and Wipat, Anil and Twycross, Jamie and Dun, Boyang Peter and Kalvala, Sara and Gheorghe, Marian and Krasnogor, Natalio},
  title = {Toward Full-Stack In Silico Synthetic Biology: Integrating Model Specification, Simulation, Verification, and Biological Compilation},
  journal = {ACS Synthetic Biology},
  year = {2021},
  volume = {10},
  number = {8},
  pages = {1931--1945},
  doi = {10.1021/acssynbio.1c00143}
}

@article{Buecherl2021,
  author = {Buecherl, Lukas and Roberts, Riley and Fontanarrosa, Pedro and Thomas, Payton J. and Mante, J. and Zhang, Zhen and Myers, Chris J.},
  title = {Stochastic Hazard Analysis of Genetic Circuits in {iBioSim} and {STAMINA}},
  journal = {ACS Synthetic Biology},
  year = {2021},
  volume = {10},
  number = {10},
  pages = {2532--2540},
  doi = {10.1021/acssynbio.1c00159}
}

@inproceedings{Yordanov2011,
  author = {Yordanov, Boyan and Belta, Calin},
  title = {A Formal Verification Approach to the Design of Synthetic Gene Networks},
  booktitle = {Proceedings of the 50th IEEE Conference on Decision and Control and European Control Conference},
  year = {2011},
  pages = {4873--4878},
  doi = {10.1109/CDC.2011.6160969}
}

@inproceedings{Abate2024,
  author = {Abate, Alessandro and Giacobbe, Mirco and Roy, Diptarko},
  title = {Stochastic Omega-Regular Verification and Control with Supermartingales},
  booktitle = {Computer Aided Verification},
  series = {Lecture Notes in Computer Science},
  volume = {14683},
  year = {2024},
  pages = {395--419},
  doi = {10.1007/978-3-031-65633-0_18}
}

@inproceedings{Abate2025,
  author = {Abate, Alessandro and Giacobbe, Mirco and Roy, Diptarko},
  title = {Quantitative Supermartingale Certificates},
  booktitle = {Computer Aided Verification},
  series = {Lecture Notes in Computer Science},
  volume = {15932},
  year = {2025},
  pages = {3--28},
  doi = {10.1007/978-3-031-98679-6_1}
}

@article{Lavaei2022,
  author = {Lavaei, Abolfazl and Soudjani, Sadegh and Abate, Alessandro and Zamani, Majid},
  title = {Automated Verification and Synthesis of Stochastic Hybrid Systems: A Survey},
  journal = {Automatica},
  year = {2022},
  volume = {146},
  pages = {110617},
  doi = {10.1016/j.automatica.2022.110617}
}

@article{Anand2024,
  author = {Anand, Mahathi and Lavaei, Abolfazl and Zamani, Majid},
  title = {Compositional Synthesis of Control Barrier Certificates for Networks of Stochastic Systems Against Omega-Regular Specifications},
  journal = {Nonlinear Analysis: Hybrid Systems},
  year = {2024},
  volume = {51},
  pages = {101427},
  doi = {10.1016/j.nahs.2023.101427}
}

@inproceedings{Anand2022,
  author = {Anand, Mahathi and Murali, Vishnu and Trivedi, Ashutosh and Zamani, Majid},
  title = {k-Inductive Barrier Certificates for Stochastic Systems},
  booktitle = {Proceedings of the 25th ACM International Conference on Hybrid Systems: Computation and Control},
  year = {2022},
  articleno = {12},
  pages = {1--11},
  doi = {10.1145/3501710.3519532}
}

@article{Nejati2022,
  author = {Nejati, Ameneh and Zamani, Majid},
  title = {From Dissipativity Theory to Compositional Construction of Control Barrier Certificates},
  journal = {Leibniz Transactions on Embedded Systems},
  year = {2022},
  volume = {8},
  number = {2},
  pages = {06:1--06:17},
  doi = {10.4230/LITES.8.2.6}
}

@inproceedings{Bartocci2013,
  author = {Bartocci, Ezio and Bortolussi, Luca and Nenzi, Laura},
  title = {A Temporal Logic Approach to Modular Design of Synthetic Biological Circuits},
  booktitle = {Computational Methods in Systems Biology},
  series = {Lecture Notes in Computer Science},
  volume = {8130},
  year = {2013},
  pages = {164--177},
  doi = {10.1007/978-3-642-40708-6_13}
}

@inproceedings{Necula1997,
  author = {Necula, George C.},
  title = {Proof-Carrying Code},
  booktitle = {Proceedings of the 24th ACM SIGPLAN-SIGACT Symposium on Principles of Programming Languages},
  year = {1997},
  pages = {106--119},
  doi = {10.1145/263699.263712}
}

@inproceedings{Wetzler2014,
  author = {Wetzler, Nathan and Heule, Marijn J. H. and Hunt, Jr., Warren A.},
  title = {{DRAT-trim}: Efficient Checking and Trimming Using Expressive Clausal Proofs},
  booktitle = {Theory and Applications of Satisfiability Testing},
  series = {Lecture Notes in Computer Science},
  volume = {8561},
  year = {2014},
  pages = {422--429},
  doi = {10.1007/978-3-319-09284-3_31}
}

@article{Gillespie1977,
  author = {Gillespie, Daniel T.},
  title = {Exact Stochastic Simulation of Coupled Chemical Reactions},
  journal = {The Journal of Physical Chemistry},
  year = {1977},
  volume = {81},
  number = {25},
  pages = {2340--2361},
  doi = {10.1021/j100540a008}
}

@article{Beal2012EndToEnd,
  author  = {Jacob Beal and Ron Weiss and Douglas Densmore and Aaron Adler and Evan Appleton and Jonathan Babb and Swapnil Bhatia and Nathan Davidsohn and Traci Haddock and Joseph Loyall and Richard Schantz and Vasileios Vasilev and Fusun Yaman},
  title   = {An End-to-End Workflow for Engineering of Biological Networks from High-Level Specifications},
  journal = {ACS Synthetic Biology},
  volume  = {1},
  number  = {8},
  pages   = {317--331},
  year    = {2012},
  doi     = {10.1021/sb300030d}
}

@article{Yaman2012Automated,
  author  = {Fusun Yaman and Swapnil Bhatia and Aaron Adler and Douglas Densmore and Jacob Beal},
  title   = {Automated Selection of Synthetic Biology Parts for Genetic Regulatory Networks},
  journal = {ACS Synthetic Biology},
  volume  = {1},
  number  = {8},
  pages   = {332--344},
  year    = {2012},
  doi     = {10.1021/sb300032y}
}

@article{Rodrigo2013AutoBioCAD,
  author  = {Guillermo Rodrigo and Alfonso Jaramillo},
  title   = {{AutoBioCAD}: Full Biodesign Automation of Genetic Circuits},
  journal = {ACS Synthetic Biology},
  volume  = {2},
  number  = {5},
  pages   = {230--236},
  year    = {2013},
  doi     = {10.1021/sb300084h}
}

@article{Roehner2014DAG,
  author  = {Nicholas Roehner and Chris J. Myers},
  title   = {Directed Acyclic Graph-Based Technology Mapping of Genetic Circuit Models},
  journal = {ACS Synthetic Biology},
  volume  = {3},
  number  = {8},
  pages   = {543--555},
  year    = {2014},
  doi     = {10.1021/sb400135t}
}

@article{Madsen2014SMC,
  author  = {Curtis Madsen and Zhen Zhang and Nicholas Roehner and Chris Winstead and Chris J. Myers},
  title   = {Stochastic Model Checking of Genetic Circuits},
  journal = {ACM Journal on Emerging Technologies in Computing Systems},
  volume  = {11},
  number  = {3},
  articleno = {23},
  pages   = {23:1--23:21},
  year    = {2014},
  doi     = {10.1145/2644817}
}

@article{Roehner2015SBOLSBML,
  author  = {Nicholas Roehner and Zhen Zhang and Tramy Nguyen and Chris J. Myers},
  title   = {Generating {Systems Biology Markup Language} Models from the {Synthetic Biology Open Language}},
  journal = {ACS Synthetic Biology},
  volume  = {4},
  number  = {8},
  pages   = {873--879},
  year    = {2015},
  doi     = {10.1021/sb5003289}
}

@article{Vaidyanathan2015Framework,
  author  = {Prashant Vaidyanathan and Bryan S. Der and Swapnil Bhatia and Nicholas Roehner and Ryan Silva and Christopher A. Voigt and Douglas Densmore},
  title   = {A Framework for Genetic Logic Synthesis},
  journal = {Proceedings of the IEEE},
  volume  = {103},
  number  = {11},
  pages   = {2196--2207},
  year    = {2015},
  doi     = {10.1109/JPROC.2015.2443832}
}

@article{Watanabe2019iBioSim,
  author = {Watanabe, Leandro and Nguyen, Tramy and Zhang, Michael and Zundel, Zach and Zhang, Zhen and Madsen, Curtis and Roehner, Nicholas and Myers, Chris J.},
  title = {{iBioSim} 3: A Tool for Model-Based Genetic Circuit Design},
  journal = {ACS Synthetic Biology},
  year = {2019},
  volume = {8},
  number = {7},
  pages = {1560--1563},
  doi = {10.1021/acssynbio.8b00078}
}

@article{Schladt2021,
  author = {Schladt, Tobias and Engelmann, Nicolai and Kubaczka, Erik and Hochberger, Christian and Koeppl, Heinz},
  title = {Automated Design of Robust Genetic Circuits: Structural Variants and Parameter Uncertainty},
  journal = {ACS Synthetic Biology},
  year = {2021},
  volume = {10},
  number = {12},
  pages = {3316--3329},
  doi = {10.1021/acssynbio.1c00193}
}

@article{Engelmann2023,
  author = {Engelmann, Nicolai and Schwarz, Tobias and Kubaczka, Erik and Hochberger, Christian and Koeppl, Heinz},
  title = {Context-Aware Technology Mapping in Genetic Design Automation},
  journal = {ACS Synthetic Biology},
  year = {2023},
  volume = {12},
  number = {2},
  pages = {446--459},
  doi = {10.1021/acssynbio.2c00361}
}

@article{Ganguly2015,
  author = {Ganguly, Arnab and Alt{\i}ntan, Derya and Koeppl, Heinz},
  title = {Jump-Diffusion Approximation of Stochastic Reaction Dynamics: Error Bounds and Algorithms},
  journal = {Multiscale Modeling \& Simulation},
  year = {2015},
  volume = {13},
  number = {4},
  pages = {1390--1419},
  doi = {10.1137/140983471}
}

@article{Kurtz1978,
  author = {Kurtz, Thomas G.},
  title = {Strong Approximation Theorems for Density Dependent {Markov} Chains},
  journal = {Stochastic Processes and their Applications},
  year = {1978},
  volume = {6},
  number = {3},
  pages = {223--240},
  doi = {10.1016/0304-4149(78)90020-0}
}

@article{Pedersen2009GEC,
  author = {Pedersen, Michael and Phillips, Andrew},
  title = {Towards Programming Languages for Genetic Engineering of Living Cells},
  journal = {Journal of the Royal Society Interface},
  year = {2009},
  volume = {6},
  number = {Suppl 4},
  pages = {S437--S450},
  doi = {10.1098/rsif.2008.0516.focus}
}

@article{Czar2009GenoCAD,
  author = {Czar, Michael J. and Cai, Yizhi and Peccoud, Jean},
  title = {Writing {DNA} with {GenoCAD}},
  journal = {Nucleic Acids Research},
  year = {2009},
  volume = {37},
  number = {Web Server issue},
  pages = {W40--W47},
  doi = {10.1093/nar/gkp361}
}

@article{Bilitchenko2011Eugene,
  author = {Bilitchenko, Lesia and Liu, Adam and Cheung, Sherine and Weeding, Emma and Xia, Bing and Leguia, Mariana and Anderson, J. Christopher and Densmore, Douglas},
  title = {{Eugene}: A Domain Specific Language for Specifying and Constraining Synthetic Biological Parts, Devices, and Systems},
  journal = {PLOS ONE},
  year = {2011},
  volume = {6},
  number = {4},
  pages = {e18882},
  doi = {10.1371/journal.pone.0018882}
}

@book{Hachtel1996,
  author = {Hachtel, Gary D. and Somenzi, Fabio},
  title = {Logic Synthesis and Verification Algorithms},
  publisher = {Kluwer Academic Publishers},
  address = {Boston, MA},
  year = {1996},
  doi = {10.1007/b117060}
}

@article{Pandey2022,
  author = {Pandey, Ayush and Incer, Inigo and Sangiovanni-Vincentelli, Alberto and Murray, Richard M.},
  title = {From Specification to Implementation: Assume-Guarantee Contracts for Synthetic Biology},
  journal = {bioRxiv},
  year = {2022},
  doi = {10.1101/2022.04.08.487709}
}

@inproceedings{Incer2024,
  author = {Incer, Inigo and Pandey, Ayush and Nolan, Nicholas and Peterman, Emma L. and Galloway, Kate E. and Sontag, Eduardo D. and Del Vecchio, Domitilla},
  title = {Guaranteeing System-Level Properties in Genetic Circuits Subject to Context Effects},
  booktitle = {63rd IEEE Conference on Decision and Control},
  pages = {5558--5565},
  year = {2024},
  doi = {10.1109/CDC56724.2024.10886081}
}

@article{Incer2025,
  author = {Incer, Inigo and Badithela, Apurva and Graebener, Josefine B. and Mallozzi, Piergiuseppe and Pandey, Ayush and Rouquette, Nicolas and Yu, Sheng-Jung and Benveniste, Albert and Caillaud, Benoit and Murray, Richard M. and Sangiovanni-Vincentelli, Alberto L. and Seshia, Sanjit A.},
  title = {Pacti: Assume-Guarantee Contracts for Efficient Compositional Analysis and Design},
  journal = {ACM Transactions on Cyber-Physical Systems},
  volume = {9},
  number = {1},
  pages = {3:1--3:35},
  year = {2025},
  doi = {10.1145/3704736}
}

@misc{CelloUCF2021,
  author       = {{CIDARLAB}},
  title        = {{Cello-UCF}: user constraint files, input sensor files, output device files, and {JSON} schemas for {Cello}},
  year         = {2021},
  howpublished = {\url{https://github.com/CIDARLAB/Cello-UCF/releases/tag/v1.0}},
  note         = {Release v1.0 at commit b8147b339a6c6189964c779c967d23bdb98dee32, accessed 26 June 2026}
}

@article{Shin2020Display,
  author  = {Shin, Jonghyeon and Zhang, Shuyi and Der, Bryan S. and Nielsen, Alec A. K. and Voigt, Christopher A.},
  title   = {Programming {Escherichia coli} to function as a digital display},
  journal = {Molecular Systems Biology},
  year    = {2020},
  volume  = {16},
  number  = {3},
  pages   = {e9401},
  doi     = {10.15252/msb.20199401}
}

@article{Raj2008,
  author = {Raj, Arjun and van Oudenaarden, Alexander},
  title = {Nature, Nurture, or Chance: Stochastic Gene Expression and Its Consequences},
  journal = {Cell},
  year = {2008},
  volume = {135},
  number = {2},
  pages = {216--226},
  doi = {10.1016/j.cell.2008.09.050}
}

@article{Golding2005,
  author = {Golding, Ido and Paulsson, Johan and Zawilski, Scott M. and Cox, Edward C.},
  title = {Real-Time Kinetics of Gene Activity in Individual Bacteria},
  journal = {Cell},
  year = {2005},
  volume = {123},
  number = {6},
  pages = {1025--1036},
  doi = {10.1016/j.cell.2005.09.031}
}

@book{EthierKurtz1986,
  author = {Ethier, Stewart N. and Kurtz, Thomas G.},
  title = {Markov Processes: Characterization and Convergence},
  publisher = {Wiley},
  year = {1986},
  doi = {10.1002/9780470316658}
}

@book{Applebaum2009,
  author = {Applebaum, David},
  title = {L{\'e}vy Processes and Stochastic Calculus},
  edition = {2},
  publisher = {Cambridge University Press},
  year = {2009},
  doi = {10.1017/CBO9780511809781}
}

\clearpage

\appendix
\noindent
The appendices collect the formal and computational material needed to audit the construction without interrupting the main argument. They specify the proof object and checker, give the complete worked example and sequence-binding obligations, record the extended model and certificate conditions, and report the supplementary experiments and the software-artifact description.

\section{Proof object format and checker procedure}
\label{app:format}
The following abstract grammar summarizes the ideal proof object. It is neither
a serialization standard nor the serialization used by the restricted exact
checker.

\begin{verbatim}
PCSProof = {
  source_hash
  derived_admissible_domain_encoding
  output_roots
  technology_hash
  exact_dna_word
  sbol_range_table
  model_binding_table
  drat_proof
  mapping_assignment
  mapping_constraint_witnesses
  local_certificates
  decision_diagram_variable_order
  cube_decision_diagrams
  joint_witness_risk_delay_diagrams
  exact_rational_terminals
  diagram_semantics_proofs
  transfer_certificates
  context_certificate
  semantic_assumption_manifest
  claimed_max_delay
  claimed_max_risk
}
\end{verbatim}

\begin{verbatim}
LocalCertificate = {
  gate_id
  cube
  requested_output
  parameter_box
  input_envelopes
  context_envelope
  operating_domain_boxes
  core_boxes
  tube_boxes
  reach_function
  stay_function
  derivative_enclosures
  diffusion_enclosures
  switching_enclosures
  jump_integral_enclosures
  reachable_jump_boundary_bounds
  alpha
  beta
  delay
}
\end{verbatim}

\begin{verbatim}
TransferCertificate = {
  obligation_id
  exact_model_id
  surrogate_model_id
  common_data_space
  exact_common_data_map
  surrogate_common_data_map
  common_start_assumption_predicate
  exact_initialization_map
  surrogate_initialization_map
  surrogate_certificate_id
  observation_path_metric
  borel_success_set
  erosion_radius
  synchronous_coupling_construction
  discrepancy_tail_or_first_moment_bound
  uniformity_domain
  bound_proof
}
\end{verbatim}

Exact and surrogate initial states need not be equal: their declared maps
construct both states from common initialization data. Cryptographic digests
establish artifact identity and integrity. They do not establish semantic
correctness.

For an unstructured artifact, the checker compares its exact bytes and trusted
digest. For a structured artifact, it parses the declared canonical
representation and verifies the identifier and digest bound to that
representation. Semantic identifiers need not themselves be digests, but the
checker cross-checks them within the bound objects. The checker rejects a
certificate that references a technology model different from the model linked
by the SBOL design. This prevents a proof for one gate characterization from
being attached to another sequence with the same informal name.

The ideal checker performs the following steps.
\begin{enumerate}
 \item Materialize and compare exact sequence bytes. Validate topology, canonical SBOL structure, graph extraction, and closed-world model derivation.
 \item Verify $m\geq1$, $\thetaSet\ne\varnothing$, $\context\ne\varnothing$, nonemptiness of the verifier-supplied $\Uadm$, and syntactic nonemptiness of $\mathcal I_{0,D}(u)$ for every $u\in\Uadm$. Reconstruct $\mathbf F$, check any derived admissible-domain encoding for equivalence, reconstruct the mapped Boolean output vector, and construct the domain-restricted canonical miter CNF. Validate the DRAT proof against those clauses.
 \item Validate gate truth tables, port typing, mapping constraints, parameter aliases, threshold compatibility, complete local restriction maps, global-to-local law inclusion, and initial and context interfaces.
 \item Recompute every interval enclosure and derive exact-model local risks, including every synchronous transfer penalty.
 \item Check every sufficient cube against the reconstructed local truth table and verify that its selected contract requests the reconstructed Boolean node value.
 \item Reconstruct fixed-order joint-witness membership, exact rational risk, and max-plus delay diagrams on $\Uadm$.
 \item Verify that source-leaf, context, and shared-node charges are included with the declared one-time counting convention.
 \item Recompute the claimed maxima and accept only when they meet $\tau$ and $1-p$.
\end{enumerate}
The ideal trusted numerical kernel comprises exact rational arithmetic, verified outward rounding, and any explicitly admitted transcendental-enclosure kernels. The ideal symbolic kernel comprises canonical parsing and normalization, exact-byte comparison, graph and model reconstruction, fixed-order decision-diagram operations, canonical CNF generation, and a DRAT checker. The restricted exact implementation has a smaller trusted base consisting of strict parsing, canonical serialization, Python integer and rational arithmetic, checker-local reconstruction algorithms, and an internal validator for proofs containing only RUP clause additions.

\section{Alternative exact-model proof back ends}
\label{app:backends}
The main text uses switching jump-diffusion certificates because they connect naturally to fast transcription surrogates. The ideal proof architecture also supports other backends. The reported implementation has two distinct finite discrete-time paths. The restricted checker reconstructs exact rational consequences of verifier-supplied finite transition tables. The separate synthetic replay backend evaluates a related model in binary64. Both finite-state backends operate on supplied benchmark transition tables independently of the archived Cello mapping-and-sequence reconstruction track.

\paragraph{Finite discrete-time Markov chain.}
For a finite transition matrix $P$ and declared grid-time path semantics, reach
and stay failure vectors are obtained by backward dynamic programming. An ideal
proof object can store the finite-state partition, a checked transition
derivation, and rationally enclosed recurrences. The restricted exact checker
instead treats the finite transition matrices as verifier-controlled
mathematical premises and checks their consequences using exact rational
arithmetic. The separate binary64 replay track approximates its own synthetic
backend numerically.

\paragraph{Finite continuous-time abstraction.}
A continuous-time Markov chain (CTMC) is a jump process on a discrete state
space whose transition rates depend on the current state. A finite-state
projection or interval abstraction can produce a probabilistic
transition system. Given a sound finite-state abstraction, Continuous
Stochastic Logic (CSL) model checking can establish reach-and-stay probability
bounds for the abstract system. A certificate must also contain evidence
transferring those bounds to the original process.

\paragraph{Direct CTMC supermartingales.}
For a countable-state chemical reaction network (CRN) with generator
\[
(\gen V)(x)=\sum_r a_r(x)\bigl(V(x+\nu_r)-V(x)\bigr),
\]
where $a_r$ is the propensity and $\nu_r$ is the stoichiometric change vector
of reaction $r$,
the reach and stay lemmas apply directly. Polynomial or rational certificates
can be checked over a truncated region together with a drift or tail
certificate.

\paragraph{Statistical model checking.}
A confidence interval from simulation is useful for screening but is not accepted as a universal proof over a continuous parameter box. It may be used only when the specification itself is statistical and the certificate records the sampling theorem, confidence allocation, random seed commitments, and a finite hypothesis family. That mode gives a different guarantee from \cref{thm:end-to-end}.

\paragraph{Deterministic ODE backend.}
An ordinary differential equation (ODE) specifies continuous state evolution
without stochastic forcing. When noise is intentionally excluded,
deterministic reachability and
invariance can be proved by barrier certificates or by a $\delta$-complete
satisfiability-modulo-theories (SMT) solver when it returns \textsc{unsat} for
the negated obligation. A $\delta$-\textsc{sat} result satisfies only a
numerically relaxed formula and is not a proof of exact validity. The global
sufficient-cube composition is unchanged, with zero stochastic failure for
exact deterministic obligations and explicit uncertainty over parameters and
context.

\section{Full sequence-binding obligations}
\label{app:sequence}

The ideal sequence-binding judgment requires a technology mapping to satisfy Boolean compatibility, molecular port compatibility, regulator orthogonality, fan-out limits, threshold compatibility, host and medium constraints, copy-number constraints, resource and burden bounds, and assembly grammar. A typical cassette is
\[
\iota_v^{-}\cdot P_v(\operatorname{op}_{v,1},\ldots,\operatorname{op}_{v,k})\cdot
 \operatorname{RBS}_v\cdot\operatorname{CDS}_v\cdot T_v\cdot\iota_v^{+}.
\]
Here $P_v$ is the promoter, $\operatorname{op}_{v,j}$ is the operator site for input $j$, $T_v$ is the terminator, and $\iota_v^{-}$ and $\iota_v^{+}$ are the left and right insulation elements. The ideal checker materializes the complete nucleotide word, checks exact bytes, topology, orientation, coordinate coverage, and every grammar-specific overlap constraint. Different SBOL feature objects are allowed to overlap when permitted by SBOL and the selected grammar. Only internally inconsistent or grammar-forbidden overlaps are rejected. The checker compares a canonical SBOL graph rather than a noncanonical serialization, reconstructs every graph interaction, and regenerates the model or replays a complete model-derivation proof. Every relevant model component must have a declared origin, and no relevant SBOL interaction or model reaction may be omitted from the binding relation. The current evidence tracks verify complementary subsets of this judgment and can be combined only through a checked sequence-to-model derivation.

\section{Extended generator and model conditions}
\label{app:generator}
This appendix supplies the generator used by the reach and stay inequalities.
Its event conventions and state coordinates must match the local certificate,
including any resets or jointly changing coordinates.

For a function $V(t,x,\mu)$ in the extended-generator domain, the
finite-activity switching jump-diffusion~\eqref{eq:sjd} has the following
generator under the simple-event and disjoint-event conventions stated there.
It acts on the certificate coordinate $(X,M)$, with predictable inputs and
context entering its coefficients:
\begin{align}
 (\partial_t+&\gen^{\theta,U,H})V(t,x,\mu)
 =\nonumber\\
& \partial_tV(t,x,\mu)
 +\nabla_xV(t,x,\mu)^{\mathsf T}f_\mu(t,x,U,H,\theta)\nonumber\\
 &+\frac12\operatorname{tr}\!\left(
 B_\mu(t,x,U,H,\theta)
 B_\mu(t,x,U,H,\theta)^{\mathsf T}
 \nabla_x^2V(t,x,\mu)
 \right)\nonumber\\
 &+\int_{\mathcal Z}\!
 \left[
 V\bigl(t,x+g_\mu(t,x,U,H,\theta,z),\mu\bigr)-V(t,x,\mu)
 \right]
 \nu_\mu(t,x,U,H,\theta,dz)\nonumber\\
 &+\sum_{\mu'\ne\mu}
 \lambda_{\mu\mu'}(t,x,U,H,\theta)
 \bigl(V(t,x,\mu')-V(t,x,\mu)\bigr).
 \label{eq:extended-generator}
\end{align}
The displayed switching term assumes that a mode transition changes only the
discrete mode. If it also resets the continuous state, replace
$V(t,x,\mu')-V(t,x,\mu)$ by
\[
V\!\left(t,R_{\mu\mu'}(t,x,U,H,\theta),\mu'\right)-V(t,x,\mu),
\]
where $R_{\mu\mu'}$ is the declared reset map. If a model permits a joint
continuous-state jump and mode change, the generator must instead use the
joint increment $V(t,x',\mu')-V(t,x,\mu)$ integrated against its declared
joint transition kernel. Adding separate increments for the two coordinates
would in general give a different generator.
All inequalities are checked uniformly over the declared parameter representation, fixed-input tubes, masked-input envelopes, initial set, and the values admitted by the certified context envelope. That envelope is generated by admissible predictable context policies. The certificate function is defined on the verification domain, its closure, and every certified reachable jump-landing set on which it is evaluated. Unbounded marks require a certified finite-region decomposition and a rigorous tail or moment bound.

\section{Local certificate language and resource treatment}
\label{app:local-format}
The proof object stores the reach and stay functions, rational box decompositions, interval bounds for derivatives, drift, diffusion, switching, jump expectations, and all boundary sets. Polynomial and rational Hill propensities are checked by outward-rounded interval evaluation on positive boxes. For unbounded geometric bursts, the jump expectation is split at a certified cutoff and the remainder is bounded by a probability-generating-function or moment inequality.

Shared resources can be explicit joint state or can be supplied by an
invariant-envelope certificate. In the latter case, every gate inequality is
uniform over the envelope and the envelope failure probability is charged once
globally.
Correlation is permitted only when every shared effect is represented in the
joint model or a certified envelope. Otherwise the composition theorem does
not apply.

On the state space $[0,\infty)$ with its relative topology, consider the
scalar process
\begin{equation}
dX_t=(a-\delta X_{t^-})\,dt+\sigma(X_{t^-})\,dW_t
+\int_{\mathcal J}j\,N(dt,dj),
\end{equation}
where $a\ge0$ is the production rate, $\delta\ge0$ is the per-molecule
degradation rate, $\sigma$ is the diffusion coefficient, and $N$ is a simple, finite-activity raw
integer-valued random measure with predictable compensator
$\nu_{X_{t^-}}(dj)\,dt$. Assume $\mathcal J\subseteq[0,\infty)$ and the existence, nonexplosion,
and state-preservation conditions of Assumption~\ref{ass:wellposed}. At zero,
the drift is $a\ge0$, the diffusion coefficient satisfies $\sigma(0)=0$,
and nonnegative jump marks preserve the state space. These boundary
conditions do not replace the required regularity and existence hypotheses.

Let the low safe domain be $D=[0,L)$ and its core be
$K=[0,L-\Delta_L]$, where $0<\Delta_L\le L$. For $\lambda\ge0$,
$\gamma>0$, and $\beta\in(0,1]$, define
\begin{equation}
S(t,x)=\beta\exp\!\bigl(-\lambda t+\gamma(x-L+\Delta_L)\bigr).
\label{eq:exponential-stay}
\end{equation}
For $t\in[0,T]$ and $x\in K$, one has $S(t,x)\le\beta$. Every continuous
exit or nonnegative jump overshoot from $D$ lands at $x\ge L$. Consequently,
\begin{equation}
\beta\exp(-\lambda T+\gamma\Delta_L)\ge1
\label{eq:exponential-boundary}
\end{equation}
ensures $S(t,x)\ge1$ at every possible exit landing. Dividing the generator
inequality by $S>0$ yields
\begin{equation}
-\lambda+\gamma(a-\delta x)
+\frac12\gamma^2\sigma^2(x)
+\int_{\mathcal J}(e^{\gamma j}-1)\nu_x(dj)
\le0
\quad\text{for }(t,x)\in[0,T]\times D.
\label{eq:exponential-generator}
\end{equation}
The jump moment-generating function must be finite at $\gamma$. If the
certificate depends on promoter mode, include the switching term from
\eqref{eq:extended-generator}. It vanishes for a mode-independent $S$ when mode changes leave $X$ unchanged.

\section{Schematic worked example}
\label{app:example-full}

If the characterized sensor cassettes are supplied elsewhere in the same cell,
set $w_{\mathrm{in}}=\varepsilon$, the empty word. If they are encoded by this
construct, let $w_{\mathrm{in}}$ denote their declared concatenated sequence.
Let $\mathsf{backbone}_{\ell}$ and $\mathsf{backbone}_{r}$ denote the declared
left and right backbone segments under the selected canonical linearization. A
symbolic sequence is
\begin{align}
 w_D={}&\mathsf{backbone}_{\ell}\cdot w_{\mathrm{in}}\cdot
 \iota_q^{-}\cdot P_q(\operatorname{op}_A,\operatorname{op}_{\neg B})\cdot
 \operatorname{RBS}_q\cdot\operatorname{CDS}(R_q)\cdot T_q\cdot\iota_q^{+}\nonumber\\
 &\cdot\iota_y^{-}\cdot P_y(\operatorname{op}_{R_q},\operatorname{op}_{\neg C})\cdot
 \operatorname{RBS}_y\cdot\operatorname{CDS}(Y)\cdot T_y\cdot\iota_y^{+}
 \cdot\mathsf{backbone}_{r}.
\end{align}
All symbols expand to exact library bytes after technology selection. When
supplied by a separate construct, the sensor producers and their dynamics
remain covered by the declared source or context contracts. The Boolean values
are listed once in \cref{tab:example-truth}.

\begin{table}[t]
\centering
\caption{Complete Boolean evaluation of the worked formula.}
\label{tab:example-truth}
\begin{tabular}{ccccccc}
\toprule
$A$ & $B$ & $C$ & $\neg B$ & $q=\neg(A\lor\neg B)$ & $\neg C$ & $y$\\
\midrule
0&0&0&1&0&1&0\\
0&0&1&1&0&0&1\\
0&1&0&0&1&1&0\\
0&1&1&0&1&0&0\\
1&0&0&1&0&1&0\\
1&0&1&1&0&0&1\\
1&1&0&0&0&1&0\\
1&1&1&0&0&0&1\\
\bottomrule
\end{tabular}
\end{table}
\FloatBarrier

For each row, the root NOR selects a one-input high cube when its output is zero
and the two-input low cube when its output is one. When the root selects $q$,
the same rule is applied recursively. The assignment-mean summaries in
Appendix~\ref{app:proof-slicing-distributions} use the uniform distribution on
$\Uadm$. The worst-case results use no averaging distribution.

Order the root inputs as $(q,\neg C)$ and the $q$-gate inputs as $(A,\neg B)$. The assignment classes have the following exact witnesses.

\paragraph{$C=0$.}
Choose the root cube $(\starv,1)$. Then
\begin{align}
 \witness(u)&=\{\neg C,y\},\\
 \Risk(u)&=\epsenv+\eps_{\neg C,1}+\eps_{y,(\starv,1)},\\
 \Delay(u)&=\tau_{\neg C,1}+d_{y,(\starv,1)}.
 \label{eq:example-c0}
\end{align}

\paragraph{$C=1$ and $q=0$.}
The root uses the two-low cube $(0,0)$. Let $j\in\{A,\neg B\}$ be the high controlling input selected by the certificate, with corresponding $q$-gate cube $c_j\in\{(1,\starv),(\starv,1)\}$. Then
\begin{align}
 \witness(u)&=\{\neg C,j,q,y\},\\
 \Risk(u)&=\epsenv+\eps_{\neg C,0}+\eps_{j,1}
 +\eps_{q,c_j}+\eps_{y,(0,0)},\\
 \Delay(u)&=d_{y,(0,0)}+
 \max\{\tau_{\neg C,0},\ d_{q,c_j}+\tau_{j,1}\}.
 \label{eq:example-q0}
\end{align}
If both $A$ and $\neg B$ are high, the certificate records the chosen controller explicitly.

\paragraph{$C=1,A=0,B=1$.}
Here $(A,\neg B)=(0,0)$, so $q=1$ and the root uses $(1,\starv)$. Thus
\begin{align}
 \witness(u)&=\{A,\neg B,q,y\},\\
 \Risk(u)&=\epsenv+\eps_{A,0}+\eps_{\neg B,0}
 +\eps_{q,(0,0)}+\eps_{y,(1,\starv)},\\
 \Delay(u)&=d_{y,(1,\starv)}+d_{q,(0,0)}
 +\max\{\tau_{A,0},\tau_{\neg B,0}\}.
 \label{eq:example-q1}
\end{align}

For illustration, if $\eps_{\neg C,1}=0.004$ and
$\eps_{y,(\starv,1)}=0.007$, the $C=0$ witness risk is
$\epsenv+0.011$. It equals $0.011$ only when $\epsenv=0$.
These values illustrate aggregation and are not experimental claims. This example is schematic because it does not include a concrete technology file, interval certificate, synchronous coupling, or checker transcript.

\section{Computational evaluation details}
\label{app:computational-details}
The following summaries supplement the main accounting comparison with
assignment means, structural growth at fixed input dimension, and the
reported semantic-mutation outcomes. They retain the evidence limitations
stated in Section~\ref{sec:computational-evaluation}.

\FloatBarrier
\subsection{Additional proof-slicing summaries}
\label{app:proof-slicing-distributions}
\paragraph{Assignment-mean savings for the named circuits.}
For completeness, \Cref{fig:sufficient-cube-savings} gives an
assignment-mean summary for the same five named circuits. Let
$R_{\mathrm{s}}(u)$ and $R_{\mathrm{f}}(u)$ denote the sufficient-cube and
full-cone risk-charge sums for assignment $u$. Define
\[
\overline R_{\mathrm{s}}
=\frac{1}{|\Uadm|}\sum_{u\in\Uadm}R_{\mathrm{s}}(u),
\qquad
\overline R_{\mathrm{f}}
=\frac{1}{|\Uadm|}\sum_{u\in\Uadm}R_{\mathrm{f}}(u).
\]
The reported reduction in assignment-mean risk charge is
\[
100\left(1-\frac{\overline R_{\mathrm{s}}}
                    {\overline R_{\mathrm{f}}}\right),
\]
where $\overline R_{\mathrm{f}}>0$ for every reported circuit. This is the
fractional reduction between two assignment means, not the mean of pointwise
percentage reductions. The witness-size metric is defined analogously. For
these synthetic circuits $\Uadm=\Bool^n$, so the means use the uniform
distribution over admitted assignments. Across the five circuits, the
mean risk-charge reduction ranges from $10.9\%$ to $61.1\%$, while the mean witness-size reduction ranges from $31.2\%$ to $69.5\%$. These quantities
describe average replay accounting. They are not used for the replay outcome, which remains
determined by the worst-assignment risk-charge sum and maximal accounting delay in the main comparison.

\begin{figure}[t]
\centering
\includegraphics[width=0.98\linewidth]
{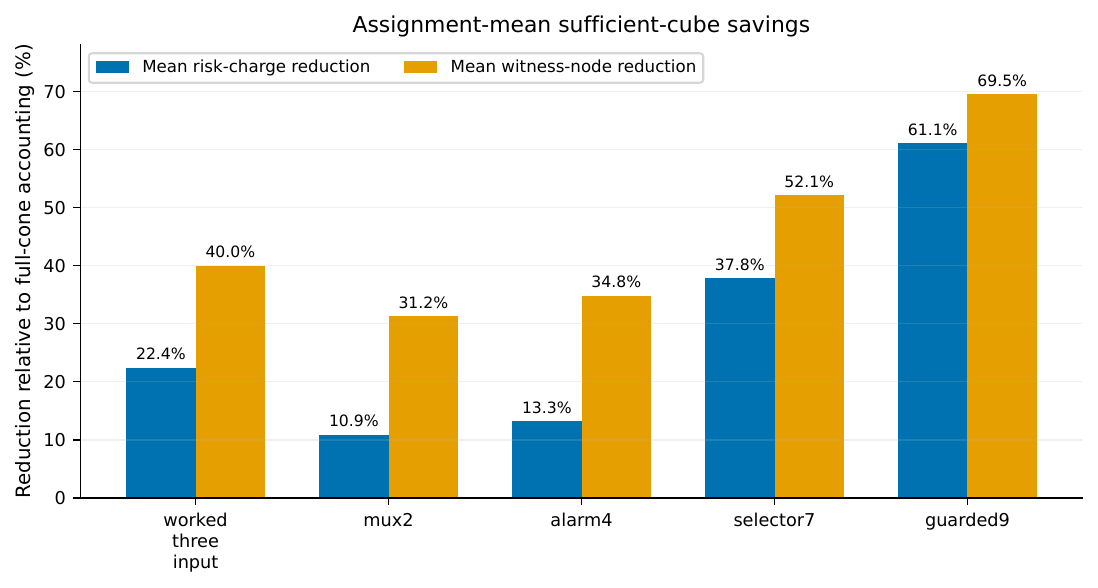}
\caption{Assignment-mean savings from sufficient-cube slicing relative to
full-cone accounting for the five named circuits. Risk-charge reduction is computed
from the sums over all admitted assignments, and witness-size reduction is
defined analogously. Reconvergent nodes are counted once, and the global
envelope charge follows the same counting convention under both methods.
These assignment-mean quantities complement the main worst-assignment replay values.}
\label{fig:sufficient-cube-savings}
\end{figure}

The archived formula suite comprises 60 deterministic expression DAGs: five
archived instances for every pair of input count in $\{4,6,8,10\}$ and
requested internal-node count in $\{8,16,32\}$. Every archived formula is
compiled with the same lazy dual-rail NOR procedure and analyzed over all
assignments. The results are descriptive of this fixed corpus. They are not
estimates for a random formula-generating distribution.

\FloatBarrier
\subsection{Fixed-input structural growth benchmark}
To study growth in circuit representation while keeping exhaustive input
enumeration constant, we fixed the Boolean input dimension at eight and
increased the number of balanced source clauses. Three deterministic archived
instances were used at each size, producing mapped circuits with 13--160 NOR
gates. The replay thresholds were $p=0.90$ and $\tau=300$~min. All 21 bundles
returned \texttt{replay-pass}. The archived proof producer enumerates all
$2^8=256$ assignments. The CNF, sequence, local mapping table, and witness
diagrams grow with the circuit. Table~\ref{tab:fixed8-scalability} and
Figure~\ref{fig:scalability} report medians over the three instances.

\begin{table}[t]
\centering
\small
\caption{Representative fixed-eight-input structural-growth points from the
archived run. Each row reports medians over three deterministic instances.}
\label{tab:fixed8-scalability}
\begin{tabular}{@{}rrrrrrr@{}}
\toprule
\shortstack{Mapped\\NOR gates} &
Depth &
\shortstack{Sequence\\(bp)} &
\shortstack{Bundle\\(KiB)} &
\shortstack{Build\\(s)} &
\shortstack{Check\\(s)} &
\shortstack{Risk-diagram\\nodes}\\
\midrule
14  & 9 & $10{,}724$ & $583.7$     & $0.392$ & $0.275$ & 44\\
60  & 7 & $36{,}152$ & $1{,}665.5$ & $0.632$ & $0.606$ & 65\\
160 & 9 & $87{,}728$ & $4{,}010.2$ & $1.519$ & $1.392$ & 265\\
\bottomrule
\end{tabular}

\medskip
\begin{minipage}{0.96\linewidth}
\footnotesize
Bundle size is reported in kibibytes (KiB), where one KiB is 1,024 bytes.
Build and checking times are reported in seconds. The risk-diagram-node column
reports the number of decision-diagram nodes. Timings describe the archived
platform and are not intended for cross-system performance comparisons.
\end{minipage}
\end{table}

\begin{figure}[t]
\centering
\includegraphics[width=0.98\linewidth]{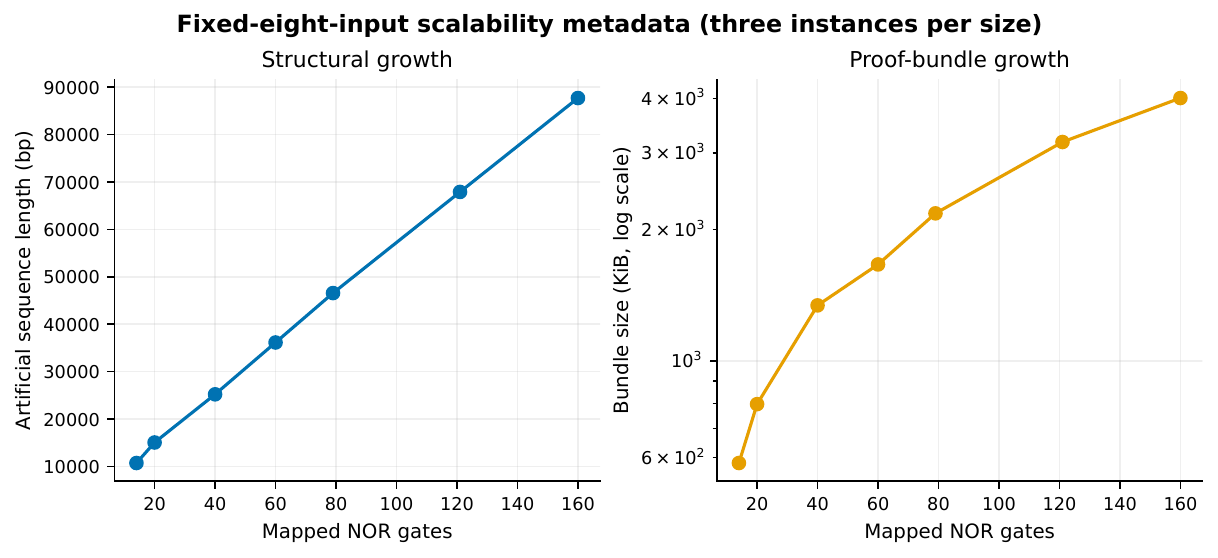}
\caption{Synthetic sequence length and proof-bundle size versus mapped NOR count at fixed eight-input dimension. Points report archived median summaries over three deterministic instances per size. Runtime values are reported in the accompanying table. The producer and replay checker use separate entry points but share semantic reconstruction routines.}
\label{fig:scalability}
\end{figure}

At 160 NOR gates, the median sequence length across the three instances is
87,728~bp and the median bundle size is approximately $3.92$ mebibytes (MiB).
The largest
numerical risk-charge sum is $0.0836$, and the largest accounting delay is
225~min. Both meet the prototype replay thresholds. Build and checking times
are archived descriptive measurements read from stored summaries. The current
reproduction metadata describe the plotting environment rather than a timing
rerun. These data support neither cross-system performance comparisons nor
claims about asymptotic scalability or growth in the number of primary inputs.

\FloatBarrier
\subsection{Software-bundle mutation test}
We recomputed the outer bundle hash after each mutation so that the replay failure could
not be attributed merely to a stale checksum. The program returned replay-pass for the unmodified
bundle and replay-fail for all eight semantically altered variants. The variants contain one changed DNA
byte, a changed feature range, a changed model identifier, a truncated proof
containing RUP clause additions,
an understated local risk charge, a changed sufficient cube, an understated
global risk-charge sum,
and an understated numerical discrepancy charge. \Cref{tab:checker-stress}
reports the resulting outcomes.

\begin{table}[t]
\centering
\caption{Software-bundle mutation test. Each semantic mutation is evaluated
after recomputing the outer bundle hash.}
\label{tab:checker-stress}
\begin{tabular}{lr}
\toprule
Case & Replay outcome \\
\midrule
Unmodified bundle & replay-pass \\
One DNA byte changed & replay-fail \\
SBOL range table changed & replay-fail \\
Model identifier changed & replay-fail \\
Final RUP empty-clause addition removed & replay-fail \\
Local risk charge understated & replay-fail \\
Sufficient cube changed & replay-fail \\
Global risk-charge sum understated & replay-fail \\
Numerical discrepancy charge understated & replay-fail \\
\bottomrule
\end{tabular}
\end{table}
\FloatBarrier
\section{Artifact and replay interface}
\label{app:software-artifact}
This section records the reported software interface and the strongest claim
associated with each command family. The separate software archive is absent
from the supplied manuscript project. The commands and expected outcomes below
therefore document the reproduction procedure rather than a replay performed
from this project. Reproduction requires the installable Python package,
versioned schemas, verifier-controlled examples, archived Cello inputs and
outputs, restricted finite-state test instance, synthetic replay cases, tests,
experiment manifests, machine-readable raw results, and release reports.
The manuscript's figures and tables alone do not establish their provenance.

On Linux, a fresh editable source installation uses
\begin{lstlisting}
python3 -m venv .venv
source .venv/bin/activate
python -m pip install -r requirements-dev-lock.txt
python -m pip install --no-deps -e .
python -m pip check
pcs-bio doctor
\end{lstlisting}
The reported software archive also includes PowerShell instructions and a Windows
continuous-integration matrix. The clean installation and release tests cited
above were observed on Linux. The supplied paper-output regeneration records
Windows~11 on ARM64 with CPython~3.13.2.

The restricted exact example is compiled and checked with
\begin{lstlisting}
pcs-bio validate-request \
  --request examples/rigorous_accept/request.json
pcs-bio compile \
  --request examples/rigorous_accept/request.json \
  --model examples/rigorous_accept/formal_model.json \
  --circuit examples/rigorous_accept/mapped_circuit.json \
  --output build/rigorous_accept
pcs-bio check \
  build/rigorous_accept/untrusted_producer_outputs/proof_bundle.json \
  --request examples/rigorous_accept/request.json \
  --model examples/rigorous_accept/formal_model.json \
  --report build/rigorous_accept/checker_report.json
\end{lstlisting}
The expected result is the restricted synthetic model-relative verdict
described in the evaluation. This command does not bind a Cello sequence. The
following excerpt shows the final verdict gate after request validation,
semantic reconstruction, and backend-eligibility checks.

\begin{lstlisting}[language=Python]
risk_ok = witnesses.max_joint_risk <= 1 - trusted_request.target_probability
delay_ok = witnesses.max_delay <= trusted_request.settling_deadline
verdict = (
    CheckerVerdict.ACCEPT
    if boolean_correct and risk_ok and delay_ok
    else CheckerVerdict.REJECT
)
_validate_claims(data["claims"], witness=witnesses, verdict=verdict)
\end{lstlisting}
The excerpt is illustrative rather than a proof. Before these lines execute,
the checker reconstructs Boolean correctness and the witnesses from the
trusted request and finite model.

The synthetic binary64 regression uses a separate command family.
\begin{lstlisting}
pcs-bio replay compile \
  --example worked_three_input \
  --name worked_three_input \
  --output build/replay_worked \
  --with-transfer
pcs-bio replay check \
  build/replay_worked/proof_bundle.json \
  --request build/replay_worked/trusted_request.json \
  --technology build/replay_worked/trusted_technology.yaml \
  --report build/replay_worked/checker_report.json
\end{lstlisting}
The successful diagnostic result is \texttt{replay-pass}, never formal
acceptance. Without additional numerical flags, the displayed command uses
five decimal places and 120 steps. Numerical settings for any other archived
replay must be taken from its manifest. They are not implied by this command.

The static and archived external Cello cases can be reconstructed without
network access or a new Cello invocation.
\begin{lstlisting}
pcs-bio cello audit \
  --kind static \
  --root examples/frozen_cello \
  --schema-root schemas/cello/upstream_v2
pcs-bio cello audit \
  --kind archived-live \
  --root examples/live_cello/nor_smoke \
  --schema-root schemas/cello/upstream_v2 \
  --pins external/cello_pins.json
\end{lstlisting}
The separate Cello results index names 20 optional presentation artifacts, of
which 18 are present. The absent files are \path{nor_gate_dpl.pdf} and
\path{nor_gate_dpl.png}. Neither export is used in mapping or sequence
reconstruction.
Both commands establish provenance and reconstruction results only. Neither
contains the temporal certificate required for ideal-checker acceptance. A new
external execution requires the pinned Cello runtime. The optional
\texttt{pcs-bio pipeline} command keeps formal and Cello evidence in one output
tree but returns \texttt{inconclusive} whenever no checked sequence-to-model
derivation connects them.

Paper-facing outputs are reconstructed with
\begin{lstlisting}
pcs-bio reproduce \
  --root . \
  --manifest experiments/manifest.yaml \
  --output paper_outputs
\end{lstlisting}
The command recomputes the BCD analysis, exact recurrence, named replay cases,
assignment distributions, archived formula corpus, the two Cello
reconstruction audits, and the restricted exact-checker result. It validates
the archived structural-growth
and mutation records and regenerates their tables and figures without
rerunning those experiments. It revalidates the immutable external Cello
archive instead of claiming a new invocation. The output manifest records a
byte count and SHA-256 digest for each listed artifact.

Trust boundaries differ by track. In the restricted exact checker, the request
and finite transition model are verifier-controlled premises, whereas the
mapped circuit and proof bundle are producer-controlled. In binary64 replay,
the request and synthetic technology file are verifier-controlled. In the
Cello track, the schema, source records, dependency lock, and external version
pins are trusted inputs. In every track, the checker reconstructs the
applicable semantic claims rather than trusting producer-supplied maxima or
digests. Verdicts from the three tracks are never combined.

\end{document}